\documentclass[11pt]{article}
\usepackage[margin=1in]{geometry}
\usepackage{amsmath,amsthm,amsfonts,amssymb}
\usepackage{graphicx,enumerate}
\usepackage[title]{appendix}
\usepackage[%dvips,
            CJKbookmarks=true,
            bookmarksnumbered=true,
            bookmarksopen=true,
            colorlinks=true,
            citecolor=blue,
            linkcolor=blue,
            anchorcolor=red,
            urlcolor=blue
            ]{hyperref}
\usepackage{bm, verbatim,dsfont,mathtools}
\usepackage{color}
\usepackage{subfigure}
\usepackage{tikz}
\usepackage{todonotes}
\usepackage{paralist}
\usepackage{algorithm}% http://ctan.org/pkg/algorithms
\usepackage{algorithmic}% http://ctan.org/pkg/algorithms
\usepackage{bbm}
\usepackage{empheq}
\UseRawInputEncoding
\newtheorem{theorem}{Theorem}
\newtheorem{lemma}{Lemma}

\newtheorem{corollary}{Corollary}
\theoremstyle{definition}
\newtheorem{definition}{Definition}

\newtheorem{remark}{Remark}
\newtheorem*{remark*}{Remark}

\usepackage{xspace,prettyref}
\newcommand{\diverge}{\to\infty}

\newcommand{\iiddistr}{{\stackrel{\text{\iid}}{\sim}}}

\newcommand{\naturals}{{\mathbb{N}}}

\newcommand{\expect}[1]{\mathbb{E}\left[ #1 \right]}

\newcommand{\prob}[1]{ \mathbb{P}\left\{ #1 \right\} }

\newcommand{\abs}[1]{ \left| #1 \right| }
\newcommand{\Indicator}[1]{ \mathbbm{1}\left( #1 \right)}

\newcommand{\var}{\mathsf{var}}

\newcommand{\Bern}{{\rm Bern}}
\newcommand{\Binom}{{\rm Binom}}

\newcommand{\ie}{i.e.\xspace}
\newcommand{\iid}{i.i.d.\xspace}
\newrefformat{eq}{(\ref{#1})}
\newrefformat{chap}{Chapter~\ref{#1}}
\newrefformat{sec}{Section~\ref{#1}}
\newrefformat{algo}{Algorithm~\ref{#1}}
\newrefformat{fig}{Fig.~\ref{#1}}
\newrefformat{tab}{Table~\ref{#1}}
\newrefformat{rmk}{Remark~\ref{#1}}
\newrefformat{clm}{Claim~\ref{#1}}
\newrefformat{def}{Definition~\ref{#1}}
\newrefformat{cor}{Corollary~\ref{#1}}
\newrefformat{lmm}{Lemma~\ref{#1}}
\newrefformat{prop}{Proposition~\ref{#1}}
\newrefformat{app}{Appendix~\ref{#1}}
\newrefformat{hyp}{Hypothesis~\ref{#1}}
\newrefformat{thm}{Theorem~\ref{#1}}

\newcommand{\indc}[1]{{\mathbf{1}_{\left\{{#1}\right\}}}}

\newcommand{\calA}{{\mathcal{A}}}

\newcommand{\ER}{Erd\H{o}s-R\'enyi\xspace}

\renewcommand{\tilde}{\widetilde}

\newcommand{\Quv}{Q_{uv}}
\newcommand{\wmin}{w_{\min}}
\newcommand{\wmax}{w_{\max}}

\newcommand{\nps}{np^2\leq\frac{1}{\log n}}
\allowdisplaybreaks[4]
\begin{document}

\title{Graph Matching with Partially-Correct Seeds}
\author{Liren Yu, Jiaming Xu, and Xiaojun Lin\thanks{
L.\ Yu  and X.\ Lin  are with School of Electrical and Computer Engineering, Purdue University, 
West Lafayette, USA, \texttt{yu827@purdue.edu, linx@ecn.purdue.edu}.
J.\ Xu is with The Fuqua School of Business, Duke University, Durham, USA, \texttt{jx77@duke.edu}.
L.~Yu and J.~Xu are supported by the NSF Grant IIS-1932630. 
}
}
\date{\today}

\maketitle
\begin{abstract}
Graph matching aims to find the latent vertex correspondence between two edge-correlated graphs and has found numerous applications across different fields. In this paper, we study a seeded graph matching problem, which assumes that a set of seeds, i.e., pre-mapped vertex-pairs, is given in advance. While most previous work requires all seeds to be correct, we focus on the setting where the seeds are partially correct. Specifically, consider two correlated graphs whose edges are sampled independently from a parent \ER graph $\mathcal{G}(n,p)$. A mapping between the vertices of the two graphs is provided as seeds, of which an unknown $\beta$ fraction is correct. We first analyze a simple algorithm that matches vertices based on the number of common seeds in the $1$-hop neighborhoods, and then further propose a new algorithm that uses seeds in the $2$-hop neighborhoods. We establish non-asymptotic performance guarantees of perfect matching for both $1$-hop and $2$-hop algorithms, showing that our new $2$-hop algorithm requires substantially fewer correct seeds than the $1$-hop algorithm when graphs are sparse. Moreover, by combining our new performance guarantees for the $1$-hop and $2$-hop algorithms, we attain the best-known results (in terms of the required fraction of correct seeds) across the entire range of graph sparsity and significantly improve the previous results in~\cite{10.14778/2794367.2794371,lubars2018correcting} when $p\ge n^{-5/6}$.  For instance, when $p$ is a constant or $p=n^{-3/4}$, we show that only $\Omega(\sqrt{n\log n})$ correct seeds suffice for perfect matching, while the previously best-known results demand $\Omega(n)$ and $\Omega(n^{3/4}\log n)$ correct seeds, respectively. Numerical experiments corroborate our theoretical findings, demonstrating the superiority of our $2$-hop algorithm on a variety of synthetic and real graphs. 

\end{abstract}

\section{Introduction}

Given a pair of two edge-correlated graphs, graph matching (also known as network alignment) aims to find a bijective mapping between the vertex sets of the two graphs so that their edge sets are maximally aligned. %The two graphs have the same set of vertices, but the identity of each vertex is hidden. 
%It reduces to graph isomorphism if there exists an edge-preserving bijection between their vertex sets.
This is a ubiquitous but difficult problem arising in many important applications, such as social network de-anonymi\-zation \cite{narayanan2009anonymizing}, computational biology \cite{singh2008global,kazemi2016proper}, computer vision \cite{conte2004thirty,schellewald2005probabilistic}, and natural language processing \cite{haghighi2005robust}. For instance, from one anonymized version of the “follow” relationships graph on the Twitter microblogging service, researchers were able to re-identify the users by matching the anonymized graph to a correlated cross-domain auxiliary graph, i.e., the “contact” relationships graph on the Flickr photo-sharing service, where user identities are known \cite{narayanan2009anonymizing}.

Existing graph matching algorithms  can be classified into two categories, seedless and seeded matching algorithms. Seedless matching algorithms only use the topological information and do not rely on any additional side information. Various seedless matching algorithms have been proposed based on either degree information \cite{dai2018performance,DMWX18}, spectral method \cite{umeyama1988eigendecomposition,cour2007balanced,feizi2019spectral,fan2019spectral,FMWX19b}, random walk \cite{gori05randomwalk},  convex relaxations~\cite{aflalo2015convex,fiori2015spectral,lyzinski2016graph,dym2017ds++,bernard2018ds}, or non-convex methods~\cite{zaslavskiy2008path,fiori2013robust,vogelstein2015fast,yu2018generalizing,maron2018probably,zhang2019kergm,xu2019gromov}. However, to the best of our knowledge, these algorithms either only provably succeed when the fraction of edges that differ between the two graphs is low, \ie, on the order of $O\left(1/\log^2 n\right)$ \cite{DMWX18}  or require at least quasi-polynomial runtime $(n^{O(\log n)})$ \cite{barak2018nearly,cullina2016improved,cullina2017exact,cullina2019partial}, where $n$ is the number of vertices in one graph. The only exception is the neighborhood tree matching algorithm recently proposed in \cite{ganassali2020tree}, which can output a partially-correct matching in polynomial-time only when two graphs are sparse and differ by a constant fraction of edges.

The other category is seeded matching algorithms \cite{pedarsani2011privacy,yartseva2013performance,korula2014efficient,Lyzinski2013Seeded,Fishkind2018Seeded,shirani2017seeded,mossel2019seeded,10.1109/TNET.2016.2553843}. These algorithms require ``seeds'', which are a set of pre-mapped vertex-pairs. Let $G_1$ and $G_2$ denote two graphs. For each pair of vertices $(u,v)$ with $u$  in $G_1$ and $v$ in $G_2$, a seed $(w,w')$ is called a \emph{1-hop witness} for $(u,v)$ if $w$ is a neighbor of $u$ in $G_1$ and $w'$ is a neighbor of $v$ in $G_2$. The basic idea of seeded matching algorithms is that a candidate pair of vertices are expected to have more witnesses if they are a true pair than if they are a fake pair. Assuming that the seeds are correct, seeded matching algorithms can find the correct matching for the remaining vertices more efficiently than seedless matching algorithm. In social network de-anonymization, such initially matched seeds are often available, thanks to users who have explicitly linked their accounts across different social networks. For other applications, the seeds can be obtained by prior knowledge or manual labeling. 

However, most existing seeded matching algorithms crucially rely on all seeds being correct, which is often difficult to guarantee in practice. For example, the seeds may be provided by seedless matching algorithms, which will likely produce some incorrect seeds. To overcome this limitation, \cite{10.14778/2794367.2794371} and \cite{lubars2018correcting} extend the idea of seeded matching algorithms to allow for incorrect seeds. In particular, \cite{10.14778/2794367.2794371} proposes a NoisySeeds algorithm, which uses percolation \cite{janson2012bootstrap,yartseva2013performance} to grow the number of 1-hop witnesses from partially-correct seeds and iteratively matches pairs whose number of witnesses exceeds a threshold $r$. However,  NoisySeeds is very sensitive to the choice of the threshold $r$ and matching errors, and thus is only guaranteed to perform well when the graphs are very sparse. More specifically, when the two graphs are correlated \ER graphs, whose edges are independently sub-sampled with probability $s$ from a \emph{parent} \ER graph $\mathcal{G}(n,p)$, and when  $\beta$ fraction of seeds are correct, it is shown in~\cite{10.14778/2794367.2794371} that  NoisySeeds with the best choice of threshold $r=2$
%which gives the best performance guarantee among all choices of $r$)
can correctly match all but $o(n)$ vertex-pairs with high probability, provided that
$n^{-1}\ll p \leq n^{-\frac{5}{6}-\epsilon}$ for  $\epsilon \in (0,1/6)$, and
\begin{equation}
    \begin{aligned}
 \beta\geq\frac{1}{2n^2p^2s^4}. \label{eq:conditionPGM}
    \end{aligned}
\end{equation}
However, for denser graphs with $p\geq n^{-\frac{5}{6}}$, no performance guarantees are established in \cite{10.14778/2794367.2794371} for the setting with incorrect seeds.

%\nbr{JX. Here I think we need to briefly discuss the NoisySeeds algorithm and mention the threshold $r$, because later we need to discuss the choice of $r$. By the way, the name of NoisySeeds is not informative. Shall we call it percolation-based matching?}
% Consider two correlated \ER graphs, whose edges are independently sub-sampled with probability $s$ from a parent \ER graph $\mathcal{G}(n,p)$, and there are $\beta$ fraction of correct seeds. When the graphs are very sparse, in particular when  $n^{-1}\ll p \leq n^{-\frac{5}{6}-\epsilon}$ for $\epsilon \in (0,1/6)$,
% the NoisySeeds algorithm in \cite{10.14778/2794367.2794371} can match almost all the nodes correctly with high probability, provided that
% \begin{equation}
%     \begin{aligned}
%     \beta\geq\frac{1}{2n^2p^2s^4}. \label{eq:conditionPGM}
%     \end{aligned}
% \end{equation}

%For denser graphs with $p\geq n^{-\frac{5}{6}}$, no performance guarantees are provided by the algorithm.   
In contrast, \cite{lubars2018correcting} proposes a different algorithm that uses the numbers of 1-hop witnesses for each candidate pair of vertices as weights, and then uses Greedy Maximum Weight Matching (GMWM) to find the vertex correspondence between the two graphs such that the total number of witnesses is large. \cite{lubars2018correcting} shows that their $1$-hop algorithm can work over a much wider range of $p$ (up to $p\leq \frac{3}{8}$) than \cite{10.14778/2794367.2794371}, and it can correctly match all vertices with high probability if
\begin{equation}
    \begin{aligned}
    \beta\geq\max\left\{\frac{16\log{n}}{nps^2}, \, \frac{8}{3}p \right\}. \label{eq:conditionSrikant}
    \end{aligned}
\end{equation}

In order to illustrate the limitations of these existing results, we plot in \prettyref{fig:condition} the scalings corresponding to the two conditions \prettyref{eq:conditionPGM} and \prettyref{eq:conditionSrikant}, as the black dotted curve and green dashed curve, respectively, where the $x$-axis is the  graph sparsity $p$ (which is bounded away from 1) and the sampling probability $s$ is a constant. We observe that, when the graphs are sparse, condition \prettyref{eq:conditionSrikant} ($\beta= \Omega\left(\log n/np\right)$) requires substantially more correct seeds than condition \prettyref{eq:conditionPGM} ($\beta= \Omega\left(1/n^2p^2\right)$), suggesting that the 1-hop algorithm in \cite{lubars2018correcting} is suboptimal. However, condition \prettyref{eq:conditionPGM} only extends to $p \leq n^{-5/6}$, and is not applicable to denser graphs. 
When the graphs are dense, condition \prettyref{eq:conditionSrikant} requires $\beta$ to increase proportionally in $p$. In particular, when $p$ is a constant, condition \prettyref{eq:conditionSrikant} demands a constant fraction of correct seeds. Such a requirement seems rather stringent as well. 
In summary, %despite the substantial progresses made in the previous work \cite{10.14778/2794367.2794371,lubars2018correcting}, 
the existing conditions on the required number of correct seeds are either pessimistic or only applicable to very sparse graphs. Since %often only very few 
the number of correct seeds is often limited in practice, it is of paramount importance in both theory and practice to understand how to better utilize partially-correct seeds to attain more accurate matching results for both sparse and dense graphs. 

In this paper, we establish new performance guarantees for perfect matching, significantly relaxing the existing requirements in \prettyref{eq:conditionPGM} and \prettyref{eq:conditionSrikant}. Specifically, we first provide a much tighter analysis than~\cite{lubars2018correcting}, showing that the $1$-hop algorithm can correctly match all vertices with high probability, provided that 
 \begin{equation}
     \begin{aligned}
   \beta\geq\max\left\{ \frac{45\log{n}}{np(1-p)^2s^2}, \, 30\sqrt{\frac{\log{n}}{n(1-p)^2s^2}} \right\}. \label{eq:conditionofbeta1}
    \end{aligned}
 \end{equation}
Moreover, we propose a new algorithm based on the number of $2$-hop witnesses and shows that it can exactly match all vertices  with high probability, provided that $np^2 \le (\log n)^{-1}$,  $\ nps^2\geq 128\log n$, and
 \begin{equation}
     \begin{aligned}
   \beta\geq\max\left\{\frac{600\log n}{n^2p^2s^4}, \, 600\sqrt{\frac{\log n}{ns^4}}, \, 600\sqrt{\frac{np^3(1-s)\log n}{s}}\right\}.
   %\ \mbox{and}\ nps^2\geq 128\log n. 
   \label{eq:conditionofbeta2}
     \end{aligned}
 \end{equation}

%To illustrate the improvement of our 
The new conditions~\prettyref{eq:conditionofbeta1} and~\prettyref{eq:conditionofbeta2}, are also plotted in \prettyref{fig:condition} as the solid red and blue curves, respectively, to illustrate the improvement compared to the previous conditions  \prettyref{eq:conditionPGM} and \prettyref{eq:conditionSrikant}. (Note that there is a factor $1-s$ in our condition \prettyref{eq:conditionofbeta2}. As a consequence, the corresponding blue curve has two branches: the top one holds for $s<1$ and the bottom one holds for $s=1$.)

\begin{figure}[ht]
\centering
\resizebox{0.7\textwidth}{!}{
\begin{tikzpicture} [scale = 0.8,
    node distance = 1.3cm, on grid, > = stealth, bend angle = 45, auto,
    vertex/.style = {circle, draw = red!50, fill = red!20, thick, minimum size = 3mm},
    seed/.style = {circle, draw = blue!50, fill = blue!20, thick, minimum size = 3mm},
    graph/.style = {rectangle, draw = black!75,thick}
    ]

\draw[->,line width=1pt] (0,0) -- (14,0);
\node (p) at (13.5,0)[label={[black]-90:$p$}]{} ;
\draw[->,line width=1pt] (0,0) --  (0,7);
\node (beta)[label={[black,rotate=90]left: Condition on $\beta$}] at (-0.2,6){} ;
\draw [color=red,line width=2pt] (7,1) -- (13,1);
\draw [color=red,line width=2pt] (1/2,5)  .. controls (2.5,2) and (4,1.5) ..  (7,1);
\draw [dashed, color=black!30!green,line width=2pt] (7,1) -- (13,3.5);
\draw [dashed, color=black!30!green,line width=2pt]  (1/2,5)  .. controls (2.5,2) and (4,1.5) ..  (7,1);
\draw [color=white!20!blue,line width=2pt] (2.5,1) -- (7,1);
\draw [color=white!20!blue,line width=2pt] (1/2,4)  .. controls (1,2.5) and (2,1.5) ..  (2.5,1);
\draw [color=white!20!blue,line width=2pt] (4,1)  .. controls (5.5,1.1) and (6,1.4) ..  (7,2);
\draw [dotted, color=black,line width=2pt] (1/2,4)  .. controls (1,2.7) and (1.25, 2.5) ..  (1.5,2.1);

\draw[dashed] (0,1)--(2.5,1);
\draw[dashed] (2.5,0)--(2.5,1);
\draw[dashed] (4,0)--(4,1);
\draw[dashed] (5.5,0)--(5.5,1.1);
\draw[dashed] (7,0)--(7,2);

\draw[dashed] (1.5,0)--(1.5,2);

\node (n6) at (2.8,2.5)[label={[red]above:$\frac{\log n}{np}$}]{} ;
\node (n7) at (0.8,2.7)[label={[black]-90:$\frac{\log n}{n^2p^2}$}]{} ;
\node (n8) at (6,1.7)[label={[font=\small,blue]above:$\sqrt{{np^3\log n}}$}]{} ;
\node (n9) at (10,2.2)[label={[black!30!green]above:$p$}]{} ;
\node (n10) at (1.5,0)[label={[black]-90:$n^{-\frac{5}{6}}$}]{} ;

\node (n1) at (0,1)[label={[black]180:$\sqrt{\frac{\log n}{n}}$}]{} ;
\node (n2) at (2.5,0)[label={[black]-90:$n^{-\frac{3}{4}}$}]{} ;
\node (n3) at (4,0)[label={[black]-90:$n^{-\frac{2}{3}}$}]{} ;
\node (n4) at (5.6,0)[label={[black]-90:$n^{-\frac{3}{5}}$}]{} ;
\node (n5) at (7,0)[label={[black]-90:$n^{-\frac{1}{2}}$}]{} ;
\node (legend) at (6.9,5.55) [graph,text height = 1.7cm, minimum width = 7.8cm]{};
\draw [color=red,line width=2pt] (2.3,5.2) --  (3.5,5.2);
\draw [color=white!20!blue,line width=2pt](2.3,4.6)--(3.5,4.6);
\draw [dashed,color=black!30!green,line width=2pt](2.3,5.8)--(3.5,5.8);
\draw [dotted,color=black,line width=2pt](2.3,6.4)--(3.5,6.4);
\node (t2) at (3.6,4.6)[label={[font=\small]0:Our condition \prettyref{eq:conditionofbeta2} for new 2-hop}]{} ;
\node (t1) at (3.6,5.2)[label={[font=\small]0:Our improved condition \prettyref{eq:conditionofbeta1} for 1-hop}]{} ;
\node (t0) at (3.6,5.8)[label={[font=\small]0:Condition \prettyref{eq:conditionSrikant} for 1-hop \cite{lubars2018correcting}}]{};
\node (t4) at (3.6,6.4)[label={[font=\small]0:Condition \prettyref{eq:conditionPGM} for NoisySeeds \cite{10.14778/2794367.2794371}}]{} ;
\end{tikzpicture}}
    \caption{
    Comparison of the requirements on the fraction  $\beta$ of correct seeds, 
    when $s$ is a fixed constant and  $p$ is bounded away from $1.$
    The lower the curve, the fewer correct seeds it requires. 
    %Comparison of the seed demand functions of  condition \prettyref{eq:conditionPGM} for the 1-hop algorithm \cite{10.14778/2794367.2794371}, condition \prettyref{eq:conditionSrikant} for the NoisySeeds algorithm \cite{lubars2018correcting}, our improved condition \prettyref{eq:conditionofbeta1} for the 1-hop algorithm, and our condition \prettyref{eq:conditionofbeta2} for the new 2-hop algorithm, when $s$ is a fixed constant, and $p=o(1)$.
    }
    \label{fig:condition}
\end{figure}
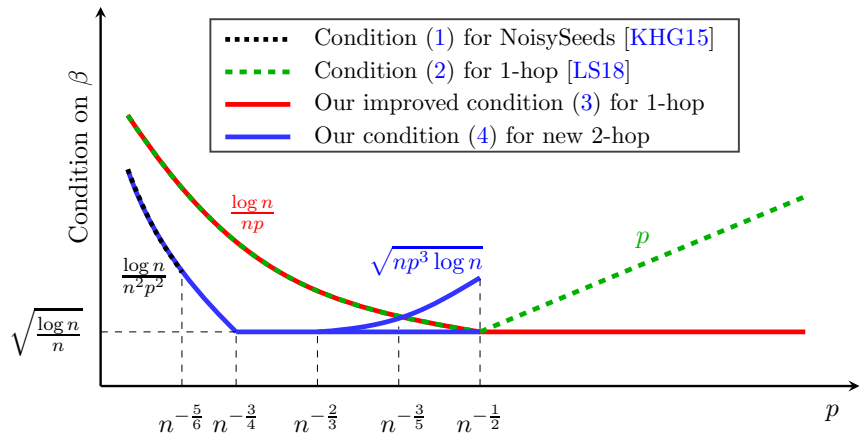

From \prettyref{fig:condition}, we can see that by combining our two conditions \prettyref{eq:conditionofbeta1} and \prettyref{eq:conditionofbeta2} (i.e., the lower envelope of blue and red solid curves), we attain the lowest requirements on the number of correct seeds across the entire range of graph sparsity $p$, and significantly improve the previous conditions when $p \ge n^{-5/6}$. In particular, 
\begin{itemize}
    \item Comparing the green dashed curve with the red solid curve, we see that when $p \gg n^{-1/2}$ our condition \prettyref{eq:conditionofbeta1} requires many fewer correct seeds than \prettyref{eq:conditionSrikant} for the $1$-hop algorithm to succeed. For instance, when $p$ is a constant, only $\Omega(\sqrt{n\log n})$ correct seeds suffice for the $1$-hop algorithm to achieve perfect matching according to our condition \prettyref{eq:conditionofbeta1}, while the previous condition \prettyref{eq:conditionSrikant} requires $\Omega(n)$ correct seeds. 
    
    \item Comparing the green dashed curve with the blue solid curve, we see that when  $p \ll n^{-1/2}$ for $s=1$ or when $p \ll n^{-3/5}$ for $s<1$, our condition~\prettyref{eq:conditionofbeta2} also requires substantially fewer correct seeds than \prettyref{eq:conditionSrikant}. This shows that our new $2$-hop algorithm is significantly better than the $1$-hop algorithm when the graphs are sparse. For instance, when $p=n^{-3/4}$, only $\Omega(\sqrt{n\log n})$ correct seeds suffice for our $2$-hop algorithm to achieve perfect matching, while the $1$-hop algorithm requires $\Omega\left(n^{3/4} \log n \right)$ correct seeds. %Further, this comparison reveals an interesting and delicate trade-off between the \emph{quantity} and the \emph{quality} of witnesses: while the $2$-hop algorithm exploits more seeds as witnesses than $1$-hop, the $2$-hop witnesses are less accurate than the $1$-hop counterparts in distinguishing true pairs from fake pairs when graphs are dense. 
    
    \item Comparing the black dotted curve with the blue solid curve, we see that our condition~\prettyref{eq:conditionofbeta2} is comparable to condition \prettyref{eq:conditionPGM} when $p \ll n^{-5/6}$. However, our condition~\prettyref{eq:conditionofbeta2} continues to hold up to $p \ll n^{-1/2}$. This shows that our $2$-hop algorithm enjoys competitive performance compared to NoisySeeds when graphs are very sparse, but is more versatile and continues to perform well over a much wider range of graph sparsity. 
\end{itemize}

Furthermore, our results precisely characterize the graph sparsity at which the $2$-hop algorithm starts to outperform $1$-hop. This reveals an interesting and delicate trade-off between the \emph{quantity} and the \emph{quality} of witnesses: while the $2$-hop algorithm exploits more seeds as witnesses than the $1$-hop algorithm, the $2$-hop witnesses can also be less discriminating (as they are further away from the node-pair under consideration). Thus, while the increased quantity helps when the graphs are sparse, the decreased quality can confuse the matching algorithm when the graphs are dense (e.g., even the fake pairs are likely to have many 2-hop witnesses).

Finally, using numerical experiments on both synthetic and real graphs, we show that our 2-hop algorithm significantly outperforms the state-of-the-art. Specifically, when the initial seeds are randomly chosen, the 2-hop algorithm significantly outperforms the 1-hop algorithm in \cite{lubars2018correcting} on sparse graphs, which agrees with our theoretical analysis. Further, the performance of our 2-hop algorithm is comparable to  the NoisySeeds algorithm when the synthetic graphs are very sparse, and much better than the NoisySeeds algorithm on other graphs. Our 2-hop algorithm is also much more robust to power-law degree variations in real graphs than the NoisySeeds algorithm. Moreover, we conduct an experiment on matching 3D deformable shapes in which the initial seeded mapping is generated by a seedless algorithm (instead of randomly chosen). We demonstrate that our 2-hop algorithm drastically boosts the matching accuracy by cleaning up most initial matching errors, and 
the performance enhancement 
is more substantial than the 1-hop algorithm and NoisySeeds algorithm.
Computationally, our $2$-hop algorithm is comparable to the $1$-hop and NoisySeeds algorithm and runs efficiently on networks with $\sim 10K$ nodes on a single PC, and can potentially scale up to even larger networks using parallel implementation. 

In passing, we remark that although we focus on matching two graphs of the same number of vertices,
our $2$-hop algorithm can be  directly applied to 
matching two graphs of different sizes and return an accurate correspondence between nodes in the common subgraph of the two graphs. Indeed, the simulation results with real data in \prettyref{sec:experiment-real} show that our 2-hop algorithm still achieves outstanding matching performance, even when two graphs are of very different sizes.

\subsection{Key Ideas and Analysis Techniques}\label{sec:idea}

%To answer \textbf{Question 1} and  \textbf{Question 2}, we describe our ideas and analysis challenging. 

Our improved performance guarantees for perfect matching exploit several key ideas and analysis techniques, which we present below and will elaborate further in later sections. For ease of discussion, we assume the true vertex correspondence between the two graphs is given by the  identity permutation. We use $\pi:[n]\to [n]$ to denote the initial seed mapping. Then, each seed $(i,\pi(i))$ is correct  if $\pi(i)=i$ and incorrect otherwise. 

%For \textbf{Question 1}, 

\begin{figure}[ht]
\centering
\resizebox{0.3\textwidth}{!}{
\begin{tikzpicture} [
    node distance = 1.3cm, on grid, > = stealth, bend angle = 80, auto,
    vertex/.style = {circle, draw = red!50, fill = red!20, thick, minimum size = 3mm},
    seed/.style = {circle, draw = blue!50, fill = blue!20, thick, minimum size = 3mm},
    graph/.style = {rectangle, draw = black!75,thick}
    ]
\node (u) at (-2.2,0) [vertex,label={[red!60]above:$u$}]{};
\node (ui) at (-0.8,0) [seed,label={[blue!60]above:$i$}]{};
\node (pivi) at (0.8,0) [seed,label={[blue!60]above:$\pi(i)$}]{};
\node (piu) at (2.2,0) [vertex,label={[red!60]above:$u$}]{};
\node (vi) at (-0.8,1.5) [seed,label={[blue!60]above:$\pi^{-1}(i)$}]{};
\node (pivj) at (0.8,1.5) [seed,label={[blue!60]above:$i$}]{};
\node (uj) at (-0.8,-1.5) [seed,label={[blue!60]above:$\pi(i)$}]{};
\node (piui) at (0.8,-1.5) [seed,label={[blue!60]above:$\pi(\pi(i))$}]{};
\draw [dashed] (-0.6, 0) -- (0.6, 0);
\draw [dashed] (-0.6, 1.5) -- (0.6, 1.5);
\draw [dashed] (-0.6, -1.5) -- (0.6, -1.5);
\draw[black!60!green,line width=1.2pt]  (-2, 0) -- (-1, 0);
\draw  (-2, 0) -- (-1, 1.5);
\draw[purple,line width=1.2pt]  (-2, 0) -- (-1, -1.5);
\draw[purple,line width=1.2pt]   (2, 0) -- (1, 0);
\draw[black!60!green,line width=1.2pt]  (2, 0) -- (1,  1.5);
\draw (2, 0) -- (1, -1.5);
\node (G1) at (-1.35,0.25) [graph,text height = 4cm, minimum width = 2.55cm,label={[black!60]above:$G_1$}]{};
\node (G2) at (1.35,0.25) [graph,text height = 4cm, minimum width = 2.55cm,label={[black!60]above:$G_2$}]{};
\end{tikzpicture}}
\caption{The two green/purple edges are correlated, because they correspond to the same edge in the parent graph.  Thus, the event that an incorrect seed $(i,\pi(i))$ becomes a 1-hop witness for $(u,u)$ is dependent on  the events that $(\pi^{-1}(i),i)$ and $(\pi(i),\pi(\pi(i)))$ become 1-hop witnesses for $(u,u)$.}
\label{fig:dependidea1}
\end{figure}
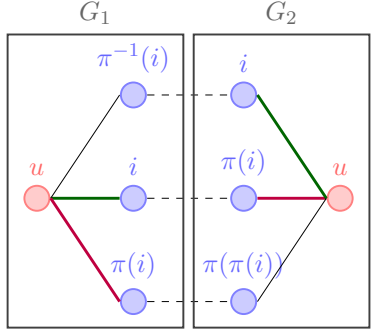

To obtain a much tighter condition than \prettyref{eq:conditionSrikant} for the success of the $1$-hop algorithm, our key observation is that,  when counting the number of witnesses for the true pairs, the analysis of  \cite{lubars2018correcting} only considers the correct seeds and ignores the incorrect seeds. It is non-trivial to consider the incorrect seeds for a true pair, because the events that the incorrect seeds become witnesses depend on each other. In particular, the event that an incorrect seed $(i,\pi(i))$ becomes a 1-hop witness for true pair $(u,u)$ is dependent on  the events that $(\pi^{-1}(i),i)$ and $(\pi(i),\pi(\pi(i)))$ become 1-hop witnesses for $(u,u)$, as these events may involve the same edges in the parent graph (See \prettyref{fig:dependidea1} for an illustrating example).  Our analysis takes into account the incorrect seeds for the true pairs and carefully deals with the above dependency issue using concentration inequalities for dependent random variables \cite{Janson2004LargeDF}. See \prettyref{sec:proofthm1hop} for details.

Further, to better utilize seeds in sparse graphs, our key idea is to match vertices by comparing the number of witnesses in the 2-hop neighborhoods. In sparse graphs, the number of 1-hop witnesses for most vertices will be very low. (For example, when $p=O( n^{-1/2})$ and $\beta =O( n^{-1/2})$, the number of 1-hop witnesses for even a true pair is about $\beta n p=O(1)$.)
Therefore, it will be difficult to use 1-hop witnesses alone to distinguish true pairs from fake pairs. In contrast, the number of 2-hop witnesses will be much larger. %most vertices have more 2-hop neighbors than 1-hop neighbors. 
Thus, compared to the algorithm in \cite{lubars2018correcting}
%and~\cite{10.14778/2794367.2794371} 
that uses only 1-hop witnesses, our 2-hop algorithm can leverage more witnesses to distinguish the true pairs from the fake pairs. The idea of using multi-hop neighborhoods to match vertices is analyzed previously in \cite{mossel2019seeded} when all seeds are correct. In comparison, our results on the 2-hop algorithm make several significant contributions. First,  our analysis with incorrect seeds is considerably more challenging, as we need to take care of the dependency on the size of the 1-hop neighborhood and the dependency between incorrect seeds. 
%For instance, for an incorrect seed $(j,\pi(j))$, whether $\pi(j)$ is connected to $u$ in $G_1$  can change the probability that $(j,\pi(j))$ becomes a 2-hop witnesses for $(u,v)$. 
Unfortunately, unlike the setting in the previous paragraph, here directly using the results of \cite{Janson2004LargeDF} will not work because the number of dependencies that we have to deal with may be very large (see \prettyref{sec:proofthm2hop} for details). Instead we deal with the dependency issues by first conditioning on the 1-hop neighborhood; and then analyzing different seeds according to different situations and applying the concentration inequalities for dependent random variables \cite{Janson2004LargeDF}. Second, our condition~\prettyref{eq:conditionofbeta2} characterizes the influence of the incorrect seeds and reveals the delicate behavior of the $2$-hop algorithm. In particular, we show that the $2$-hop algorithm requires at least $\Omega(\sqrt{n\log n})$ correct seeds, irrespective of the graph sparsity. Also, somewhat surprisingly, we discover that when $s<1$, the $2$-hop algorithm may require more seeds as $p$ increases from $n^{-2/3}$, due to the larger fluctuation of $1$-hop neighborhood sizes. All these new phenomenons are absent when seeds are all correct and thus are not captured by the theoretical results in \cite{mossel2019seeded}.
%showing that the $2$-hop algorithm requires at least $\sqrt{n\log n}$ correct seeds, 
Third, the computational complexity of the algorithm in \cite{mossel2019seeded} is $O(n^{3})$, which is much higher than that of our algorithm -- $O(n^{\omega}+n^2\log n)$, where $n^{\omega}$ with $2 \le \omega \le 2.373$ denotes the time complexity for $n\times n$ matrix multiplication (see \prettyref{sec:algorithm} for detail).

\section{Model}\label{sec:model}

In this section, we formally introduce the model and  the graph matching problem with partially-correct seeds.

We use the $\mathcal{G}(n,p;s)$ graph model proposed by \cite{pedarsani2011privacy}, which has been widely used in the study of graph matching. Let $G_0$ denote the parent graph with $n$ vertices $\{1,2,...,n\}\triangleq [n]$. The parent graph $G_0$ is generated from the \ER model $\mathcal{G}(n, p)$, \ie, we start with an empty graph on $n$ vertices and connect any pair of two vertices independently with probability $p$. Then, we obtain a subgraph $G_1$ by sampling each edge of $G_0$ into $G_1$ independently with probability $s$. Repeat the same sub-sampling process independently and relabel the vertices according to an \emph{unknown} permutation $\pi^*:[n]\to [n]$ to construct another subgraph $G_2$. Throughout the paper, we denote a vertex-pair by $(u,v)$, where $u\in G_1$ and $v\in G_2$. For each vertex-pair $(u,v)$, if $v=\pi^*(u)$, then $(u,v)$ is a true pair; if $v\neq \pi^*(u)$, then $(u,v)$ is a fake pair. 
%The goal of graph matching is to recover this unknown permutation $\pi^*$ based on $G_1$ and $G_2$.

As a motivating example, the parent graph $G_0$ can be some underlying friendship network among $n$ persons, while $G_1$ is the Flicker contact network  and $G_2$ is a Twitter follow network among these $n$ persons. 

Prior literature proposes various algorithms to recover $\pi^*$ based on $G_1$  and $G_2$. The output of these graph matching algorithms can be interpreted as a set of partially correct seeds. Taking these partially correct seeds as input, we wish to efficiently correct all of the errors. However, it is difficult to perfectly model the correlation between the output of these algorithms and the graphs. One way to get around this issue is to treat these partially correct seeds as adversarially chosen and to design an algorithm that with high probability corrects all errors for all possible initial error patterns. However, the existing theoretical guarantees in this adversarial setting are pessimistic, requiring the fraction of incorrectly matched seeds to be $o(1)$ (cf.~\cite[Lemma 3.21]{barak2018nearly} and \cite[Lemma 5]{DMWX18}). 

In this paper, we adopt a mathematically more tractable model introduced by \cite{lubars2018correcting}, where the partially correct seeds are assumed to be generated independently from the graphs $G_1$ and $G_2$. More specifically, we use $\pi:[n]\to [n]$ to denote an initial mapping and generate $\pi$ in the following way. For $\beta\in [0,1)$, we assume that $\pi$  is  uniformly and randomly chosen from all the permutations $\sigma: [n]\rightarrow [n]$ such that $\sigma(u)=\pi^*(u)$ for exactly $\beta n$ vertices.
%Thus, only $\beta$ fraction of the seeds are correct.
The benefit of this model is that $\pi$ is independent of the graph $G$ and the sampling processes that generate $G_1$ and $G_2$, and it is convenient for us to obtain theoretical results. For each seed $(u,\pi(u))$, if $\pi(u)=\pi^*(u)$, then $(u,\pi(u))$ is a correct seed; if $\pi(u)\neq \pi^*(u)$, then $(u,\pi(u))$ is an incorrect seed. Thus, only $\beta$ fraction of the seeds are correct. 
% let $v$ be the underlying vertex matched to $u$, i.e., $\pi(u)=v$.
% If $u=v$, then $(u,\pi(u))$ is a correct seed; if $u\neq v$, then $(u,\pi(u))$ is an incorrect seed. 
Given $G_1$, $G_2$ and $\pi$, our goal is to find a mapping $\tilde{\pi}:[n]\to [n]$ such that $\lim\limits_{n\to\infty}\prob{\tilde{\pi}=\pi^*}=1$. 

\paragraph{Notation}
For any $n \in \naturals$, let $[n]=\{1,2,\cdots,n\}$. 
We use standard asymptotic notation: for two positive sequences $\{a_n\}$ and $\{b_n\}$, we write $a_n = O(b_n)$ or $a_n \lesssim b_n$, if $a_n \le C b_n$ for some an absolute constant $C$ and for all $n$; $a_n = \Omega(b_n)$ or $a_n \gtrsim b_n$, if $b_n = O(a_n)$; $a_n = \Theta(b_n)$ or $a_n \asymp b_n$, if $a_n = O(b_n)$ and $a_n = \Omega(b_n)$; 
$a_n = o(b_n)$ or $b_n = \omega(a_n)$, if $a_n / b_n \to 0$ as $n\diverge$.

\section{Algorithm Description}\label{sec:algorithm}
  
In this section, we present a general class of algorithms, shown in Algorithm \ref{alg:De-anonymization}, that we will use to recover $\pi^*$. Similar to \cite{lubars2018correcting}, our algorithm also uses the notion of ``witnesses". However, unlike \cite{lubars2018correcting}, our algorithm leverages witnesses that are $j$-hop away. Given any graph $G$ and two vertices $u,v$ in $G$, we denote the length of the shortest path from $u$ to $v$ in $G$ by $d^G(u,v)$. Then, for each vertex-pair $(u,v)$, the seed $(w,\pi(w))$ becomes a $j$-hop witness for $(u,v)$ if  $d^{G_1}(u,w)=j$ and $d^{G_2}(v,\pi(w))=j$. %The total number of witnesses of every vertex-pair can be efficiently calculated as follows (see Algorithm \ref{alg:De-anonymization}). 

We define the $j$-hop adjacency matrices $A_j\in \{0,1\}^{n\times n}$ of $G_1$. Each element of $A_j$  indicates whether a pair of vertices are $j$-hop neighbors in graph $G_1$, \ie,
$A_j(u,v)=1$ if $d^{G_1}(u,v)=j$ and $A_j(u,v)=0$ otherwise.
% \begin{align*}
% A_j(u,v)=\left\{
%     \begin{aligned}
%         & 1 & &\text{if } d^{G_1}(u,v)=j,  \\
%         & 0 & &\text{otherwise}.
%     \end{aligned}
% \right.
% \end{align*}
Similarly, let $B_j\in \{0,1\}^{n\times n}$ denote the $j$-hop adjacency matrix of $G_2$. Equivalently express the seed mapping $\pi$ by forming a permutation matrix $\Pi\in \{0,1\}^{n\times n}$, where $\Pi(u, v) = 1 $ if $\pi(u) = v$, and $\Pi(u,v) = 0$ otherwise. We can then count the number of $j$-hop witnesses for all vertex-pairs by computing $W_j=A_j\Pi B_j$, where the $(u,v)$-th entry of $W_j$ is equal to the number of $j$-hop witnesses for the vertex-pair $(u,v)$. This step has computational complexity same as matrix multiplication $O(n^{\omega})$ with $2 \le \omega \le 2.373$ \cite{gall2014powers}. As we have mentioned, a true pair tends to have more witnesses than a fake pair, and thus we want to find the vertex correspondence between  the  two  graphs  that  maximizes  the  total  number  of witnesses. In other words, given a weighted bipartite graph $G_m$ with the vertex set being a collection of all vertices in $G_1$ and $G_2$, the edges connecting every possible vertex-pairs, and weight of an edge defined as $w(u,v)=W_j(u,v)$, we want to find the matches in $G_m$ with large weights.  To reduce computational complexity, instead of computing the maximum weight matching, we use Greedy Maximum Weight Matching (GMWM) with computational complexity $O(n^2\log{n})$. GMWM first chooses the vertex-pair with the largest weight from all candidate vertex-pairs in $G_m$, removes all edges adjacent to the chosen vertex-pair, and then chooses the vertex-pair with the largest weight among the remaining candidate vertex-pairs, and so on. The total computational complexity of Algorithm \ref{alg:De-anonymization} for any constant $j$ is $O(n^{\omega}+n^2\log n)$ for $2 \le \omega \le 2.373$.

When graphs are sufficiently sparse with average degree $c$, we can improve the time complexity of Algorithm \ref{alg:De-anonymization} to $O( n c^{2j}+n^2\log n)$ by computing the number of $j$-hop witnesses via neighborhood exploration. Moreover, we can further improve the scalability of the $j$-hop algorithm via parallel implementation. 
See \prettyref{app:scalable} for details. 

\begin{algorithm}[h]
 \caption{Graph Matching based on Counting $j$-hop Witnesses.   }
  \begin{algorithmic}[1]
  \STATE \textbf{Input:} $G_1, G_2, \pi,j$
  \STATE Generate $j$-hop adjacency matrices $A_j$ and $B_j$ based on $G_1$ and $G_2$, and $\Pi$ based on $\pi$;
  %\STATE $A=Countneighbor(G_1,2)$
  %\STATE $B=Countneighbor(G_2,2)$
  \STATE \textbf{Output}\ $\tilde{\pi}=\mathrm{GMWM}(W_j)$, where $W_j=A_j\Pi B_j$.
 % \STATE \\
 \end{algorithmic}\label{alg:De-anonymization}
\end{algorithm}

\section{Main Results}\label{sec:results}

In this section, we present the performance guarantees for the 1-hop  and 2-hop algorithms. 
%It is shown in  \cite{lubars2018correcting} that if condition \prettyref{eq:conditionSrikant} holds, the 1-hop algorithm can exactly recover $\pi^*$ with high probability. However, the analysis in \cite{lubars2018correcting} only takes into account the contribution of correct seeds when counting the 1-hop witnesses for true pairs. 
%In this paper, by carefully bounding the contribution of incorrect seeds, we show that
%the $1$-hop algorithm succeeds in exact recovery with many fewer correct seeds. 

\begin{theorem}\label{thm:thm1hop}
If condition \prettyref{eq:conditionofbeta1} holds and $n$ is sufficiently large, then Algorithm \ref{alg:De-anonymization} with $j=1$ outputs $\tilde{\pi}$ such that $\prob{\tilde{\pi}=\pi^*}\geq 1-n^{-1}$. 
\end{theorem}

Comparing our condition \prettyref{eq:conditionofbeta1}  with the previous condition \prettyref{eq:conditionSrikant} in \cite{lubars2018correcting} as depicted in \prettyref{fig:condition},
%, we plot them as a function of $p$ when $s$ is a fixed constant and $p$ is bounded away from $1$ in \prettyref{fig:condition}.
%The blue dashed curve depicts condition \prettyref{eq:conditionSrikant} in \cite{lubars2018correcting}. The two segments correspond to: when $p=O(\sqrt{{\log n}/{n}})$, $\frac{\log n}{np}$ dominates the right-hand side of condition \prettyref{eq:conditionSrikant}; when $p=\Omega(\sqrt{{\log n}/{n}})$, $p$ dominates the right-hand side of condition \prettyref{eq:conditionSrikant}. The red curve depicts condition \prettyref{eq:conditionofbeta1} in \prettyref{thm:thm1hop}. The two segments correspond to: when $p=O\left(\sqrt{{\log n}/{n}}\right)$, $\frac{\log n}{np}$  dominates the right-hand side of condition \prettyref{eq:conditionofbeta1}; When $p=\Omega\left(\sqrt{{\log n}/{n}}\right)$, $\sqrt{\frac{\log n}{n}}$  dominates the right-hand side of condition \prettyref{eq:conditionofbeta1}. 
we see that condition \prettyref{eq:conditionofbeta1} requires significantly fewer correct seeds than condition \prettyref{eq:conditionSrikant} for dense graphs when $p=\Omega(\sqrt{\log n / n})$; thus, the 1-hop algorithm succeeds in exact recovery even when the fraction of correct seeds is significantly lower than the theoretical prediction in \cite{lubars2018correcting}. 

However, when $p=O(\sqrt{\log n / n })$,  condition \prettyref{eq:conditionofbeta1} still requires $\beta$ to grow inversely proportional to $np$. As we have discussed, this is because when the graph is sparse, there are not enough 1-hop witnesses among the true pairs. 
%In order to improve performance for sparse graphs, %we  leverage witnesses in a larger neighborhood to match vertices. 
%Algorithm \ref{alg:De-anonymization} with $j=2$ utilizes the 2-hop witnesses. 
Next, we show that by utilizing 2-hop witnesses, our $2$-hop algorithm succeeds in exact recovery with many fewer correct seeds in sparse graphs. 
%under a more relaxed condition for sparse graphs. 

\begin{theorem}\label{thm:thm2hop}
Suppose that $\nps$ and $nps^2 \ge 128\log n$. If condition \prettyref{eq:conditionofbeta2} holds and $n$ is sufficiently large, then Algorithm \ref{alg:De-anonymization} with $j=2$ outputs $\tilde{\pi}$ such that $\prob{\tilde{\pi}=\pi^*}\geq 1-n^{-1}$. 
\end{theorem}

%To compare the seed requirement of the NoisySeeds algorithm, the 1-hop and 2-hop algorithms,  we plot the conditions on $\beta$ as a function of $p$ when $s$ is a fixed constant and $p$ is bounded away from $1$ in \prettyref{fig:condition}. 
%The black dotted curve depicts condition \prettyref{eq:conditionPGM} in \cite{10.14778/2794367.2794371}, which applies only to very sparse graphs with $p\ll n^{-\frac{5}{6}}$. %Further, note that the result of \cite{10.14778/2794367.2794371} only guarantees that a $(1-o(1))$ fraction of vertex-pairs are correctly matched, while all other results guarantee all vertex-pairs are correctly matched with high probability. The red curve depicting condition \prettyref{eq:conditionofbeta1} in \prettyref{thm:thm1hop} has been mentioned above. 
%To better appreciate our condition \prettyref{eq:conditionofbeta2}, we depict it
%as the blue curve in \prettyref{fig:condition}
%with three segments:
%depicts  for the $2$-hop algorithm and has three segments:

Our condition \prettyref{eq:conditionofbeta2} is depicted as the blue curve in \prettyref{fig:condition}. At a high level, the three terms in \prettyref{eq:conditionofbeta2} can be interpreted as follows:
\begin{itemize}
    \item The first term $\beta \gtrsim \frac{\log n}{n^2p^2}$ is to ensure that every true pair has more $2$-hop witnesses contributed by the correct seeds than every fake pair. 
    To see this, recall that there are $n\beta$ correct seeds. Since a vertex has about $np$ 1-hop neighbors, each correct seeds becomes a 2-hop witness for a true pair with probability about $np\cdot p=np^2$ and for a fake pair with probability about $(np^2)^2=n^2p^4$. %Then, for the success of the $2$-hop algorithm, it suffices to ensure that each true pair has more $2$-hop witnesses than fake pair. Hence, 
    Hence, to ensure that a true pair has more $2$-hop witnesses from the correct seeds than a  fake pair, we at least need the difference between their means, i.e.,  $(n\beta) np^2-(n\beta)n^2p^4$, to be positive. This is guaranteed by $np^2 \lesssim 1$, in which case the mean difference can be approximated by $\beta n^2p^2$. However, due to randomness, we also need this mean difference to be larger than the standard deviation, which is on the order of $\sqrt{\beta n^2p^2}$. This is guaranteed by  $\beta \gtrsim \frac{1}{n^2p^2}$. 
    %\Leftarrow \beta \gtrsim \frac{1}{n^2p^2}$
    %in condition \prettyref{eq:conditionofbeta2} except for the $\log n$ factor. 
    Adding the extra $\log n$ factor ensures that the above claim holds for every pair with high probability.
    %the 2-hop algorithm will succeed with high probability. 
    This condition coincides with the seed requirement established in \cite{mossel2019seeded} when the seeds are all correct;

    \item The second term $\beta\gtrsim \sqrt{\frac{\log n}{n}}$ is due to the negative impact of the incorrect seeds. 
    %However, the incorrect seeds may have a negative impact on matching graphs. 
    Note that there are $n(1-\beta)$ incorrect seed, and each seed becomes a 2-hop witness for both a true pair and a fake pair with probability about $n^2p^4$. Although this contributes the same mean number of witnesses to both a true pair and fake pair, its randomness may contribute more to a fake pair than to a true pair. Thus, we need its standard deviation (on the order of $\sqrt{n(1-\beta)n^2p^4}$) to be less than the mean difference $\beta n^2p^2$ estimated in the first bullet. This is guaranteed by $\beta \gtrsim \sqrt{\frac{1}{n}}$. Again, adding the $\log n$ factor ensures that the above claim holds for every pair with high probability.
    %the 2-hop algorithm will succeed with high probability.
    
    \item The third term $\beta \gtrsim \sqrt{np^3 (1-s) \log n}$ is 
    also caused by the incorrect seeds. However, the reason is more subtle than the second bullet, and is 
    due to the fluctuation of the number of 1-hop neighbors of a true pair. 
    %Note that all the above calculation ignores the fluctuation of the number of 1-hop neighbors. 
    Note that if $s=1$, then, in both $G_1$ and $G_2$, the two vertices corresponding to a true pair has the same set of 1-hop neighbors, and 
    thus the aforementioned fluctuation disappears. 
    %$\beta \gtrsim \sqrt{\frac{\log n}{n}}$ is sufficient to ensure the success of the 2-hop algorithm. %still dominates condition \prettyref{eq:conditionofbeta2} 
   % \nbr{Here we haven't clarified the regime of sparsity $p$, so I am not sure whether it is clear to say ``$\beta \gtrsim \sqrt{\frac{\log n}{n}} $ still dominates condition \prettyref{eq:conditionofbeta2}''}.    
   If instead $s<1$, then the vertices corresponding to a true pair will have a different set of $1$-hop neighbors in $G_1$ and $G_2$. This variation makes it even harder to distinguish the true pairs from the fake pairs based on the number of 2-hop witnesses as $p$ increases, which gives to the condition $\beta \gtrsim \sqrt{np^3(1-s) \log n}$.
   %Then, we need a more restricted condition $\beta \gtrsim \sqrt{np^3\log n}$. 
   Please refer to \prettyref{eq:fluctation-1-hop-new} in \prettyref{sec:derivation-tight-condition} for detailed derivation. 
   %\nbr{where we can find this discussion?}
   As a consequence, the blue curve has two branches: the top branch holds for $s<1$
    and the bottom one holds for $s=1.$
\end{itemize}

As readers can see, in the latter two cases, our condition captures the new effect of the incorrect seeds and thus are significantly different from the theoretical results in \cite{mossel2019seeded} where the seeds are all correct. Please refer to \prettyref{rmk:2_hop_cond_not_tight} in \prettyref{sec:proofthm2hop} for more detailed discussions.
%which is in sharp contrast to the theoretical results in \cite{mossel2019seeded} when the seeds are all correct;

%It is instructive to compare our condition \prettyref{eq:conditionofbeta2} to the seed requirement established for the $2$-hop algorithm in \cite{mossel2019seeded} when the seeds are all correct. It turns ou

Pictorially, the three terms lead to the three segments in the blue curve in \prettyref{fig:condition}:
\begin{itemize}
\item When $p \lesssim \left(\frac{\log n }{ n^3}\right)^{\frac{1}{4}}$,
the first term of \prettyref{eq:conditionofbeta2} dominates, as  the graphs are so sparse that the $2$-hop witnesses contributed by the incorrect seeds become negligible;
\item When $ \left(\frac{\log n }{ n^3}\right)^{\frac{1}{4}} \lesssim p \lesssim n^{-\frac{2}{3}}$,
the second term dominates, as the influence of the incorrect seeds cannot be ignored. In this case, the bottleneck for the success of the $2$-hop algorithm is due to the statistical fluctuation of the $2$-hop witnesses contributed by the incorrect seeds;
\item When $ n^{-\frac{2}{3}} \lesssim p \lesssim (n\log n)^{-\frac{1}{2}}$ and $s<1$,
the third term dominates, as 
the fluctuation of the $1$-hop neighborhood sizes of the true pair increases with $p$ 
and becomes the new bottleneck. 
\end{itemize} 

From  \prettyref{fig:condition}, we observe that 
our 2-hop algorithm requires substantially fewer correct seeds to succeed than the 1-hop algorithm when the graphs are sparse. Moreover, our 2-hop algorithm is comparable to the NoisySeeds algorithm for very sparse graphs when $p \ll n^{-\frac{5}{6}}$, but continues to perform well over a much wide range of graph sparsity up to $p \lesssim (n\log n)^{-1/2}.$

% The comparison between the blue curve and the red curve in \prettyref{fig:condition} implies that our 2-hop algorithm requires many fewer correct seeds to succeed than the 1-hop algorithm when the graphs are sparse. We observe that 
% condition \prettyref{eq:conditionofbeta2}  is more relaxed than condition \prettyref{eq:conditionSrikant} and \prettyref{eq:conditionofbeta1} for sparse graphs, provided that  $p=O((\frac{\log n }{ n^3})^{\frac{1}{5}})$ when $s$ is a constant smaller than 1 and that $p=O\left(n^{-\frac{1}{2}}\right)$ when $s=1$. 

% The overlap of the blue curve and the black curve in \prettyref{fig:condition} implies that the performance guarantee of our 2-hop algorithm is comparable to that of the NoisySeeds algorithm for very sparse graphs.
% For $p< n^{-\frac{5}{6}}$, condition \prettyref{eq:conditionofbeta2} differs from condition \prettyref{eq:conditionPGM}  by only a logarithmic factor. On the other hand, condition  \prettyref{eq:conditionofbeta2} provides a stronger guarantee of fully recovery while condition \prettyref{eq:conditionPGM} only guarantees $(1-o(1))$-partial recovery. 

% In summary, the performance guarantee of the 2-hop algorithm not only is comparable to that of \cite{10.14778/2794367.2794371} for very sparse graphs ($p< n^{-\frac{5}{6}}$), but also significantly improves that of \cite{lubars2018correcting} for sparse graphs ($p \le n^{-\frac{1}{2}}$). 

\section{Analysis}\label{sec:intuition}

In this section, we explain the intuition and sketch the proofs for \prettyref{thm:thm1hop} and \prettyref{thm:thm2hop}.
In the analysis, we assume without loss of generality that the true mapping $\pi^*$ is the identity mapping, i.e., $\pi^*(i)=i$.

\subsection{Intuition and Proof of \prettyref{thm:thm1hop}}\label{sec:proofthm1hop}

To understand the intuition behind \prettyref{thm:thm1hop} and why it provides a better result than \cite{lubars2018correcting}, recall that the 1-hop algorithm will succeed (in recovering $\pi^*$) if the number of 1-hop witnesses for any true pair is larger than the number of 1-hop witnesses for any fake pair. For any correct seed, it is a 1-hop witness for a true pair with probability $ps^2$ and is a 1-hop witness for a fake pair with probability $p^2s^2$. In contrast, for any incorrect seed, it is a 1-hop witness with probability $p^2s^2$ for both true pairs and fake pairs. Since there are $n\beta$ seeds that are correct, it follows that
\begin{subequations}
   \begin{empheq}[left={W_1(u,v)\overset{\cdot}{\sim}\empheqlbrace\,}]{align}
        &\Binom\left(n\beta,ps^2\right)+\Binom\left(n(1-\beta),p^2s^2\right) & \text{if } u&=v,   \label{eq:binom1-hop1}\\
    &\Binom(n,p^2s^2)  & \text{if }  u&\neq v.\label{eq:binom1-hop2}
  \end{empheq} 
\end{subequations}
where $\overset{\cdot}{\sim}$ denotes ``approximately distributed".

For fake pair $u\neq v$, using Bernstein’s inequality given in \prettyref{thm:bernstein} in Appendix \ref{sec:bound}, we show that $W_1(u,v)$ is upper bounded by $np^2s^2+ O(\sqrt{np^2s^2\log n})+O(\log n)$ with high probability. More precisely, we have the following lemma, with the proof deferred to Appendix \ref{sec:proof1hopl2}. 
 \begin{lemma}\label{lmm:W1uv} For any two vertices $u,v\in [n]$  with $u\neq v$ and sufficiently large $n$, the following holds
\begin{equation}
\begin{aligned}
\prob{{W_1(u,v)}<\psi_{\max}}\geq 1- n^{-\frac{7}{2}},
\end{aligned}\label{eq:W1uv}
\end{equation}
where $\psi_{\max}=np^2s^2+\sqrt{7np^2s^2\log n}+\frac{7}{3}\log n+2$.
\end{lemma}
 
For true pair $u=v$, %by using Bernstein’s inequality given in \prettyref{thm:bernstein}, a random variable distributed as $\Binom(n\beta,ps^2)$ is lower bounded by $n\beta ps^2- O(\sqrt{nps^2})-O(\log n)$ with high probability. However, 
the first binomial distribution in \prettyref{eq:binom1-hop1} can be lower bounded by  $n\beta ps^2- O(\sqrt{n\beta ps^2\log n})-O(\log n)$ with high probability using Bernstein’s inequality. However, the second Binomial distribution in \prettyref{eq:binom1-hop1} is not precise because the events that each incorrect seed becomes a witness for a true pair are dependent on each other, as we discussed in \prettyref{sec:idea}. We address this dependency issue using the concentration inequality for dependent random variables \cite{Janson2004LargeDF}, and get the following lower bound on the number of 1-hop witnesses for the true pairs.
 
 \begin{lemma}\label{lmm:W1uu} For any vertex $u\in [n]$ and sufficiently large $n$, the following holds
\begin{equation}
\begin{aligned}
\prob{{W_1(u,u )}> x_{\min}+y_{\min}}\geq 1-n^{-\frac{7}{3}},\label{eq:W1uu} \end{aligned}   
\end{equation}
where 
\begin{align*}
x_{\min} & =(n\beta-1)ps^2-\sqrt{5n\beta ps^2\log n}-\frac{5}{3}\log n,\\ 
y_{\min} &=(n(1-\beta)-2)p^2s^2 -5\sqrt{np^2s^2\log n}-\frac{25}{3}\log n.
\end{align*}
\end{lemma}
\begin{proof}[Proof of \prettyref{lmm:W1uu}]
Recall that $A_1$ and $B_1$ are the adjacency matrix for $G_1$ and $G_2$, respectively. 
Let $F\triangleq \{ i: \pi(i)=i\}$ denote the set of fixed points of $\pi$.
Then $F$ corresponds to the set of correct seeds with $|F|=n\beta$ (recall that we assume the true matching $\pi^*$ to be the identity mapping). 
By the definition of $1$-hop witness, we have 
\begin{align}
    W_1(u,u) & =\sum_{i\in F} A_1(u,i)B_1(u,i)+\sum_{i\in [n]\setminus F} A_1(u,i)B_1(u,\pi(i)) \nonumber \\
    & = \sum_{i\in F\setminus \{u\} } A_1(u,i)B_1(u,i)+\sum_{i\in [n]\setminus (F\cup \{u,\pi^{-1}(u) \})} A_1(u,i)B_1(u,\pi(i))
    \label{eq:1_hop_true_decomp},
\end{align}
where the second equality holds because $A_1(u,u)=B_1(u,u)=0.$
Let $X_i \triangleq A_1(u,i)B_1(u,i)$ for $i \in F\setminus\{u\}$ and $Y_i\triangleq A_1(u,i)B_1(u,\pi(i))$ for $i \in [n]\setminus (F\cup \{u,\pi^{-1}(u) \})$.
%Fix a true pair $(u,\pi(u))$, we next bound the number of 1-hop witnesses for $(u,\pi(u))$. For each seed $(u_i,\pi(u_i))$, let $\pi(v_i)$ be the underlying vertex matched to $u_i$, \ie, $\pi(u_i)=\pi(v_i)$. Then, $(u_i,{\pi}(v_i))$ is a correct seed if $u_i=v_i$ and is an incorrect seed if $u_i\neq v_i$. Among all seeds, some of them may be of the form that $u_i=u$ or $v_i=u$. Then, they can not become 1-hop witnesses for $u$ and $\pi(u)$. The number of such seeds is at most 2. In the following, we exclude such seeds.

%We first count the contribution to $W_1(u,\pi(u))$ by correct seeds.
%We first bound $\sum_{i\in T} X_i$ from below. 
For all $i \in F \setminus \{u\}$,  $X_i\iiddistr \Bern(ps^2)$. 
 %For any correct seed $(u_i,\pi(u_i))$, let $X_i$ be a binary random variable such that $X_i=1$ if $(u_i,\pi(u_i))$ is a 1-hop witness for $u$ and $\pi(u)$, and $X_i=0$ otherwise. 
%We have $\prob{X_i=1}=ps^2$ because $X_i=1$ if and only if the edge $(u,u_i)$ is in $G$, and is sampled into both $G_1$ and $G_2$. Note that $X_i$ only hinges on the edge $(u,u_i)$, and the edges $(u,u_i)$'s do not interact each other for different $u_i$. Thus, $X_i$'s are mutually independent across $u_i$.  
%Let $X$ denote the number of 1-hop witnesses contributed by the correct seeds. Since $u_i\neq u$, there are at least $n\beta-1$ correct seeds that could be 1-hop witnesses for the true pair $(u,u)$. It follows that $X\geq \sum_{i=1}^{n\beta-1}X_i\sim\Binom(n\beta-1,ps^2)$. 
%Recall that $x_{\min}=(n\beta-1)ps^2-\sqrt{5n\beta ps^2\log n}-\frac{5}{3}\log n$.
It follows that
\begin{equation}
\begin{aligned}
%\prob{X\leq x_{\min}}\leq&\prob{\sum_{i=1}^{n\beta-1}X_i\leq x_{\min}}\\
\prob{\sum_{i\in F} X_i \le x_{\min} }
\le \prob{ \Binom(n\beta-1, ps^2) \le x_{\min} }  
%\leq &\mathbb{P}\left\{\sum_{i=1}^{n\beta-1}X_i\leq(n\beta-1)ps^2-\sqrt{5(n\beta-1)ps^2(1-ps^2)\log n}-\frac{5}{3}\log n\right\}\\
%\leq &\exp\left(-\frac{5}{2}\log n\right)
\le n^{-\frac{5}{2}},\label{eq:lowboundcorrectseed1}
\end{aligned}
\end{equation}
where that last inequality follows from Bernstein’s inequality given in \prettyref{thm:bernstein} with  $\gamma=\frac{5}{2}\log n$ and $K=1$.

For all $i \in [n]\setminus \left(F\cup  \{u, \pi^{-1} (u) \}\right)$, $Y_i \sim \Bern(p^2s^2)$. However, $Y_i$'s are dependent and thus
we cannot directly apply Bernstein's inequality. To see this,  $A_1(u,i)$ and $B_1(u,i)$ are correlated,
but $\{A_1(u,i), B_1(u,i)\}$ are independent across different $(u,i)$.
Let $S_i=\{\{u,i\},\{u,\pi(i) \}\}$.
Thus, $Y_i$ only depends on the set 
$S_i$ of entries of $A_1$ and $B_1$. 
Since $S_i \cap S_{i'} \neq \emptyset$ if and only if $i'=\pi(i)$ or $i'=\pi^{-1}(i)$, 
it follows that $Y_i$ is dependent on $Y_{i'}$ if and only if $i'=\pi(i)$ or $i'=\pi^{-1}(i)$. 
Then we can construct a dependency graph $\Gamma$ for $\{Y_i\}$, where
the maximum degree of $\Gamma$, $\Delta(\Gamma)$, equals to two.
Hence, 
%we cannot apply Bernstein’s Inequality. 
%Instead, 
applying the concentration inequality for the sum of dependent random variables given in \prettyref{thm:drv} 
%Specifically, 
%Hence, it follows from \prettyref{thm:drv} 
with $\gamma=\frac{8}{3}\log n$ and $K=1$ yields that 
\begin{equation}
\begin{aligned}
%\prob{\sum_{i\in [n]\setminus F} Y_i \le y_{\min} }
%=
\prob{\sum_{i \in [n]\setminus \left(F\cup  \{u, \pi^{-1} (u) \}\right) } Y_i \le y_{\min} }
%&\prob{Y\leq y_{\min}}\\
%\leq&\prob{\sum_{i=1}^{n(1-\beta)-2}Y_i\leq y_{\min}}\\
%\leq&\prob{\sum_{i=1}^{n(1-\beta)-2}Y_i\leq (n(1-\beta)-2)p^2s^2 -5\sqrt{(n(1-\beta)-2)p^2s^2(1-p^2s^2)\log n}-\frac{25}{3}\log n}\\
%\leq &\exp\left(-\frac{8}{3}\log n\right)
\le n^{-\frac{8}{3}}.\label{eq:lowboundincorrectseed1}
\end{aligned}
\end{equation}

Finally, combining \prettyref{eq:1_hop_true_decomp},  (\ref{eq:lowboundcorrectseed1}) and (\ref{eq:lowboundincorrectseed1}) and applying  union bound yields the desired conclusion~\prettyref{eq:W1uu}. 
%Finally, since ${W_1(u,\pi(u))}= X+Y$ and $n$ is sufficiently large, (\ref{eq:lowboundcorrectseed1}) and (\ref{eq:lowboundincorrectseed1}) yield that
%\begin{align*}
%\prob{{W_1(u,\pi(u))}\leq x_{\min}+y_{\min}}\leq\prob{X\leq x_{\min}}+\prob{Y\leq y_{\min}}\leq n^{-\frac{5}{2}}+n^{-\frac{8}{3}}< n^{-\frac{7}{3}}.
%\end{align*}
\end{proof}

Combining \prettyref{lmm:W1uv} and \prettyref{lmm:W1uu}, for the 1-hop algorithm to succeed, it suffices to ensure that $x_{\min}+y_{\min}\geq \psi_{\max}$. Note that
\begin{equation}
 \begin{aligned}
&x_{\min}+y_{\min}-\psi_{\max}\geq 0\\
%\Leftrightarrow&n\beta p(1-p)s^2-\sqrt{5n\beta ps^2\log n}-(5+\sqrt{7})\sqrt{np^2s^2\log n}-2p(1+p)s^2-\frac{37}{3}\log n-2\geq 0\\
\Leftarrow& \frac{1}{3}n\beta p(1-p)s^2\geq \sqrt{5n\beta ps^2\log n}, \ \mbox{and}\ \frac{1}{3}n\beta p(1-p)s^2\geq (5+\sqrt{7})\sqrt{np^2s^2\log n}, \ \mbox{and}\\
&\frac{1}{3}n\beta p(1-p)s^2\geq   \frac{37}{3}\log n+2+ps^2+2p^2s^2,\label{eq:meangstd1hop}
\end{aligned}   
\end{equation}
which is implied by condition \prettyref{eq:conditionofbeta1} in \prettyref{thm:thm1hop}. Thus, by taking the union bound over \prettyref{eq:W1uv} and \prettyref{eq:W1uu}, we complete the proof of \prettyref{thm:thm1hop}. Please refer to Appendix \ref{sec:proof1hop} for details. The above argument suggests that the sufficient condition \prettyref{eq:conditionofbeta1} is also  close to necessary (differing from the necessary condition by a constant factor) for the 1-hop algorithm to succeed, which is confirmed by our simulation results in Appendix \ref{sec:exp}.

\subsection{Intuition and Proof of \prettyref{thm:thm2hop}}\label{sec:proofthm2hop}
We next explain the intuition and 
sketch the proof of \prettyref{thm:thm2hop} 
%for the correlated \ER model $\mathcal{G}(n,p;s)$ 
when $np^2 \le \frac{1}{\log n}$ 
%\nbr{I was wondering where we explain why we need this condition. LY: this condition guarantees that grpahs are sparse} 
%\nb{I know this. I meant have we explained this in the paper? LY: I add (15)}
and $nps^2\geq 128 \log n$. 

We start by explaining why \prettyref{thm:thm2hop} requires $np^2 \le \frac{1}{\log n}$ and $nps^2\geq 128 \log n$. First, note that 
the intersection graph $G_1 \land G_2$ (which includes edges appearing in both $G_1$ and $G_2$) is an \ER random graph with average degree $(n-1)ps^2$. Thus, we need $nps^2\geq 128 \log n$ so that 
there is no isolated vertex in $G_1 \land G_2$. Otherwise, it is impossible to match the isolated vertices and reach the goal of perfect matching~\cite{cullina2016improved}. 
Moreover, we will use this condition in \prettyref{sec:bound-1-hop} to  ensure that the number of 1-hop neighbors is concentrated.
%\textbf{}\nbr{Where we need this condition in the analysis?}
%, because 
%it is impossible to correctly match those isolated nodes. 
Second, the condition $np^2 \le 1/\log n$ ensures that the graph is not too dense so that
the true pair is expected to have more $2$-hop witnesses than the fake pair. 
Please see later in \prettyref{eq:meanstd2hop} how this condition arises. 
%we need that the average
%degree of every vertex should be high enough so that there is no isolated vertex in $G_1$ or $G_2$. For the correlated \ER model, $nps^2-\log n\to +\infty$ ensures the intersection graph $G_1 \land G_2$ to be connected \cite{Erdos59}. 

Then, analogous to the 1-hop algorithm, we derive the condition on $\beta$ by comparing the number of 2-hop witnesses for true pairs and for fake pairs. However, the dependency issue is more severe  here when we bound the number of 2-hop witnesses. Specifically, in the analysis of \prettyref{lmm:W1uu}, the event that an incorrect seed becomes a 1-hop witness for a true pair is dependent on that of at most  two other incorrect seeds. However, for 2-hop witnesses, any two seeds could be dependent through the 1-hop neighborhoods of the candidate vertex-pair. Thus, directly using the concentration inequality in \cite{Janson2004LargeDF} will lead to a poor bound. To address this new difficulty, we will condition on the 1-hop neighborhoods first. After this conditioning, the remaining dependency becomes more manageable, which is handled by either 
classifying the seeds or by applying the concentration inequality in \cite{Janson2004LargeDF} again. 

\subsubsection{Bound on the 1-hop Neighbors}\label{sec:bound-1-hop}

In order to condition on the typical sizes of the 1-hop neighborhoods, we first bound the number the number of 1-hop neighbors. For any vertex $u$ in graph $G$, we use $N^{G}(u)$ to denote the set of $1$-hop neighbors of $u$ in $G$, \ie,
$
   N^G(u)= \left\{v\in G:d^G(u,v)=1 \right\}.
$
For any two vertices $u,v\in [n]$, let $C(u,v)$ denote the set of 1-hop ``common" neighbors of $u$ and $v$ across $G_1$ and $G_2$, \ie,
$
   C(u,v)= N^{G_1}(u) \cap N^{G_2}(v).
$
For ease of notation, let 
$$
d_u =\abs{N^{G_0}(u)}, \quad a_u=\abs{N^{G_1}(u)}, \quad b_v=\abs{N^{G_2}(v)},
\quad c_{uv}=\abs{C(u,v)}.
$$
% and
% $$
% a_{u\backslash v}=\abs{N^{G_1}(u)\setminus\{v\}}, \quad b_{u\backslash v}=\abs{N^{G_2}(u)\setminus\{v\}}, \quad b_{v\backslash u}=\abs{N^{G_2}(v)\setminus\{u\}}.
% $$
By definition, we have $a_u, b_v \sim \Binom(n-1,ps)$, $c_{uu} \sim \Binom(n-1,ps^2)$, and $c_{uv} \sim \Binom(n-1,p^2s^2)$ for $u\neq v$. 
%$$
%a_u,a_v,b_u,b_v\overset{\cdot}{\sim}\Binom(n,ps),\quad c_{uu},c_{vv}\overset{\cdot}{\sim}\Binom(n,ps^2).
%$$
Thus, by using concentration inequalities for binomial distributions and letting
\begin{equation}
    \epsilon=\sqrt{\frac{12\log n}{(n-1)ps^2}}\leq\frac{1}{3},\label{eq:epsilon}
\end{equation}
we can show that with high probability, $a_u, b_u$ are bounded by $(1\pm \epsilon)nps$, $c_{uu}$ is bounded by $(1\pm \epsilon)nps^2$, and 
$c_{uv}$ is upper bounded by $\psi_{\max}$ in \prettyref{lmm:Ruv} below. 
%. Since $c_{uv},{W_1(v,u)}\overset{\cdot}{\sim}\Binom(n,p^2s^2)$, 
%they can be bounded by sub-exponential tail bounds. Thus, we can 
In particular, we arrive at the following lemma with the proof deferred to Appendix \ref{sec:proofRuv}.

\begin{lemma}\label{lmm:Ruv}
Given any two vertices $u,v\in [n]$ with $u\neq v$, let $R_{uv}$ denote the event such that the followings hold simultaneously:
 \begin{gather*}
(1-\epsilon)(n-1)ps<a_u,a_v,b_u,b_v<(1+\epsilon)(n-1)ps,\\ (1-\epsilon)(n-1)ps^2<c_{uu},c_{vv}<(1+\epsilon)(n-1)ps^2,\\
c_{uv}, \, W_1(v,u)<\psi_{\max},
\end{gather*} 
where $\psi_{\max}=np^2s^2+\sqrt{7np^2s^2\log n}+\frac{7}{3}\log n+2$.
%\nb{This is not my point... It is OK to derive a simplified upper bound $3\log n$ by imposing extra conditions on $n$ and $p.$ LY: I do not get your point}
%where $\epsilon$ is given in $\prettyref{eq:epsilon}$, \ie, $\epsilon=\sqrt{\frac{12\log n}{(n-1)ps^2}}\leq\frac{1}{3}$.

If  %\nb{Do you really need this strong assumption for this lemma? If yes, please explain why.}
$nps^2\ge 128\log n$, then for
all sufficiently large $n$,
\begin{align}\label{eq:Ruv}
\prob{R_{uv}}\geq 1- n^{-\frac{7}{2}}.
\end{align}
\end{lemma}
%\nbr{The above lemma needs assumptions on $n, p$, etc. Can you add all the assumptions to make the statement of lemma independent from the main theorem? Can you also check the other lemma? If some lemma has almost identical distributions to the main theorem, then you can say ``suppose the assumption in Theorem XXX hold...''LY: addressed}

 \subsubsection{Bound on the 2-hop Witnesses}
In the sequel, we condition on the 1-hop neighborhoods of $u$ and $v$ such that event $R_{uv}$ holds, and bound the 2-hop witnesses for both the true pairs and fake pairs. To compute the probability that a seed $(j,\pi(j))$ becomes a 2-hop witness
for pair $(u,v)$, 
we calculate the \emph{joint probability} that
$j$ connects to some 1-hop neighbor of $u$ in $G_1$ and $\pi(j)$ connects to some 1-hop neighbor of $v$ in $G_2$. 
For any correct seed, it is a 2-hop witness for a true pair $(u,u)$ with probability about $c_{uu} ps^2$ and is a 2-hop witness for a fake pair $(u,v)$ with probability about $a_u b_v p^2s^2$. In contrast, for any incorrect seed, it is a 2-hop witness for a true pair $(u,u)$ with probability about $a_u b_u p^2s^2$ and is a 2-hop witness for a fake pair $(u,v)$ with probability about $a_u b_v p^2s^2$. Thus we have
\begin{subequations}
   \begin{empheq}[left={W_2(u,v)\overset{\cdot}{\sim}\empheqlbrace\,}]{align}
        & \Binom\left(n\beta, c_{uu} ps^2\right) +\Binom\left(n(1-\beta),a_u b_u p^2s^2\right) & \text{if } u&=v,    \label{eq:binom2-hop1}\\
    &\Binom\left(n,a_u b_v p^2s^2\right) & \text{if } u&\neq v.\label{eq:binom2-hop2}
  \end{empheq} 
\end{subequations}
%Note that $a_u, b_u=nps\pm O\left(\sqrt{nps}\right)$ and $c_u=nps^2\pm O\left(\sqrt{nps^2}\right)$ with high probability.

To ensure that the numbers of $2$-hop witnesses are separated between true pairs and fake pairs, we need $\expect{W_2(u,u)} \ge \expect{W_2(u,v)}$ for $u \neq v$, which, in view of \prettyref{eq:binom2-hop1} and \prettyref{eq:binom2-hop2}, 
$a_u, b_u, b_v \approx nps$, 
and $c_{uu} \approx nps^2$, amounts to 
%We first consider the basic condition that the difference between the conditional expected values of the two cases in \prettyref{eq:binom2-hop} to be greater than zero. Then, we need 
\begin{equation}
\begin{aligned}
n\beta(nps^2)ps^2 +n(1-\beta) (nps)^2p^2s^2 -n(nps)^2p^2s^2\geq 0 \Leftrightarrow np^2\leq 1.\label{eq:meanstd2hop}
\end{aligned}
\end{equation}
This shows that the $2$-hop algorithm is only effective when the graphs are sufficiently sparse.
For this reason, we assume $np^2 \le 1/\log n$ so that \prettyref{eq:meanstd2hop} is satisfied.
%which corresponds to the condition of sparse graphs.

%Note that $a_u, b_u=nps\pm O\left(\sqrt{nps}\right)$ and $c_u=nps^2\pm O\left(\sqrt{nps^2}\right)$ with high probability. 
For the $2$-hop algorithm to be effective, 
we also need to consider the statistical fluctuation of $W_2(u,v)$. For true pair $u=v$, using Bernstein’s inequality, $\Binom(n\beta,c_{uu}ps^2)$ is lower bounded by $n\beta c_{uu}ps^2- O(\sqrt{n\beta c_{uu}ps^2\log n})-O(\log n)$ with high probability. However, the second Binomial distribution in \prettyref{eq:binom2-hop1} is not precise because the events that each incorrect seed becomes a 2-hop witness for a true pair are dependent on other incorrect seeds. Fortunately, similar to the proof of \prettyref{lmm:W1uu}, we can deal with this dependency issue using the concentration inequality for dependent random variables \cite{Janson2004LargeDF}. Thus, we can get the following lower bound on the number of 2-hop witnesses for the true pairs conditional on the 1-hop neighborhoods.

\begin{lemma}\label{lmm:W2uu}
Given any two vertices $u, v\in [n]$ with $u\neq v$, we use $\Quv$ to collect all information of 1-hop neighborhood of $u$ and $v$, i.e.,
\begin{align*}
\Quv=\left\{N^{G_1}(u),N^{G_2}(u),N^{G_1}(v),N^{G_2}(v)\right\}.
\end{align*}

%If $R_{uv}$ occurs, i.e., if $Q_{uv}$ satisfies the condition $R_{uv}$ in \prettyref{lmm:Ruv}, then 
If $n$ is sufficiently large and $nps^2\ge 128\log n$, then 
\begin{align}\label{eq:W2uu}
\prob{{W_2(u,u)}\leq l_{\min}+m_{\min}\mid\Quv}\cdot\mathds{1}(R_{uv})\leq n^{-\frac{7}{2}},
\end{align}
where 
\begin{align}
l_{\min} =&\frac{7}{24}(1-\delta_1) \beta n^2 p^2s^4  - \sqrt{\frac{35}{16} \beta n^2 p^2s^4\log n}-\frac{5}{2}\log n,\label{eq:lmin}\\%\label{eq:lmin} \\
%m_{\min}=&(1-\delta_2)n(1-\beta)\left(1-(1-ps)^{\abs{N^{G_1}(u)\setminus\{v\}}}\right)\left(1-(1-ps)^{\abs{N^{G_2}(u)\setminus\{v\}}}\right)\\&- \frac{15}{2}\sqrt{\frac{3}{2}n^3p^4s^4\log n}-\frac{25}{2}\log n.\label{eq:mmin}
m_{\min}  =&n(1-\beta)\left(1-(1-ps)^{a_{u\backslash v}}\right)\left(1-(1-ps)^{b_{u\backslash v}}\right)\nonumber\\&-21n^3p^5s^5- \frac{15}{2}\sqrt{\frac{3}{2}n^3p^4s^4\log n}-\frac{25}{2}\log n, \label{eq:mmin}%\label{eq:mmin}
\end{align}
with $\delta_1=\frac{6ps}{\beta}$, 
$a_{u\backslash v} = \abs{N^{G_1}(u)\setminus\{v\}}$,
and $b_{u\backslash v}= \abs{N^{G_2}(u)\setminus\{v\}}$.
\end{lemma}
\begin{remark}
Note that $l_{\min}$ is contributed by the correct seeds and
$m_{\min}$ is contributed by the incorrect seeds. Specifically, conditional on the 1-hop neighbors, a correct seed becomes a 2-hop witness for the true pair $(u, u)$ with probability about $c_{uu} ps^2 \approx
np^2s^4$. Multiplying by $n\beta$ gives an expression close to the first term of $l_{\min}$. Similarly, an incorrect seed becomes a 2-hop witness for the true pair $(u,u)$ with probability about $\left(1-(1-ps)^{a_{u\backslash v}}\right)\left(1-(1-ps)^{b_{u\backslash v}}\right)$. Multiplying by $n(1-\beta)$ gives the first term of $m_{\min}$. In summary, the first term in $l_{\min}$ and $m_{\min}$ is a lower bound of the expectation, and the rest of the terms are due to the tail bounds.
\end{remark}
Due to the conditioning of 1-hop neighborhoods, we exclude seeds that are 1-hop neighbors of $u$ when
bounding $W_2(u,u)$, giving rise to the additional $\delta_1$ and $21n^3p^5s^5$ terms in \prettyref{lmm:W2uu}. 
Please refer to Appendix \ref{sec:proofW2uu} for the proof.

%\begin{remark}
%\prettyref{lmm:W2uu} provides a lower bound on the number of 2-hop witnesses for the true pair, $u$ and $\pi(u)$, conditioned on $\Quv$. Note that $l_{\min}$ is contributed by the correct seeds and $m_{\min}$ is contributed by the incorrect seeds. Conditional on $Q_{uv}$, a correct seed becomes a 2-hop witness for the true pair $(u,u)$ with probability about $c_{uu}ps^2\approx np^2s^4$. An incorrect seed becomes a 2-hop witness for the true pair $(u,u)$ with probability about $(1-(1-ps)^{a_u})(1-(1-ps)^{b_u})$.Both the expressions of $l_{\min}$ and $m_{\min}$ consist of three parts. Specifically, the first term of \prettyref{eq:lmin} and \prettyref{eq:mmin} is a lower bound of the expectation, the second term is due to the sub-Gaussian tail bound, and the third term is due to the sub-exponential tail bound. Please refer to Appendix \ref{sec:proofW2uu} for the proof. \end{remark}
\medskip
For the fake pair $u\neq v$, 
%whether the seeds are $1$-hop neighbors of $v$ or $u$ would influence the probability that they become $2$-hop witnesses for $(u,v)$. Thus, we need to divide the seeds into several types and get 
we have the following upper bound on the number of 2-hop witnesses for the fake pairs conditional on the 1-hop neighborhoods.

\begin{lemma}\label{lmm:W2uv}
For any two vertices $u,v\in [n]$ with $u\neq v$, if   $nps^2\ge 128\log n$, then for all sufficiently large $n$,
\begin{equation}\label{eq:W2uv}
\prob{{W_2(u,v)}\geq x_{\max}+y_{\max}+2z_{\max}+\psi_{\max}+28\log n\mid\Quv}\cdot\Indicator{R_{uv}}\\
\leq n^{-\frac{7}{2}}.
\end{equation}
where 
\begin{align}
x_{\max} & =2 n\beta\left( \psi_{\max}ps+\frac{9}{4}n^2p^4s^4\right), \label{eq:xmax}\\
y_{\max} &= n(1-\beta)\left(1-(1-ps)^{a_{u\backslash v}}\right)\left(1-(1-ps)^{b_{v\backslash u}}\right)+n^2p^3s^3
+\frac{5}{2}\sqrt{15n^3p^4s^4\log n},\label{eq:ymax} \\
z_{\max} &=\frac{9}{2}n^2p^3s^3,\nonumber\\
\psi_{\max}&=np^2s^2+\sqrt{7np^2s^2\log n}+\frac{7}{3}\log n+2.\nonumber
\end{align}
\end{lemma}

\begin{remark}
Note that if $u$ and $v$ are connected in $G_1$, the conditioning on $Q_{uv}$ changes the probability that the seed $(j,\pi(j))$ with $j\in N^{G_1}(v)$ becomes a 2-hop witness for $(u,v)$. Thus, we have to divide the seeds into several types depending on whether $j\in N^{G_1}(v)$ or $\pi(j)\in N^{G_2}(u)$, and consider their contribution to the number of 2-hop witnesses separately:
\begin{itemize}
\item [1)] $x_{\max}+y_{\max}$ is the major term in \prettyref{eq:W2uv} and is contributed by the seeds such that $j\notin N^{G_1}(v)\cup\pi^{-1}\left(N^{G_2}(u)\right)$. In the analysis, we further divide such seeds into two categories, where $x_{\max}$ is contributed by the correct seeds,  and $y_{\max}$ is contributed by the incorrect seeds. Specifically, conditional on the 1-hop neighbors, a correct seed becomes a 2-hop witness for the fake pair $(u, v)$ either when the two vertices of the seed connect to different 1-hop neighbors of $u$ and $v$, respectively, or when they connect to a common 1-hop neighbor of  $u$ and $v$. Thus, the conditional probability of such event is about $c_{uv}ps+a_u b_v ps^2$.  According to \prettyref{lmm:Ruv}, $\psi_{\max}$ is an upper bound estimate of $c_{uv}$, and both $a_u$ and $b_v$ are approximately $nps$. Therefore, the above conditional probability can be approximately estimated as $\psi_{\max}ps+n^2p^4s^4$. Multiplying by $n\beta$ gives an expression close to $x_{\max}$. Similarly, an incorrect seed becomes a 2-hop witness for the fake pair $(u,v)$ with probability about $\left(1-(1-ps)^{a_{u\backslash v}}\right)\left(1-(1-ps)^{b_{v\backslash u}}\right)$. Multiplying by $n(1-\beta)$ gives the first term of $y_{\max}$. In summary, the first term in $x_{\max}$ and $y_{\max}$ is an upper bound of the expectation, and the rest of the terms are due to the tail bounds.
\item [2)] One multiple of $z_{\max}$ in (\ref{eq:W2uv}) is contributed by the seeds such that $j\in N^{G_1}(v)\setminus\pi^{-1}\left( N^{G_2}(u)\right)$. To see this, note that there are roughly $nps$ such seeds $(j,\pi(j))$. If $u$ and $v$ are connected in $G_1$, then $j$ must be a 2-hop neighbor of $u$, i.e., $A_2(u,j)=1$. On the other hand, the probability that $\pi(j)$ becomes a 2-hop neighbor of $v$ is approximately $np^2s^2$. Thus, the expected number of 2-hop witnesses contributed by this type of seeds is approximately $n^2p^3s^3$. The other multiple of $z_{\max}$ in (\ref{eq:W2uv}) is for the opposite case: it is contributed by the seeds such that $j\in \pi^{-1}\left(N^{G_2}(u)\right)\setminus N^{G_1}(v)$. 
\item [3)] The term $\psi_{\max}$ in (\ref{eq:W2uv}) is contributed by the seeds such that $j\in N^{G_1}(v)\cap \pi^{-1}\left( N^{G_2}(u)\right)$. In this case, $(j,\pi(j))$ becomes a 1-hop witness for $(v,u)$. Since $W_1(v,u)<\psi_{\max}$ according to \prettyref{lmm:Ruv},  there are at most $\psi_{\max}$
such seeds. 
\item [4)] The term $28\log n$ in (\ref{eq:W2uv}) comes from the sub-exponential tail bounds when applying concentration inequalities. 
%we obtain the upper bound $31 \log n$.
%\item [4)] $2$ in (\ref{eq:W2uv}) is contributed by the seeds such that $u_i=v$ or $\pi(v_i)=\pi(u)$.
\end{itemize}

Please refer to Appendix \ref{sec:proofW2uv} for the proof.
\end{remark}

\subsubsection{Derivation of a Sub-optimal Version of Condition \prettyref{eq:conditionofbeta2}}\label{rmk:2_hop_cond_not_tight}
By combining \prettyref{lmm:W2uu} and \prettyref{lmm:W2uv}, we are ready to derive a sufficient (but not tight) condition for the success of the 2-hop algorithm. First, analogous to the proof of \prettyref{thm:thm1hop}, for the 2-hop algorithm to succeed, it suffices that
\begin{align}
\min_{u} W_2 (u, u) > \max_{u\neq v} W_2 (u,v).
\label{eq:2_hop_cond_weak}
\end{align}
Then by combining \prettyref{lmm:W2uu} and \prettyref{lmm:W2uv}, \prettyref{eq:2_hop_cond_weak}
is guaranteed when 
\begin{align}\label{eq:lmxyz}
l_{\min}+m_{\min}\geq x_{\max}+y_{\max}+2z_{\max}+\psi_{\max}+28\log n. 
\end{align}
Finally, to ensure \prettyref{eq:lmxyz} is satisfied when $np^2 \le \frac{1}{\log n}$ and $nps^2 \ge 128\log n$, we 
arrive at the following sufficient condition:
\begin{align}\label{eq:oldcriteria}
    \beta\gtrsim \max\left\{ \frac{\log n}{n^2p^2s^4}, \, \sqrt{\frac{\log n}{ns^4}},\, \sqrt{ \frac{np^3\log n}{s}}\right\}.
\end{align}
%It is instructive to see how the three terms in~\prettyref{eq:oldcriteria} arise: 
Note that condition \prettyref{eq:oldcriteria} is similar to condition \prettyref{eq:conditionofbeta2} except for the third term. It is instructive to see  how \prettyref{eq:oldcriteria} implies \prettyref{eq:lmxyz}:
\begin{itemize}
    \item When $\beta \gtrsim \frac{\log n}{n^2p^2s^4}$, $\beta \gtrsim \sqrt{\frac{\log n}{ns^4}}$, and $np^2 \le \frac{1}{\log n}$, we have from \prettyref{eq:lmin} that $l_{\min} \geq c\cdot\beta n^2  p^2 s^4$ for some constant $c$. In other words, the true pair should have sufficiently many $2$-hop witnesses from the correct seeds. 
    \item When $np^2 \le \frac{1}{\log n}$ and $nps^2 \ge 128 \log n$, we have from \prettyref{eq:xmax} that $x_{\max} \lesssim \beta n ps^2 \log n \le \frac{c}{3}\beta n^2  p^2 s^4 $, ensuring that the fake pairs have fewer $2$-hop witnesses from the correct seeds than the true pairs.
    \item For the $2$-hop witnesses from the incorrect seeds, we have from \prettyref{eq:mmin} and \prettyref{eq:ymax} that
    \begin{align*}
        m_{\min} & \approx n (1-\beta) a_u b_u p^2 s^2 - \Delta \\
        y_{\max} & \approx n (1-\beta) a_u b_v p^2s^2 + \Delta, 
    \end{align*}
    where $\Delta = O(\sqrt{n^3 p^4 s^4 \log n}) + O(\log n)$ captures the statistical deviation.
    \begin{itemize}
        \item When $\beta \gtrsim \frac{\log n}{n^2p^2s^4}$ and $\beta \gtrsim \sqrt{\frac{\log n}{ns^4}}$, we have
        $\Delta \le \frac{c}{3} \beta n^2p^2s^4$. %ensuring that the statistical fluctuation of the $2$-hop witnesses from the incorrect seeds is dominated by the $2$-hop witnesses from the correct seeds.
        \item When $\beta \gtrsim \sqrt{ \frac{np^3\log n}{s}}$, in view of $a_u \lesssim nps$ and $b_v-b_u \lesssim \sqrt{nps\log n}$ (the latter one is due to the fluctuation of the $1$-hop neighborhood sizes), we have that         \begin{equation}\label{eq:fluctuation-1-hop}
            n (1-\beta) a_u b_v p^2s^2 -n(1-\beta) a_u b_u p^2 s^2 
         \lesssim n (1-\beta) nps \sqrt{nps\log n} p^2 s^2 \le \frac{c}{3}\beta n^2p^2s^4.     \end{equation}
        %ensuring that 
    \end{itemize}
    The above two claims together ensure that $y_{\max}-m_{\min}$, i.e., the difference between the true pairs and fake pairs of the $2$-hop witnesses from the incorrect seeds, is dominated by $\frac{2c}{3}\beta n^2 p^4 s^4.$
    %the $2$-hop witnesses from the correct seeds for the true pairs. 
    \item Finally, when $np^2 \le \frac{1}{\log n}$, 
    %$z_{\max} \lesssim n^2 p^3 s^3 \lesssim \Delta$. Hence, 
    $2z_{\max} + \psi_{\max}+28\log n \lesssim \Delta$ and hence is negligible.
\end{itemize}
In sum, if condition \prettyref{eq:oldcriteria} holds, then with high probability \prettyref{eq:2_hop_cond_weak} is satisfied and thus the $2$-hop algorithm exactly recovers the
true vertex mapping $\pi^*$.
 
\subsubsection{Derivation of  the Tight Condition \prettyref{eq:conditionofbeta2}}\label{sec:derivation-tight-condition}
Unfortunately, condition \prettyref{eq:oldcriteria} does not completely coincide with the desired  condition \prettyref{eq:conditionofbeta2}. This is because the criteria \prettyref{eq:2_hop_cond_weak} that we used for GMWM to succeed is too strict. 
 %because we have imposed a very strict criteria for GMWM to succeed, which requires the \emph{minimum} number of 2-hop witnesses among true pairs to be greater than the \emph{maximum} number of 2-hop witnesses among fake pairs. 
 Indeed, the GMWM algorithm may succeed even when \prettyref{eq:2_hop_cond_weak} does not hold. For example, consider the case in \prettyref{fig:Newcriteria} when $a_u$ and $b_u$ are both small, while $b_v$ is large . Then, $W_2(u,v)$ may be larger than $W_2(u,u)$ and hence  \prettyref{eq:conditionofbeta2} is not satisfied. However, since $N^{G_1}(v)$ and $N^{G_2}(v)$ are expected to overlap significantly, when $b_v$ is large, $a_v$ is also likely to be large. Hence $W_2(v,v)$ is likely to be larger than $W_2(u,v)$. Thus, GMWM will still select the true pair $(v,v)$ and eliminate the fake pair $(u,v)$. From the above example, we can see that, for the 2-hop algorithm to succeed, it is sufficient to satisfy the following new criteria: 
 \begin{align}
 W_2(u,v) < W_2 (u, u) \text{ or } W_2(u,v) < W_2(v,v), \quad \forall u \neq v. \label{eq:2_hop_cond_strong}
\end{align}
 %that each fake pair, $(u,v)$, has fewer 2-hop witnesses than either the true pair $(u,u)$ or the true pair $(v,v)$.
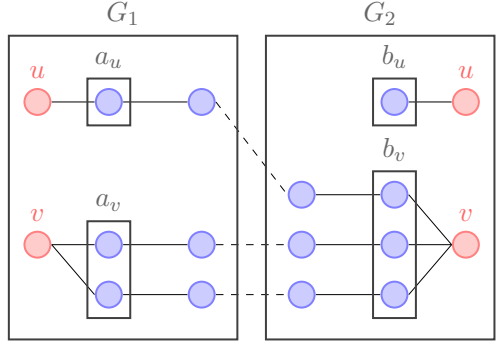
\begin{figure}[h]
\centering
\resizebox{0.4\textwidth}{!}{
\begin{tikzpicture} [
    node distance = 1.3cm, on grid, > = stealth, bend angle = 45, auto,
    vertex/.style = {circle, draw = red!50, fill = red!20, thick, minimum size = 3mm},
    seed/.style = {circle, draw = blue!50, fill = blue!20, thick, minimum size = 3mm},
     graph/.style = {rectangle, draw = black!75,thick}
    ]
\node (u) at (-3,1) [vertex,label={[red!60]above:$u$}]{};
\node (piu) at (3,1) [vertex,label={[red!60]above:$u$}]{};
\node (v) at (-3,-1) [vertex,label={[red!60]above:$v$}]{};
\node (piv) at (3,-1) [vertex,label={[red!60]above:$v$}]{};
\node (n2) at (-2,1) [seed]{};
\node (n3) at (-2,-1.7) [seed]{};
\node (n4) at (-2,-1) [seed]{};
\node (n1) at (2,1) [seed]{};
\node (n2) at (2,-0.3) [seed]{};
\node (n3) at (2,-1) [seed]{};
\node (n4) at (2,-1.7) [seed]{};
\node (s1) at (-0.7,1) [seed]{};
\node (s1) at (-0.7,-1.7) [seed]{};
\node (s1) at (-0.7,-1) [seed]{};
\node (s1) at (0.7,-0.3) [seed]{};
\node (s1) at (0.7,-1) [seed]{};
\node (s1) at (0.7,-1.7) [seed]{};
\draw [dashed] (-0.5, 1) -- (0.5, -0.3);
\draw [dashed] (-0.5, -1) -- (0.5, -1);
\draw [dashed] (-0.5, -1.7) -- (0.5, -1.7);
\draw  (-1.8, 1) -- (-0.9, 1);
\draw  (-1.8, -1) -- (-0.9, -1);
\draw  (-1.8, -1.7) -- (-0.9, -1.7);
\draw  (1.8, -0.3) -- (0.9, -0.3);
\draw  (1.8, -1) -- (0.9, -1);
\draw  (1.8, -1.7) -- (0.9, -1.7);
\draw  (-2.8, 1) -- (-2.2, 1);
\draw  (-2.8, -1) -- (-2.2, -1.7);
\draw  (-2.8, -1) -- (-2.2, -1);
\draw  (2.8, 1) -- (2.2, 1);
\draw  (2.8, -1) -- (2.2, -0.3);
\draw  (2.8, -1) -- (2.2, -1);
\draw  (2.8, -1) -- (2.2, -1.7);
\node (au) at (-2,1) [graph,text height = 0.4cm, minimum width = 0.6cm,label={[black!60]above:$a_u$}]{};
\node (av) at (-2,-1.35) [graph,text height = 1.1cm, minimum width = 0.6cm,label={[black!60]above:$a_v$}]{};
\node (bu) at (2,1) [graph,text height = 0.4cm, minimum width = 0.6cm,label={[black!60]above:$b_u$}]{};
\node (bv) at (2,-1) [graph,text height = 1.8cm, minimum width = 0.6cm,label={[black!60]above:$b_v$}]{};
\node (G1) at (-1.8,-0.2) [graph,text height = 4cm, minimum width = 3.2cm,label={[black!60]above:$G_1$}]{};
\node (G2) at (1.8,-0.2) [graph,text height = 4cm, minimum width = 3.2cm,label={[black!60]above:$G_2$}]{};
\end{tikzpicture}}
\caption{The 2-hop algorithm with GMWM still selects the true pair $(v,v)$ and eliminate the fake pair $(u,v)$ when $W_2(u,v)>W_2(u,u)$ but $W_2(v,v)>W_2(u,v)$.}
\label{fig:Newcriteria}
\end{figure}
Next, we show that under condition \prettyref{eq:conditionofbeta2}, 
with high probability the new criteria \prettyref{eq:2_hop_cond_strong}
is satisfied and hence the 2-hop algorithm succeeds.
%\begin{equation}
%\begin{aligned}
%&n^2\beta p^2s^4+n(1-\beta)n^2p^4s^4-n^3p^4s^4\geq\sqrt{n^2\beta p^2s^4}+\sqrt{n^3p^4s^4}\\
%\Leftrightarrow&n^2\beta p^2\left(1-np^2\right)s^4\geq\sqrt{n^2\beta p^2s^4}+\sqrt{n^3p^4s^4}\\
%\Leftarrow&n^2\beta p^2s^4\geq2\sqrt{n^2\beta p^2s^4}+2\sqrt{n^3p^4s^4}\\
%\Leftarrow&n^2\beta p^2s^4\geq4\sqrt{n^2\beta p^2s^4}\ and\ n^2\beta p^2s^4\geq 4\sqrt{n^3p^4s^4}\label{eq:meanstd2hop}
%\end{aligned}
%\end{equation}
Since $N^{G_1}(u)$ and $N^{G_2}(u)$ are both generated by sampling with probability $s$ from $N^{G_0}(u)$ in the parent graph $G_0$, we have $a_u,b_u\sim\Binom(d_u,s)$. Similarly, $a_v,b_v\sim\Binom(d_v,s)$. Therefore, if $d_u\leq d_v$, 
%using the Bernstein's inequality to control the  binomial fluctuation of $a_u$ and $a_v$,
%we can show $a_v-a_u$ is no larger than
%$O(\sqrt{nps(1-s)\log n})+O(\log n)$ with high probability. 
we expect  $a_u-a_v$ not to be too large. 
If instead $d_u> d_v$,
we expect $b_v-b_u$ not to be too large. More precisely, we have the following lemma, with
the proof deferred to Appendix \ref{sec:proofTuv}.
%using Bernstein’s inequality, $a_u$
% is upper bounded by $d_us+ O(\sqrt{d_us(1-s)\log n})+ O(\log n)$ and $a_v$ is lower bounded by $d_vs- O(\sqrt{d_vs(1-s)\log n})- O(\log n)$ with high probability. Since $d_u,d_v\approx np$, $a_u-a_v$ is upper bounded by $O(\sqrt{nps(1-s)\log n})+O(\log n)$ with high probability. If  $d_u> d_v$, the we can bound $b_u-b_v$ analogously.  
 %we can show that $b_v-b_u$ is upper bounded by $O(\sqrt{nps(1-s)\log n})+O(\log n)$ with high probability analogously using Bernstein’s inequality.
%Building upon this observation, we can see that for any two vertices $u\neq v$
%for any two vertices $u\neq v$ in graph $G_1$, we can derive that
% \begin{align*}
% &\left(\abs{N^{G_1}(u)}-\abs{N^{G_1}(v)}\right)+\left(\abs{N^{G_2}(v)}-\abs{N^{G_2}(\pi(u))}\right)\\
% =&\left(\abs{N^{G_1}(u)}-\abs{N^{G_2}(\pi(u))}\right)+\left(\abs{N^{G_2}(v)}-\abs{N^{G_1}(v)}\right)
% \leq 4\sqrt{nps(1-s)}.
% \end{align*}   
%This inequality implies that 
%Consequently,  
%for any two vertices $u\neq v$,
%either $a_u-a_v$ or $b_v-b_u$ is no larger than $O(\sqrt{nps(1-s)\log n})+O(\log n)$ with high probability. 
%\nbr{Here the explanation does not really match Lemma 6.Lemma 6 has extra log factors. LY: I add more precise exposition}

\begin{lemma}\label{lmm:Tuv}
Given any $u ,v\in [n]$, let $T_{uv}$ denote the event:
\begin{align}
T_{uv}= &\left\{a_u-a_v\leq\tau\right\}\cup\left\{ b_v-b_u\leq \tau\right\},\label{eq:tuv}
\end{align}
where 
\begin{align}
\tau \triangleq 2\sqrt{10nps(1-s)\log n}+5\log n. \label{eq:def_tau}
\end{align} 
If $n$ is sufficiently large %\nbr{Are you sure you also need this assumption?}
%\nbr{Can you also check other lemmas? I am not sure if you always need these assumptions.}
and $nps^2\ge 128\log n$, then 
\begin{align*}
\mathbb{P}(T_{uv})\geq1-n^{-\frac{7}{2}}.
\end{align*}
\end{lemma}

Next we show the new criteria \prettyref{eq:2_hop_cond_strong}
is satisfied by separately considering two cases: $b_v-b_u\leq \tau$ and $a_u-a_v\leq \tau$. We first consider the case $b_v-b_u\leq \tau$. When $\beta \gtrsim \sqrt{ \frac{np^3(1-s)\log n}{s}}$, $\beta \gtrsim \sqrt{\frac{\log n}{ns^4}}$ and $np^2\leq \frac{1}{\log n}$, in view of $a_u \lesssim nps$, we can get a tighter upper bound to the left hand side of \prettyref{eq:fluctuation-1-hop}: 
\begin{align}\label{eq:fluctation-1-hop-new}
 &n (1-\beta) a_u b_v p^2s^2 -n(1-\beta) a_u b_u p^2 s^2\nonumber \\
         \lesssim &n (1-\beta) nps \left(\sqrt{nps(1-s)\log n}+\log n \right) p^2 s^2 \nonumber\\
         \overset{(a)}{\le} & n^2p^3s^3\sqrt{nps(1-s)\log n}+n^2p^2s^3\sqrt{\frac{\log n}{n}} \overset{(b)}{\le} \frac{c}{3}\beta n^2p^2s^4,
\end{align}
where the inequality $(a)$ holds due to $p\le \sqrt{1/(n\log n)}$; the inequality $(b)$ is guaranteed by the last two terms in  condition \prettyref{eq:conditionofbeta2}.

To be more precise, the following lemma combined with
\prettyref{lmm:W2uu} and \prettyref{lmm:W2uv} ensures that 
$W_2(u,u) > W_2(u,v)$ with high probability. 
%it suffices to ensure that the lower bound to $W_2(u,u)$ is no smaller than the upper bound to $W_2(u,v)$.
\begin{lemma}\label{lmm:wmingwmax}
Given any two vertices $u,v\in [n]$ with $u\neq v$, if $R_{uv}$ occurs,  $b_v- b_u\leq \tau$, $nps^2\geq 128\log n$, $\nps$, and condition \prettyref{eq:conditionofbeta2} holds, %\ie, $\beta\geq\max\left\{900\sqrt{\frac{np^3(1-s)\log n}{s}},600\sqrt{\frac{\log n}{ns^4}},\frac{900\log n}{n^2p^2s^4}\right\}$,
then for sufficiently large $n$,
\begin{align*}
l_{\min}+m_{\min}\geq x_{\max}+y_{\max}+2z_{\max}+\psi_{\max}+28\log n.
\end{align*}
where  $l_{\min}$, $m_{\min}$, $x_{\max}$, $y_{\max}$, $z_{\max}$ and $\psi_{\max}$ are given in \prettyref{lmm:W2uu} and \prettyref{lmm:W2uv}.
\end{lemma} 
Please refer to Appendix \ref{sec:proofwmingwmax} for details.

For the other case that $a_u-a_v\leq \tau$,  we instead bound $W_2(v,v)$ from below analogous to \prettyref{lmm:W2uu}, and prove that 
$W_2(v,v)>W_2(u,v)$ with high probability
analogous to \prettyref{lmm:wmingwmax}.
%prove that the lower bound is no smaller than the upper bound to $W_2(u,v)$ analogous to \prettyref{lmm:wmingwmax}.

Thus, by combining the two cases and applying union bound, we ensure that
with high probability
the new criteria \prettyref{eq:2_hop_cond_strong} is satisfied and hence the GMWM outputs
the true matching,
completing the proof of \prettyref{thm:thm2hop}. Please refer to Appendix \ref{sec:proof2hop} for details.  Note that, when $1-s=o(1)$, condition \prettyref{eq:conditionofbeta2} given by the new criteria \prettyref{eq:2_hop_cond_strong} requires a smaller $\beta$ than \prettyref{eq:oldcriteria} given by the old criteria \prettyref{eq:2_hop_cond_weak}.

\section{Numerical experiments}\label{sec:experiment}
In this section, we present numerical studies,
comparing the performance of our $2$-hop algorithm to
the 1-hop algorithm \cite{lubars2018correcting} and the NoisySeeds algorithm \cite{10.14778/2794367.2794371}, which are the state-of-the-art for graph matching with imperfect seeds.
Additional numerical studies to verify our theoretical results are deferred to Appendix \ref{sec:exp}. In all our experiments except for the last one on the computer vision dataset in \prettyref{sec:exp-3d}, 
the initial seeded mappings are constructed in the same way as our model given in~\prettyref{sec:model}, i.e., the initial mappings are uniformly chosen at random with a given number of correctly matched pairs. In contrast, in the computer vision experiment in \prettyref{sec:exp-3d}, the initial seeded mapping is from the output of a seedless matching algorithm. The computational environment is  Matlab R2017a on a standard PC with
$2.4$ GHz CPU and $8$ GB RAM.
Our code has been released on GitHub at \url{https://github.com/Leron33/Graph-matching}.

%Intel(R) Core(TM) i5-4210U CPU @2.40GHz, 8.00 GB RAM and Windows 10 Pro Operation System.

\subsection{Performance Comparison with Synthetic Data}

For our experiments on synthetic data, we generate $G_1$, $G_2$ and $\pi^*$  according to the correlated \ER model. We calculate the accuracy rate as the median of the proportion of vertices that are correctly matched, taken over 10 independent simulations. In \prettyref{fig:Compare1}, we fix $n=10000$ and $s=0.9$, and plot the accuracy rates for $p=n^{-\frac{3}{4}}$ and $p=n^{-\frac{6}{7}}$. 
We observe that the 2-hop algorithm significantly outperforms the 1-hop algorithm. For the NoisySeeds algorithm, its performance is sensitive to the threshold value $r$.
%Hence, we compare the performance curves for $r=2,3,4$. 
%We observe that for both cases, 
The 2-hop algorithm performs either comparably or better than the NoisySeeds algorithm even with the best choice of $r$. 
Note that it is \emph{a priori} unclear how to choose the best value of $r$ for the NoisySeeds algorithm, while our 2-hop algorithm does not need to tune any parameters. 
Computationally, when we match two graphs of size $10000$ with $p=n^{-\frac{6}{7}}$ and $\beta=0.5$, %\nb{can you add the details on PC configuration at the beginning? Also, it seems that we haven't incorporated our responses to edits. Can you check all our previous responses, and incorport them accordingly?}, 
the average running times of the 1-hop algorithm, 2-hop algorithm, and NoisySeeds algorithm are about 52s, 86s and 101s, respectively. 
%The 1-hop  algorithm requires the least time but performs worst. The 2-hop algorithm and NoisySeeds algorithm requires similar running time. 
Similar to the  NoisySeeds algorithm, we can modify GMWM for parallel implementation to make our $2$-hop algorithm even more scalable. Please refer to \prettyref{app:scalable} 
%of Supplementary Material 
for details.

\begin{figure}[h]
\centering
\subfigure[$p=n^{-\frac{3}{4}}$.]{
\includegraphics[scale=0.37]{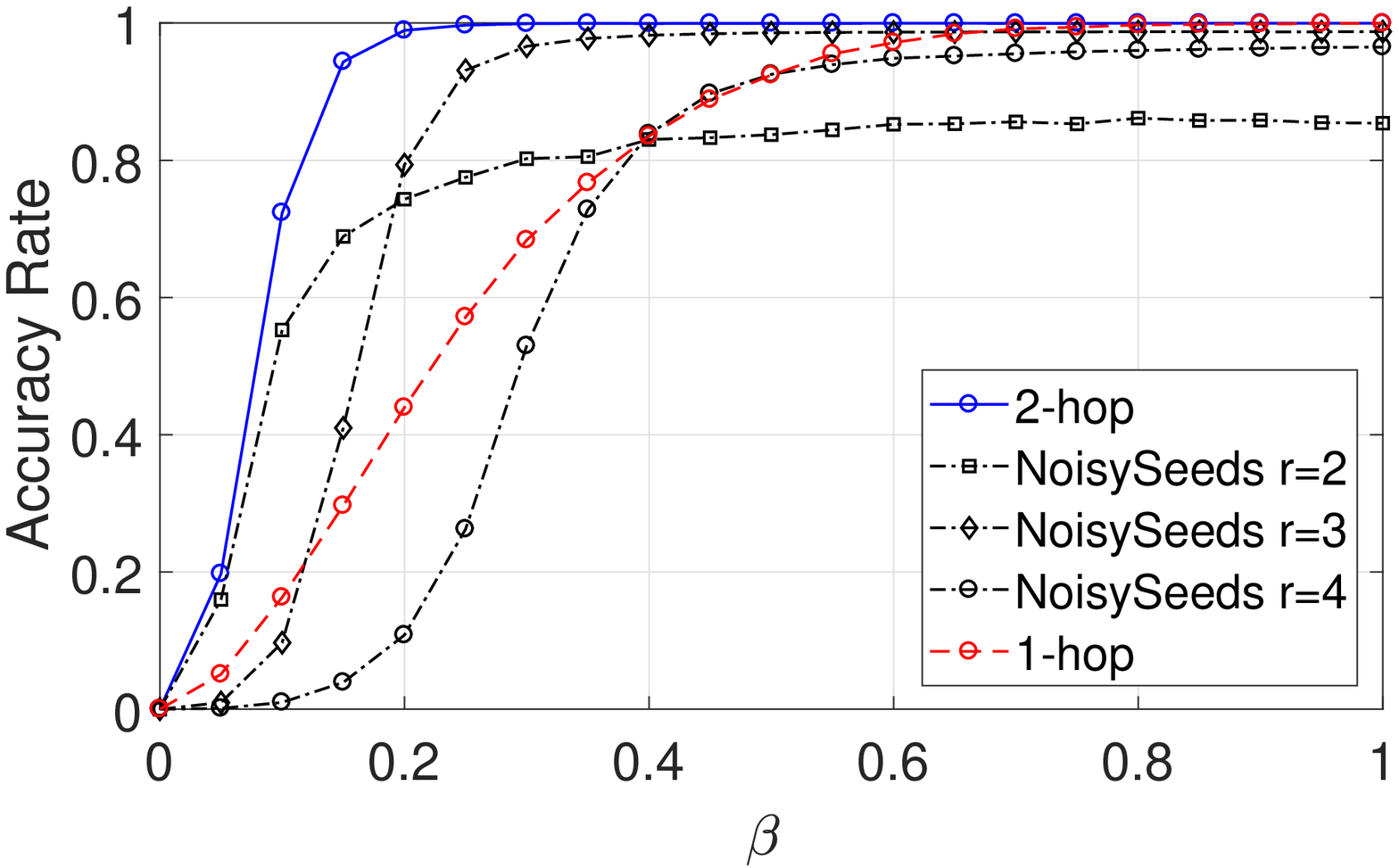}\label{fig:CompareP1}}
\subfigure[$p=n^{-\frac{6}{7}}$.]{
 \includegraphics[scale=0.37]{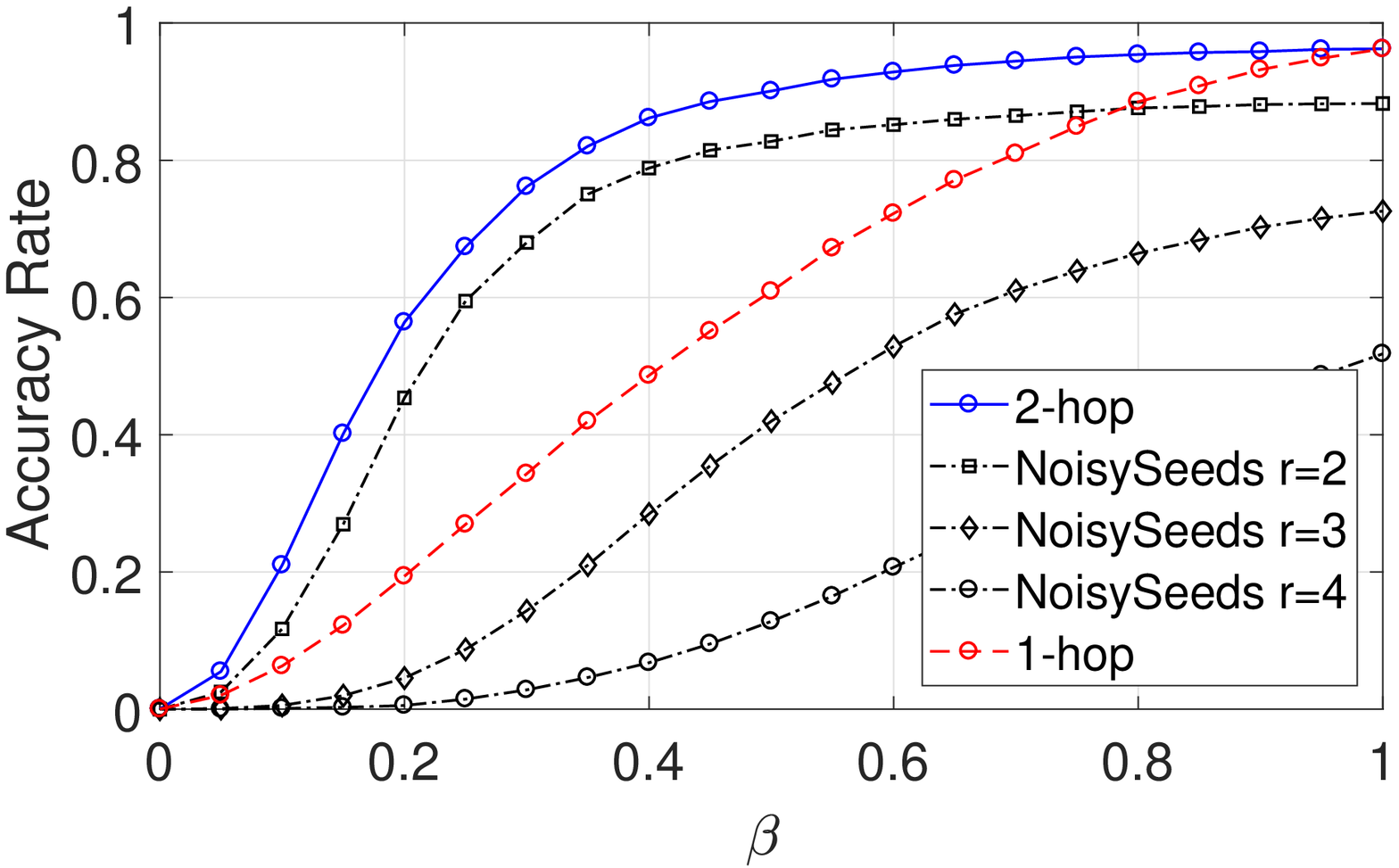}\label{fig:CompareP2}}
\caption{Performance comparison of the 1-hop algorithm, 2-hop algorithm and NoisySeeds algorithm with $p=n^{-\frac{3}{4}}$ and $p=n^{-\frac{6}{7}}$. Fix $n=10000$ and $s=0.9$.}
\label{fig:Compare1}
\end{figure}

In \cite{lubars2018correcting}, the authors suggest iteratively applying the 1-hop algorithm to further boost its accuracy. Thus, we further iteratively apply the 1-hop algorithm, 2-hop algorithm and NoisySeeds algorithm and compare their performance.
By using the matching output of the previous iteration as the new partially-correct seeds for the next iteration, we run the three algorithms with a given number of iterations $L=0,1,2$. In \prettyref{fig:CompareIteration}, we consider the same setup as in \prettyref{fig:Compare1}. We fix $n=10000$ and $s=0.9$, and plot the accuracy rates for $p=n^{-\frac{3}{4}}$ and $p=n^{-\frac{6}{7}}$. For the NoisySeeds algorithm, we choose the threshold $r=3$ for $p=n^{-\frac{3}{4}}$ and $r=2$ for $p=n^{-\frac{6}{7}}$. We observe that iteratively applying these algorithms boost their performance and the 2-hop algorithm still performs the best among the three algorithms when the number of iterations is the same. In particular, when $p=n^{-\frac{6}{7}}$, while the matching accuracy of the 2-hop with multiple iterations gets close to $1$,  the matching accuracy of NoisySeeds saturates at $0.8\sim0.9$. This is because about $10\%$ true pairs  have only one common 1-hop neighbor and thus cannot be correctly matched by the NoisySeeds algorithm with $r=2$.

\begin{figure}[h]
\centering
\subfigure[$p=n^{-\frac{3}{4}}$.]{
\includegraphics[scale=0.37]{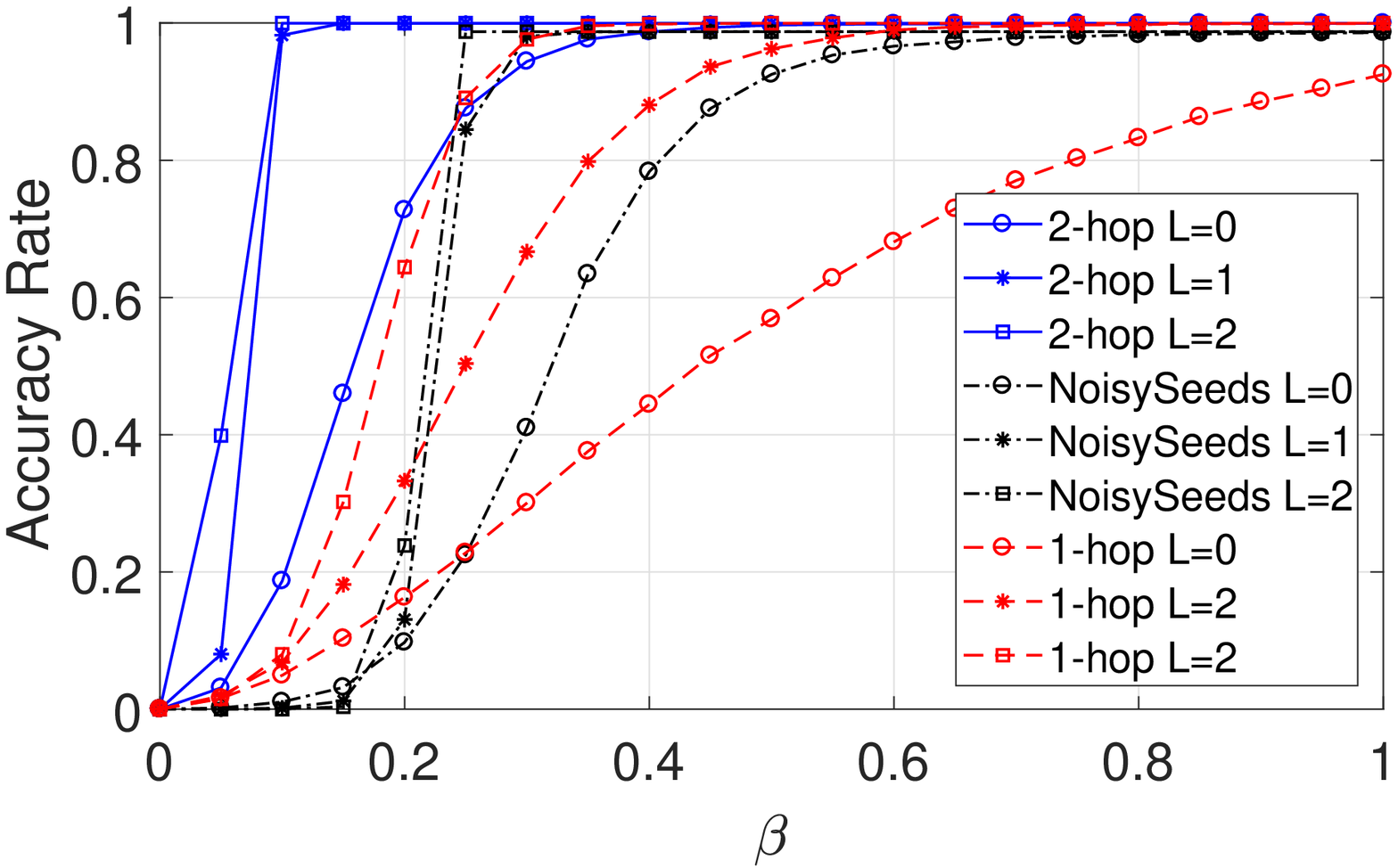}\label{fig:CompareI1}}
\subfigure[$p=n^{-\frac{6}{7}}$.]{
 \includegraphics[scale=0.37]{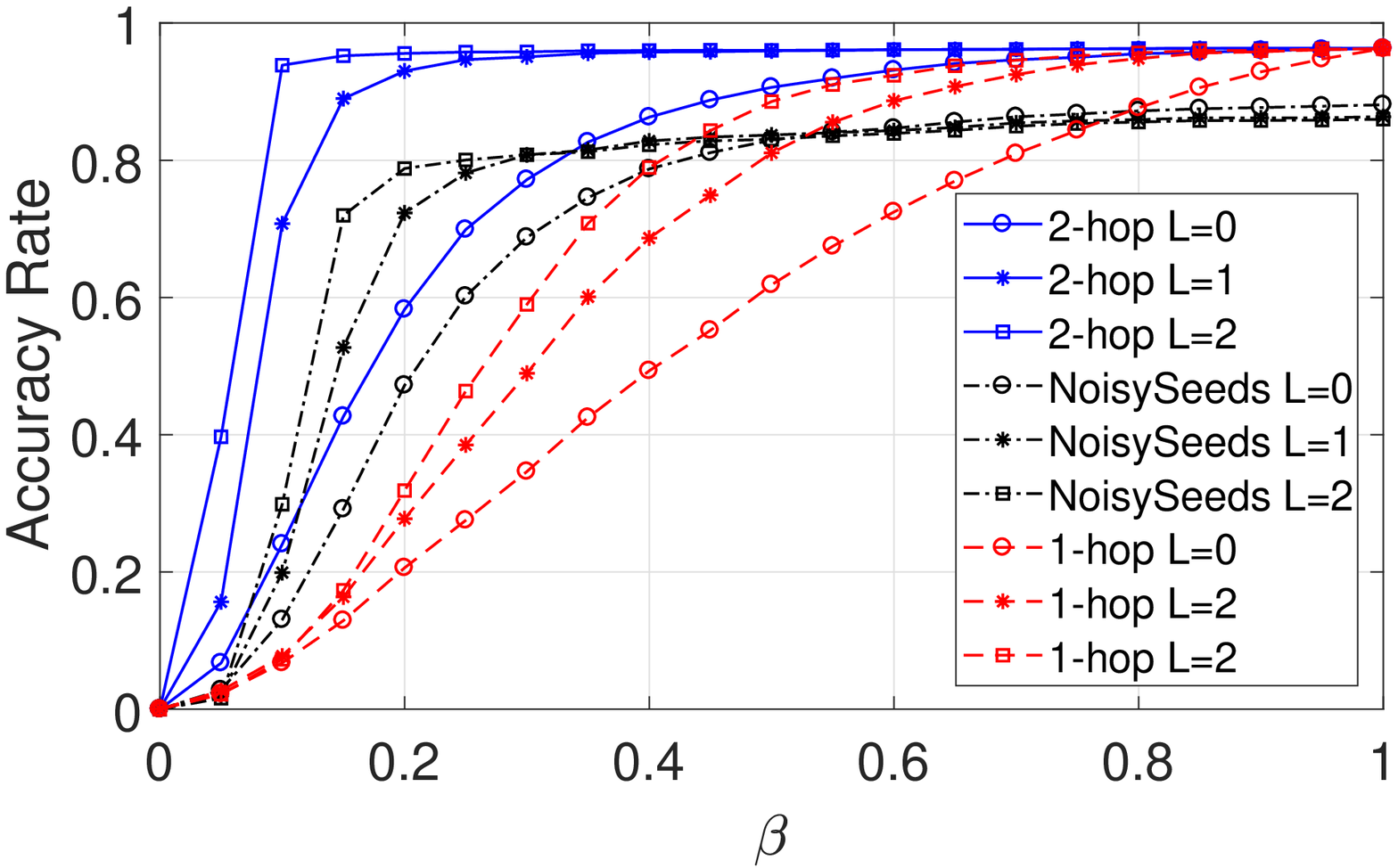}\label{fig:CompareI2}}
\caption{Performance comparison of the 1-hop algorithm, 2-hop algorithm and NoisySeeds algorithm applied iteratively with $p=n^{-\frac{3}{4}}$ and $p=n^{-\frac{6}{7}}$. Fix $n=10000$ and $s=0.9$.}
\label{fig:CompareIteration}
\end{figure}

\subsection{Performance  Comparison with Real Data}\label{sec:experiment-real}

In this section, we will show that our 2-hop algorithm also performs well on real-world graphs. Further, departing from our simulation with synthetic data where the two correlated graphs have the same number of vertices, we will evaluate the performance of the algorithms when the two correlated graphs have a different number of vertices. In \prettyref{sec:exp-facebook}, we consider de-anonymizing social networks, which is a popular application of graph matching. %\nbr{references? or do we want to say this is commonly evaluated? If so, why not we use some existing examples as benchmarks? }. 
In this example, the two correlated graphs are generated by sampling a Facebook friendship network. In \prettyref{sec:exp-auto}, we match networks of Autonomous Systems in which the correlated graphs
are povided by the real dataset. In both \prettyref{sec:exp-facebook} and \prettyref{sec:exp-auto}, the initial seeds are chosen uniformly randomly. In \prettyref{sec:exp-3d}, we match deformable 3D shapes, where not only the correlated graphs are provided by the real dataset, but also the initial seeds are generated by a seedless graph match algorithm.

\subsubsection{Facebook Friendship Networks}\label{sec:exp-facebook}
We use a Facebook friendship network of 11621 students and staffs from Standford university provided in \cite{Traud_2012} as the parent graph $G_0$. There are 1136660 edges in $G_0$. The Facebook social network has an approximate power-law degree distribution with $p(d)\sim d^{-1}$ with average degree about 100. 
To obtain two correlated subgraphs $G_1$ and $G_2$ of different sizes, 
we independently sample each edge of $G_0$ twice with probability $s=0.9$ and sample each vertex of $G_0$ twice with probability $\alpha=0.8$. 
Then, we relabel the vertices in $G_2$ according to a random permutation $\pi^*:[n_2] \to [n_2]$,
where $n_2$ is the number of nodes in $G_2$. Let $m$ denote the number of common vertices that appear in both $G_1$ and $G_2$. The initial seed mapping is constructed by uniformly and randomly choosing a mapping $\pi:[m] \to [m]$ between the common vertices of the two subgraphs such that $\beta$ fraction of vertices are correctly matched, \ie,
$\pi(u)=\pi^*(u)$ for exactly $\beta m$ common vertices.
We treat $G_1$ as the public network and $G_2$ as the private network, and the goal is to de-anonymize the node identities in $G_2$ by matching $G_1$ and $G_2$. 
We show the performance of the 1-hop algorithm, 2-hop algorithm, and NoisySeeds algorithm in \prettyref{fig:Compare3}. %\nbr{need to update this, because now AS networks appear later.} 
We choose the threshold $r=5,10,15$ for the NoisySeeds algorithm to search for the best value of $r$. We observe that our proposed 2-hop algorithm significantly outperforms the 1-hop algorithm and NoisySeeds algorithm.
Note that the matching accuracy is 
saturated at around $80\%$, because 
there are about $15\%$ common vertices that are isolated in the intersection graph 
$G_1 \wedge G_2$ and thus
can not be correctly matched.
%\nbr{Note that about $15\%$ common vertices become isolated in the intersection graph 
%$G_1 \wedge G_2$ after the sampling process and thus can not be matched by the three algorithms.}
%\nb{Here the purpose of this explanation is not clear. I guess you want to explain why the matching performance of $2$-hop saturates at $0.8$, right? If so, I think you need to check whether there are indeed about $20\%$ of common vertices are isolated in the intersection graph 
%$G_1 \wedge G_2$. In other words, I want a more precise explanation.}
Due to the power-law degree variation, the number of witnesses for some fake pairs could be larger than the threshold $r$. Thus, even the NoisySeeds algorithm with the best value of $r$ does not perform well.

\begin{figure}[h]
\centering
\includegraphics[width=0.6\columnwidth]{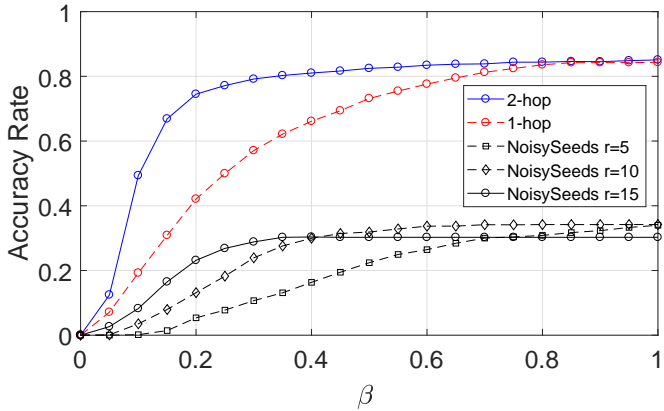}
\caption{Performance comparison of 1-hop algorithm, 2-hop algorithm and NoisySeeds algorithm applied to the Facebook networks.}
\label{fig:Compare3}
\end{figure}

\subsubsection{Autonomous Systems Networks}\label{sec:exp-auto}
Following~\cite{pmlr-v119-fan20a}, we use the Autonomous Systems (AS) dataset from \cite{snapnets} to test the graph matching performance on real graphs. The dataset consists of 9 graphs of Autonomous Systems peering information inferred from Oregon route-views between March 31, 2001, and May 26, 2001. Since some vertices and edges are changed over time, these nine graphs can be viewed as correlated versions of each other. The number of vertices of the 9 graphs ranges from 10,670 to 11,174 and the number of edges from 22,002 to 23,409. The Autonomous Systems networks have an approximate power-law degree distribution with $p(d)\sim d^{-2}$ with average degree about 2. 
%\nb{move the average degree here. Also, is the degree distribution power-law? I remember we address this in one of our response.}
We apply the 1-hop algorithm, 2-hop algorithm and NoisySeeds algorithm (with the best-performing threshold $r=2$) to match each graph to that on March 31, with vertices randomly permuted. %Although there is a generalized percolation-based algorithm, DDM, in \cite{10.1109/TNET.2016.2553843}, to match power-law graphs, it only considers the initial seeds to be all correct. Since the influence of incorrect seeds with high degree could not be ignored, the DDM algorithm is not suitable in the case of partially-correct seed. Matching power-law graphs with partially-correct seeds is an interesting  problem to be studied for future work.
To obtain the initial seed mapping, we uniformly and randomly choose the mapping between the common vertices among the two given graphs such that $\beta$ fraction of vertices are correctly matched.

The performance comparison of the three algorithms is plotted in \prettyref{fig:Compare2} for $\beta=0.3,0.6,0.9$. We observe that our proposed 2-hop algorithm significantly outperforms the 1-hop algorithm and NoisySeeds Algorithm. The NoisySeeds algorithm does not perform well due to its thresholding component: There exists high degree variation in these real graphs and thus a significant fraction of true pairs have only 1 witness, which falls below even the smallest threshold $r=2$. Note that the accuracy rates for all algorithms decay in time because over time the graphs become less correlated with the initial one on March 31. Computationally, when we match two real graphs with $\beta=0.6$, the average running time of the 1-hop algorithm, 2-hop algorithm, and NoisySeeds algorithm is about 46s, 73s and 91s, respectively.

%a large number of nodes are isolated in the underlying graph, making those nodes impossible to match by any algorithm. 
\begin{figure}[h]
\centering
\includegraphics[width=0.6\columnwidth]{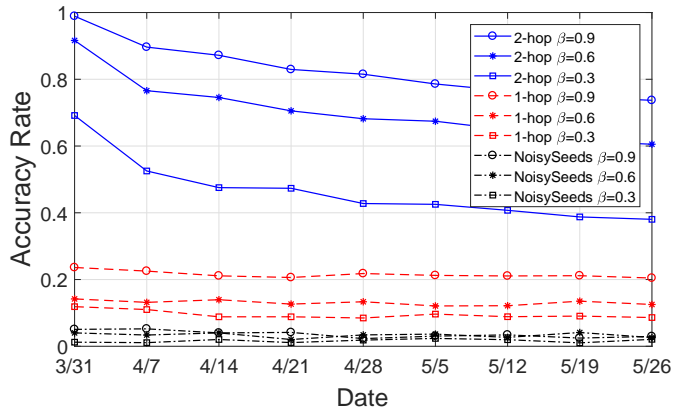}
\caption{Performance comparison of 1-hop algorithm, 2-hop algorithm and NoisySeeds algorithm applied to the Autonomous Systems  graphs.}
\label{fig:Compare2}
\end{figure}

\subsubsection{Computer Vision Dataset}\label{sec:exp-3d}
In this experiment, we use the output of seedless graph matching algorithms as partially correct seeds, and test the performance of 1-hop, 2-hop, and NoisySeeds 
in correcting the initial matching errors. 
We focus on the application of deformable shape matching. 
Matching 3D deformable shapes is a fundamental and ubiquitous problem in
computer vision with numerous applications such as object recognition, and has been extensively studied for decades (see~\cite{van2011survey} and~\cite{sahilliouglu2020recent} for
surveys). At a high-level, each 3D shape is represented as a mesh graph. 
For two 3D shapes corresponding to the same object but with different poses, their mesh graphs are approximately isomorphic. However, the exact vertex mapping is not knowm. Thus, the goal of deformable shape matching is to retrieve  the correct vertex correspondence by matching the two mesh graphs.

The previous work \cite{pmlr-v119-fan20a}  applied the 1-hop algorithm iteratively to boost the matching accuracy of their seedless graph matching algorithm, GRAMPA (GRAph Matching by Pairwise eigen-Alignments). Their experiment is carried on the SHREC'16 dataset  in \cite{lahner2016shrec}. 
%\nb{The GRAMPA algorithm corrected by the 1-hop algorithm  significantly outperforms the existing methods, as shown in \cite{pmlr-v119-fan20a}.} %\nb{It seems not the best?}\nbr{For myself, need to point out GRAMPA is the best.} 
The SHREC'16 dataset provides 25 deformable 3D shapes (15 for training and 10 for testing) undergoing different topological changes. At the lower resolution, each shape is represented by a triangulated mesh graph consisting of %\nbr{8833-11413} 
around 8K-11K vertices with 3D coordinates and  17K-22K
%\nbr{17626-22766} 
triangular faces, with vertex degrees highly concentrated on $6$.
%\nb{And over $90\%$ vertices have degree between 4 and 8.} \nbr{Thanks. But can you further check whether it is a regular graph with degree $6$?} 
It is demonstrated  in~\cite{pmlr-v119-fan20a}  that
the GRAMPA followed by the iterative 1-hop algorithm achieves  much higher matching accuracy compared to the existing methods tested in~\cite{lahner2016shrec}. 

We also use the SHRED'16 dataset in our experiment. When we match each pair of test shapes, we first apply the GRAMPA algorithm, and then repeatedly use the 1-hop algorithm, 2-hop algorithm and NoisySeeds algorithm with 100 iterations to boost the matching accuracy of the output of the GRAMPA algorithm. \prettyref{fig:Visual3D} provides a visualization of our results, where the matched vertices are colored with the same color. We can see that the 2-hop algorithm corrects most matching errors of the GRAMPA algorithm.

\begin{figure}[h]
\centering
\subfigure[]{
\includegraphics[scale=0.5]{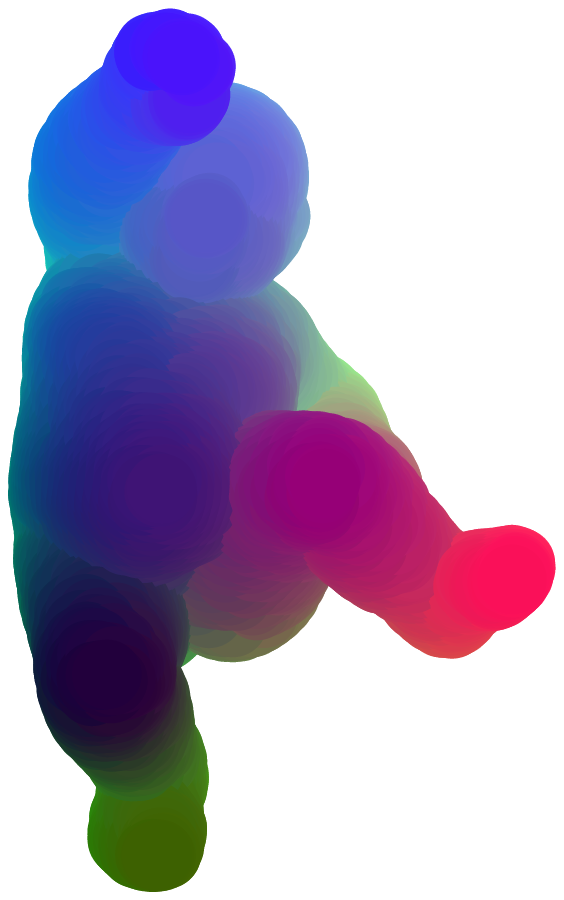}
\label{fig:visual3da}}
\subfigure[]{
\includegraphics[scale=0.5]{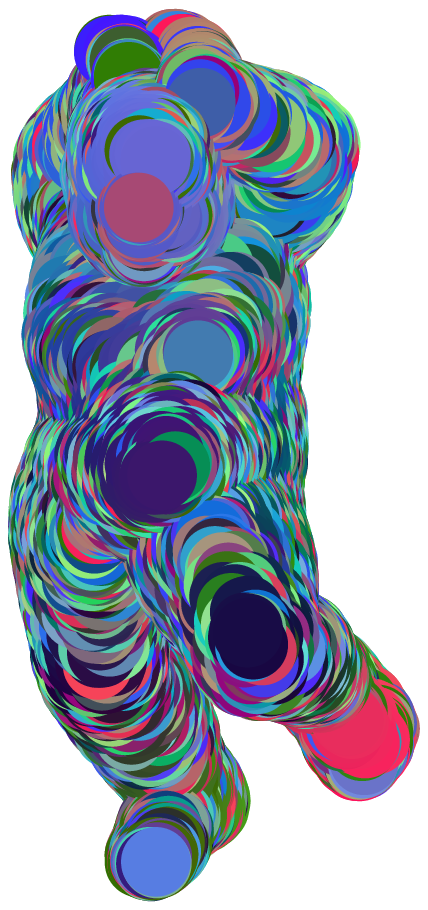}
\label{fig:visual3db}}
\subfigure[]{
\includegraphics[scale=0.5]{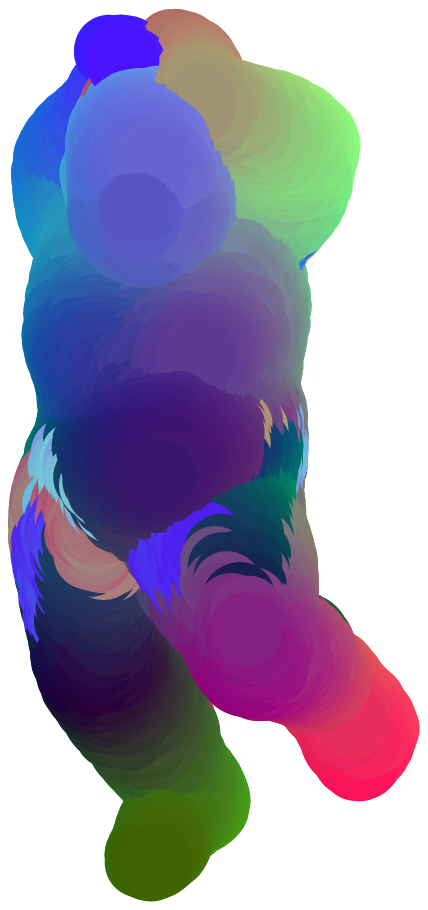}
\label{fig:visual3dc}}
\caption{Visualization of the matching results. \prettyref{fig:visual3da} is the mesh graph of a 3D shape randomly chosen from the SHREC'16 dataset, whose vertices are colored in a gradient. \prettyref{fig:visual3db} and \prettyref{fig:visual3dc} are the same mesh graph with a different pose that needs to be matched with \prettyref{fig:visual3da}, where vertices in \prettyref{fig:visual3db} and \prettyref{fig:visual3dc} are labeled with the same color as that of the vertices in \prettyref{fig:visual3da} matched by the corresponding graph matching algorithms.  \prettyref{fig:visual3db} shows the matching result of the 
GRAMPA algorithm. \prettyref{fig:visual3dc} shows the matching result when we use the output of the GRAMPA algorithm as the initial seeds and apply the 2-hop algorithm iteratively. We can observe that the 2-hop algorithm corrects most matching errors of the GRAMPA algorithm.}
\label{fig:Visual3D}
\end{figure}

We follow the Princeton benchmark protocol in \cite{10.1145/2010324.1964974} to evaluate the matching quality. Assume that a vertex-pair $(i,j)\in \mathcal{M}\times \mathcal{N}$ is matched between shapes $\mathcal{M}$ and $\mathcal{N}$, while the ground-truth correspondence is $(i, j^*)$.
Then the normalized geodesic error of this correspondence at vertex $i$ is defined as $\varepsilon(i) = \frac{d_{\mathcal{N}} (j,j^*)}{\sqrt{\text{area}(\mathcal{N})}}$, where $d_{\mathcal{N}}$ denotes the geodesic distance on $\mathcal{N}$ and $\text{area}(\mathcal{N})$ is the total surface area of $\mathcal{N}$. 
%Finally, we plot the cumulative distribution curve to show the proportion of matches that have an error smaller than a variable threshold in \prettyref{fig:Compare3D}. 
Finally, we plot the cumulative distribution function of $\{\varepsilon(i)\}_{i=1}^n$ in~\prettyref{fig:Compare3D},
where $\text{cdf}(\epsilon)$ is the fraction of vertices $i$ such that $\varepsilon(i)\le \epsilon$.
In particular, $\text{cdf}(0)$ is the fraction of correctly matched vertices in shape $\mathcal{M}$.

In \prettyref{fig:Compare3D}, we  observe that all three algorithms improve the initial matching accuracy of the GRAMPA algorithm, but the performance improvement of our 2-hop algorithm is most substantial. In particular, our 2-hop algorithm increases the fraction of correctly matched vertices to more than $80\%$, while the 1-hop algorithm 
and NoisySeeds (with the best choice threshold $r=2$) only correctly match around $60\%$ and $30\%$ of vertex-pairs, respectively.
%The performance comparison of the three algorithms  demonstrates that our 2-hop algorithm significantly outperforms the 1-hop algorithm and NoisySeeds Algorithm.
%\nbr{Maybe here we can also add a visualization plot, just to give a reader a quick impression on how the 3-D shape matching works. We can choose a best result of our $2$-hop algorithm to visualize: on the left is the matching result of GRAMPA itself, on the right is the matching result of GRAMPA with 2-hop algorithms. I expect on the left, there are lots of matching errros, while on the left there are only a few matching errors. This clearly shows the performance improvement of our 2-hop algorithm.}

\begin{figure}[h]
\centering
\includegraphics[width=0.6\columnwidth]{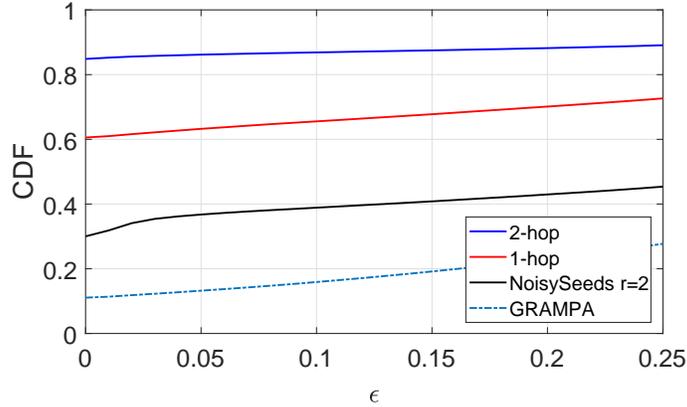}
\caption{Performance comparison of 1-hop algorithm, 2-hop algorithm and NoisySeeds algorithm applied to the SHREC'16 dataset, with initial noisy seeds generated by the GRAMPA algorithm. The higher the curve, the better the algorithm performance.
%\nbr{mention the higher, the beta.}
}
\label{fig:Compare3D}
\end{figure}
%\nbr{The results look great! Do you have a sense of the matching performance of GRAMPA itself? I loosely remember the performance is low. I am thinking maybe we can add it to show the performance boost and remove the curve corresponding to NoisySeeds $r=4.$ }

\section{Conclusion}\label{sec:con}
In this work, we tackle the graph matching problem with partially-correct seeds. Under the correlated \ER model, we first present a sharper characterization of the condition for the 1-hop algorithm to perfectly recover the vertex matching for dense graphs, which requires many fewer correct seeds than  the prior art when graphs are dense.
%is more relaxed than the prior art for dense graphs.
Then, for sparse graphs, by exploiting 2-hop neighbourhoods, we propose an efficient 2-hop algorithm that perfectly recovers the true vertex correspondence with even fewer correct seeds than the 1-hop algorithm in sparse graphs. 
Our performance guarantees for the $1$-hop and $2$-hop algorithm combined together achieve the best-known results across the entire range
of graph sparsity and significantly improve the state-of-the-art. %In particular, the fraction of correct seeds required diminishes to $0$ when the graph size $n$ gets large, showing that the $1$-hop algorithm and our $2$-hop algorithm combined are able to correct a large number of matching errors and significantly boost the matching performance. 
Moreover, our results precisely characterize the graph sparsity at which the $2$-hop algorithm starts to outperform $1$-hop.
This reveals an interesting and delicate trade-off between the \emph{quantity} and the \emph{quality} of witnesses: while the $2$-hop algorithm exploits more seeds as witnesses than $1$-hop, the $2$-hop witnesses are less accurate than the $1$-hop counterparts in distinguishing true pairs from fake pairs when graphs are dense.

%ore effectively for sparse graphs.
% if $\beta\geq \max\left\{900\sqrt{\frac{np^3(1-s)\log n}{s}},600\sqrt{\frac{\log n}{ns^4}}, \frac{900\log n}{n^2p^2s^4}\right\}$, $nps^2\geq 128 \log n$ and $np^2<\frac{1}{135}$.
Experimental results validate our theoretical analysis, demonstrating that our $2$-hop algorithm continues to perform well in real graphs with power-law degree variations and different number of nodes. 
There are many exciting future directions such as %designing algorithms that requires less correct seeds to recover true vertex correspondence in dense graphs
analyzing the performance of $j$-hop algorithms for $j\ge 3$, investigating fundamental limits of seeded graph matching, and studying graph matching under other random graph models beyond \ER random graphs. 
%\section*{Acknowledgment}
%
%L.~Yu and J.~Xu are supported by the NSF Grant IIS-19326
%\input{main.bbl}

\newpage
\begin{appendices}

\section{Numerical  Experiments to Verify The Scalings}\label{sec:exp}
In this section, we conduct numerical studies to verify the scaling results given in \prettyref{thm:thm1hop} and \prettyref{thm:thm2hop}. 
We observe that 
conditions~\prettyref{eq:conditionofbeta1} and~\prettyref{eq:conditionofbeta2} are not only sufficient, but also  close to necessary (differing from the necessary conditions by a constant factor) for the 1-hop and 2-hop algorithms to succeed, respectively.
%For our experiments using synthetic data, 
Throughout, we generate $G_1$, $G_2$ and $\pi^*$ according to the correlated \ER model with fixed sampling probability $s=0.8$, and vary the number of vertices from 2000 to 8000. 
%We calculate the accuracy rate as the median of the proportion of vertices that are correctly matched, taken over 10 independent simulations.

%We first verify that the condition \prettyref{eq:conditionofbeta1} in \prettyref{thm:thm1hop} is both sufficient and close to necessary for the 1-hop algorithm to exactly recover $\pi^*$. We simulate the performance of the 1-hop algorithm for $p=n^{-\frac{1}{3}}$ and $p=n^{-\frac{2}{3}}$. The results are presented in \prettyref{fig:onehop1sub1} and \prettyref{fig:onehop2sub1} as a function of $\beta$. Since \prettyref{thm:thm1hop} predicts that the 1-hop algorithm succeeds in exact recovery with high probability at $\beta\gtrsim \sqrt{\frac{\log n}{n}}$ when $p=\Omega\left(\sqrt{{\log n}/{n}}\right)$ and at $\beta\gtrsim\frac{\log n}{np}$ when $p=O\left( \sqrt{{\log n}/{n}}\right)$, we rescale the x-axis in \prettyref{fig:onehop1sub2} and \prettyref{fig:onehop2sub2} as $\beta/\sqrt{\frac{\log n}{n}}$ for $p=n^{-\frac{1}{3}}$ and $\beta/\left(\frac{\log n}{np}\right)$ for $p=n^{-\frac{2}{3}}$. As we can see in \prettyref{fig:onehop1sub2} and \prettyref{fig:onehop2sub2}, the curves for different $n$ align well with each other, which suggests that condition \prettyref{eq:conditionofbeta1} is both sufficient and close to necessary for the 1-hop algorithm to succeed.

We first simulate the performance of the 1-hop algorithm for $p=n^{-\frac{1}{3}}$ and $p=n^{-\frac{2}{3}}$. The results are presented in \prettyref{fig:onehop1sub1} and \prettyref{fig:onehop2sub1} as a function of $\beta$. \prettyref{thm:thm1hop} predicts that the 1-hop algorithm succeeds in exact recovery when $\beta\gtrsim \sqrt{\log n/n}$ for $p=n^{-\frac{1}{3}}$ and 
when $\beta\gtrsim \frac{\log n}{np}$ for $p=n^{-\frac{2}{3}}$. Thus, we 
rescale the $x$-axis in \prettyref{fig:onehop1sub2} and \prettyref{fig:onehop2sub2} as $\beta/\sqrt{\frac{\log n}{n}}$ and $\beta/\left(\frac{\log n}{np}\right)$, respectively. 
We see that  after rescaling the curves for different $n$ align well with each other, suggesting that condition \prettyref{eq:conditionofbeta1} is both sufficient and close to necessary for the 1-hop algorithm to succeed.

\begin{figure}[H]
\centering
\subfigure[x-axis is $\beta$.]{
\includegraphics[scale=0.37]{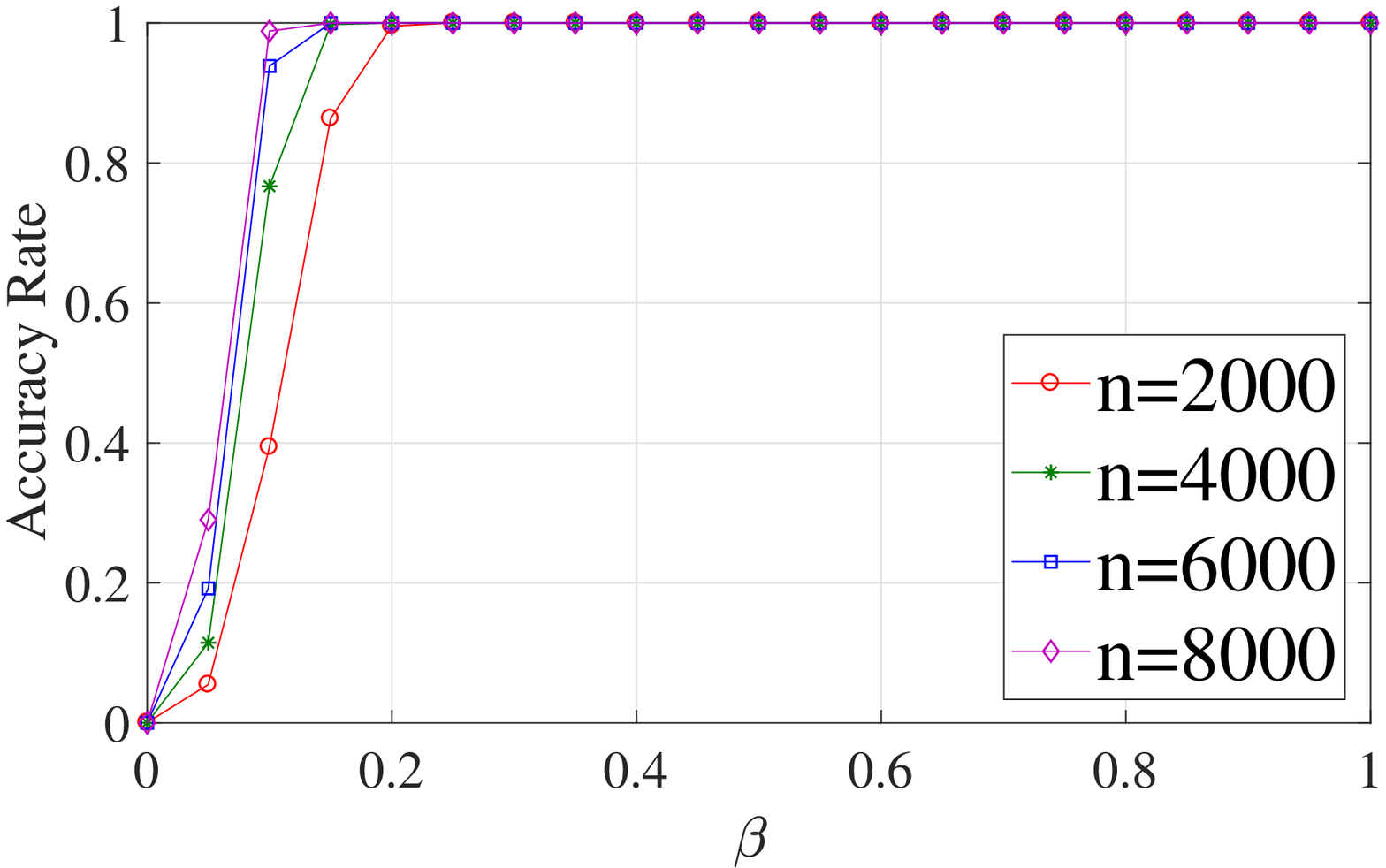}\label{fig:onehop1sub1}}
\subfigure[x-axis is $\beta/\sqrt{\log n/n}$.]{
 \includegraphics[scale=0.37]{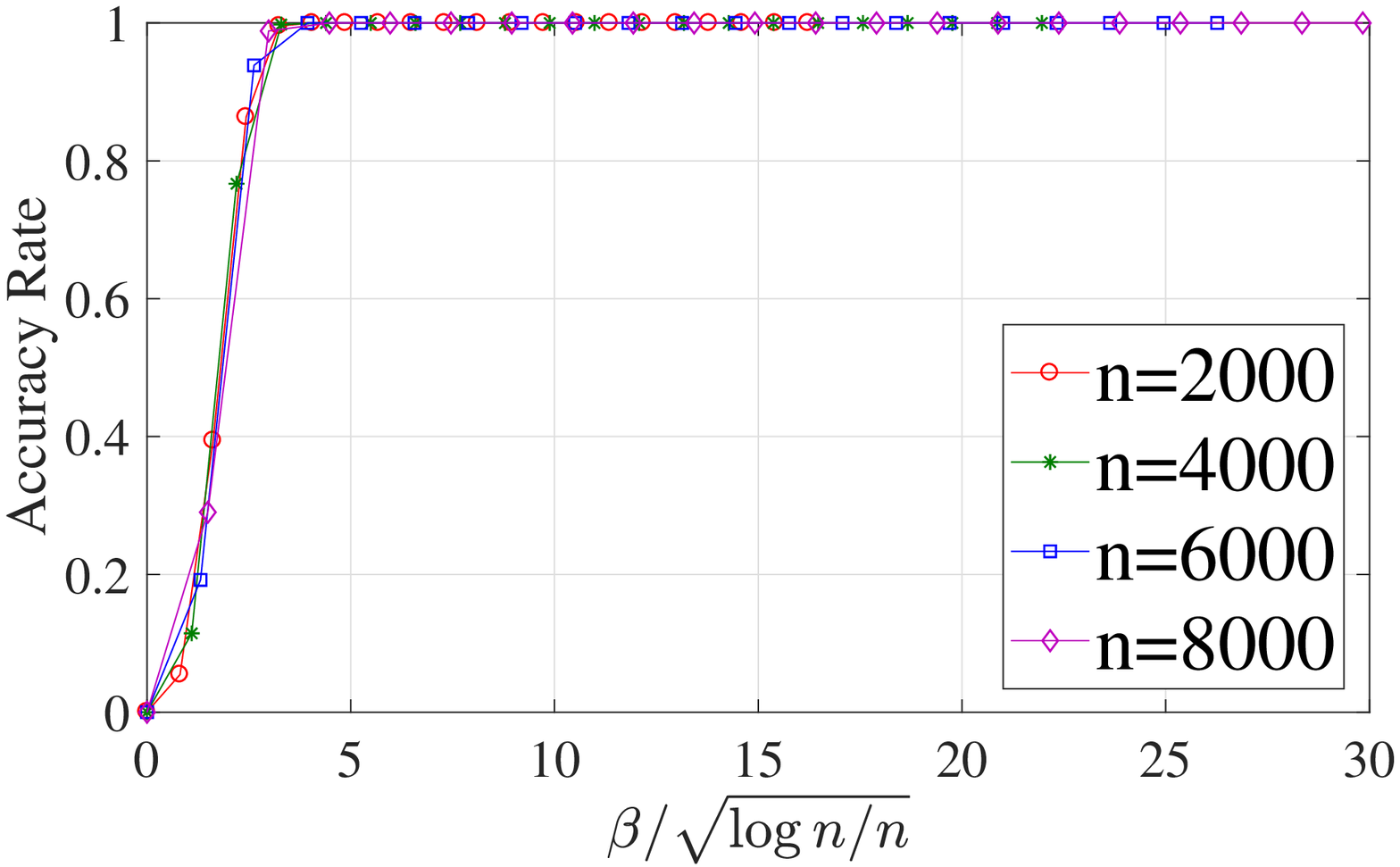}\label{fig:onehop1sub2}}
\caption{The 1-hop algorithm with varying $n$ and $p=n^{-\frac{1}{3}}$. Fix $s=0.8$.}
\label{fig:onehop1}
\end{figure}

\begin{figure}[H]
\centering
\subfigure[x-axis is $\beta$.]{
\includegraphics[scale=0.37]{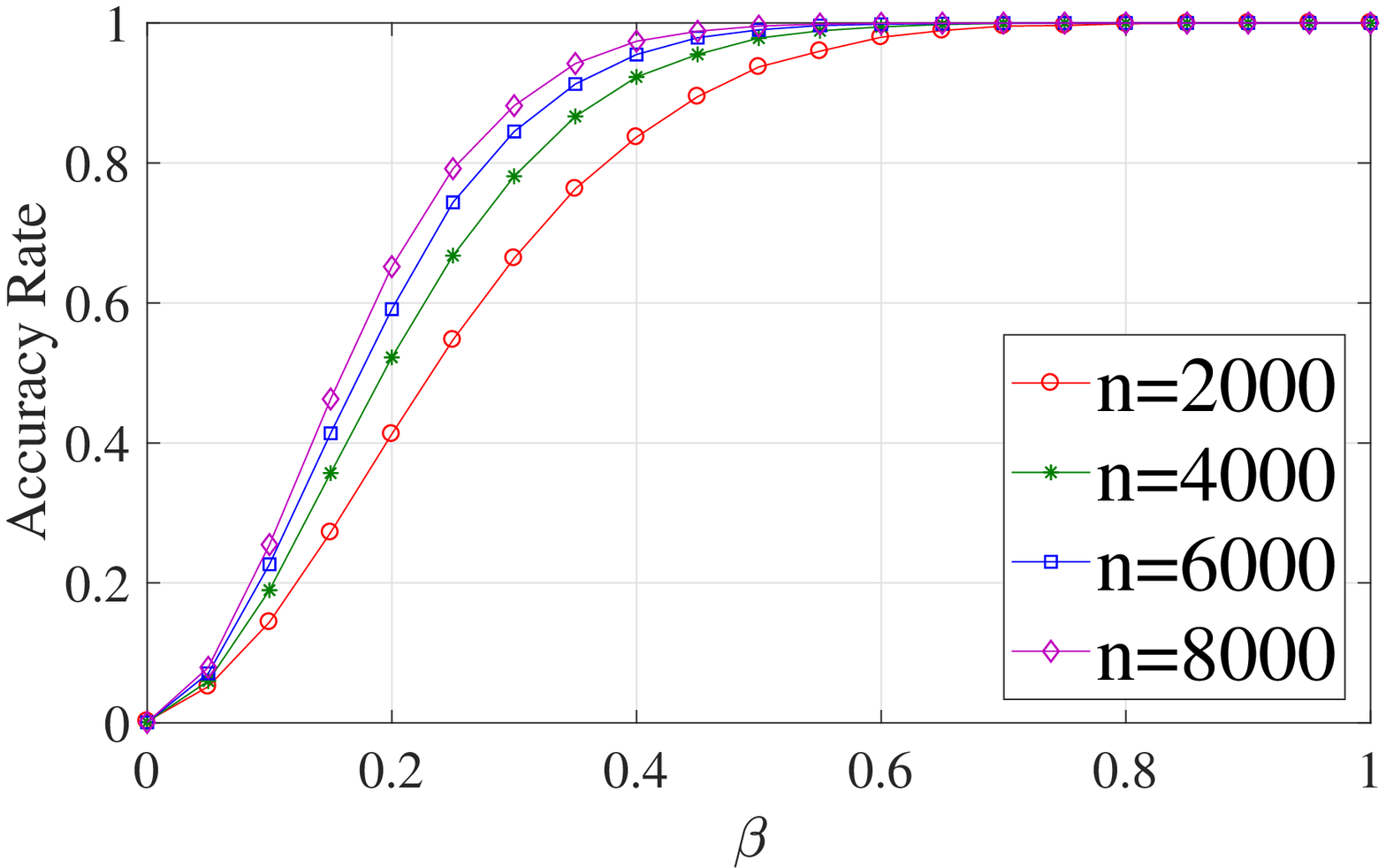}\label{fig:onehop2sub1}}
 \subfigure[x-axis is $\beta/(\log n/np)$.]{
 \includegraphics[scale=0.37]{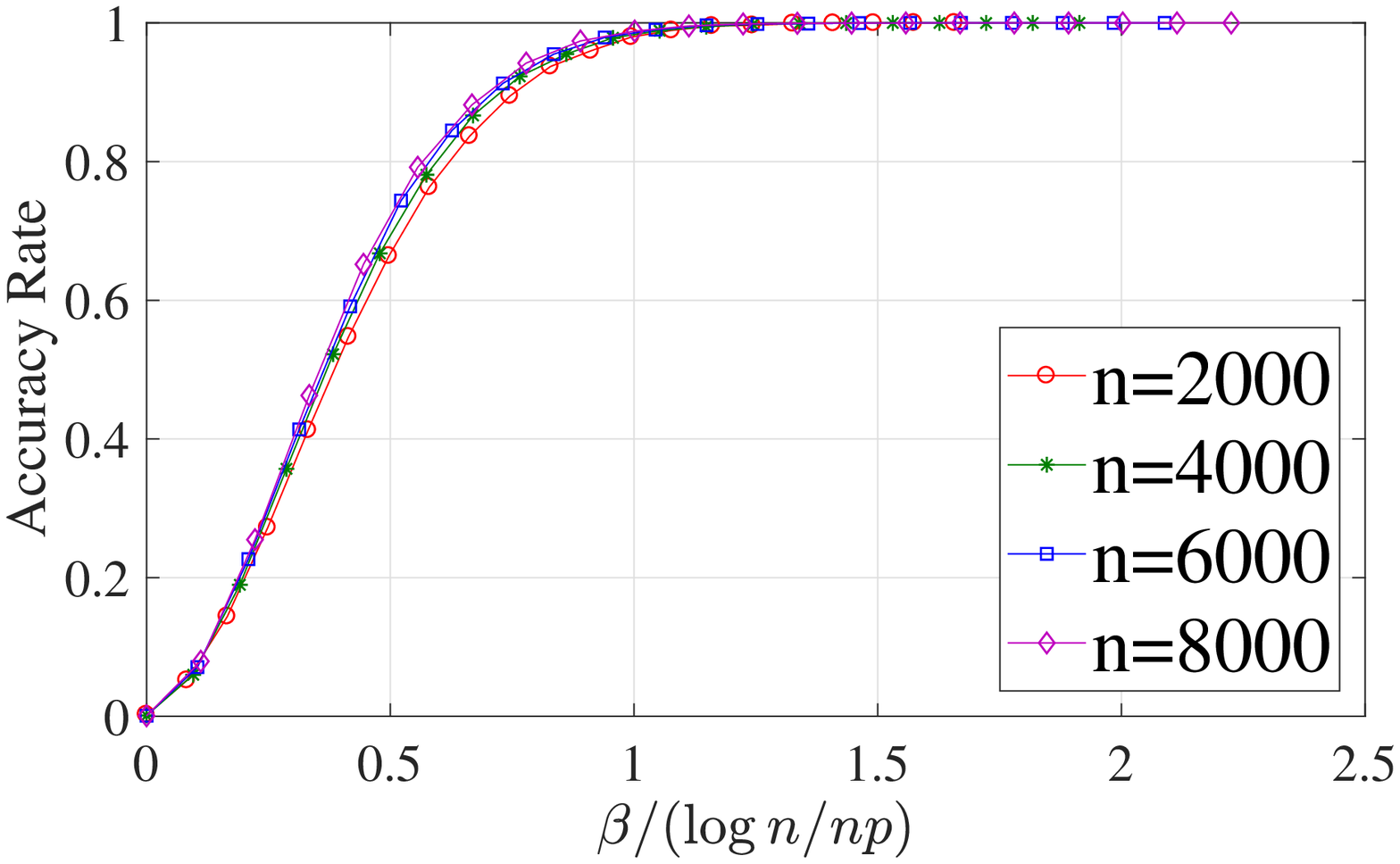}\label{fig:onehop2sub2}}
\caption{The 1-hop algorithm with varying $n$ and $p=n^{-\frac{2}{3}}$. Fix $s=0.8$.}
\label{fig:onehop2}
\end{figure}

Next, we simulate the performance of the 2-hop algorithm for $p=n^{-\frac{3}{5}}$, $p=n^{-\frac{17}{24}}$, and $p=n^{-\frac{4}{5}}$. 
 The results are presented in \prettyref{fig:twohop1sub1},  \prettyref{fig:twohop2sub1} and \prettyref{fig:twohop3sub1}. 
% Next, we verify that the condition \prettyref{eq:conditionofbeta2} in \prettyref{thm:thm2hop} is  both sufficient and close to necessary for the 2-hop algorithm to exactly recover $\pi^*$. We simulate the performance of the 2-hop algorithm for $p=n^{-\frac{3}{5}}$, $p=n^{-\frac{17}{24}}$, and $p=n^{-\frac{4}{5}}$. 
% The results are presented in \prettyref{fig:twohop1sub1},  \prettyref{fig:twohop2sub1} and \prettyref{fig:twohop3sub1}. 
Since \prettyref{thm:thm2hop} predicts that the 2-hop algorithm succeeds in exact recovery with high probability when $\beta\gtrsim \max \left\{\sqrt{np^3\log n},\sqrt{\frac{\log n}{n}}, \frac{\log n}{n^2p^2}\right\}$, we rescale the x-axis in \prettyref{fig:twohop1sub2}, \prettyref{fig:twohop2sub2} and \prettyref{fig:twohop3sub2} as $\beta/\sqrt{np^3\log n}$ for $p=n^{-\frac{3}{5}}$, $\beta/\sqrt{\frac{\log n}{n}}$ for $p=n^{-\frac{17}{24}}$ and $\beta/\left(\frac{\log n}{n^2p^2}\right)$ for $p=n^{-\frac{4}{5}}$. As we can see in \prettyref{fig:twohop1sub2},  \prettyref{fig:twohop2sub2} and \prettyref{fig:twohop3sub2}, the curves for different $n$ align well with each other,  suggesting that condition \prettyref{eq:conditionofbeta2} is both sufficient and close to necessary for the 2-hop algorithm to succeed.

\begin{figure}[H]
\centering
\subfigure[x-axis is $\beta$.]{
\includegraphics[scale=0.37]{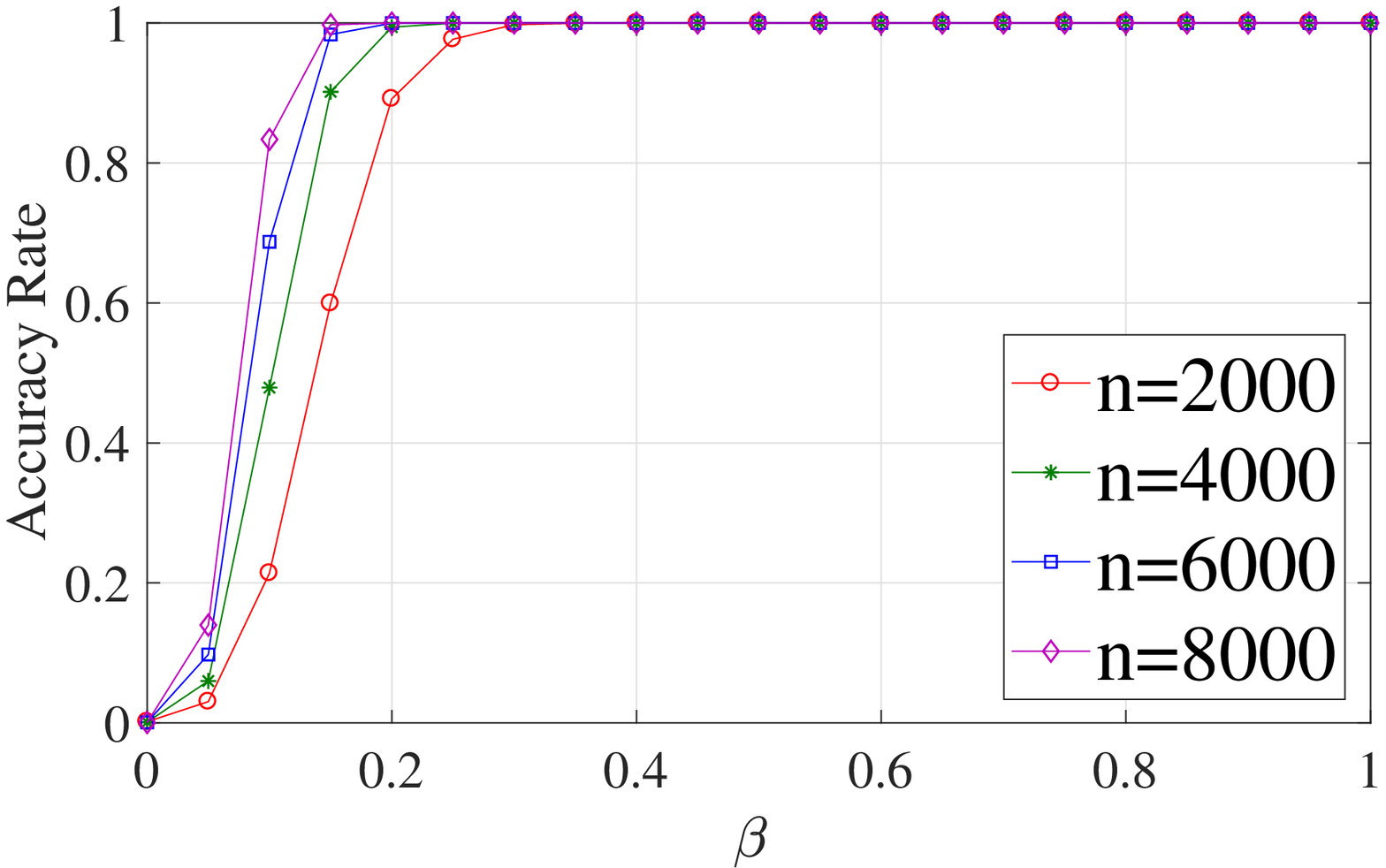}\label{fig:twohop1sub1}}
\subfigure[x-axis is $\beta/\sqrt{np^3\log n}$.]{
\includegraphics[scale=0.37]{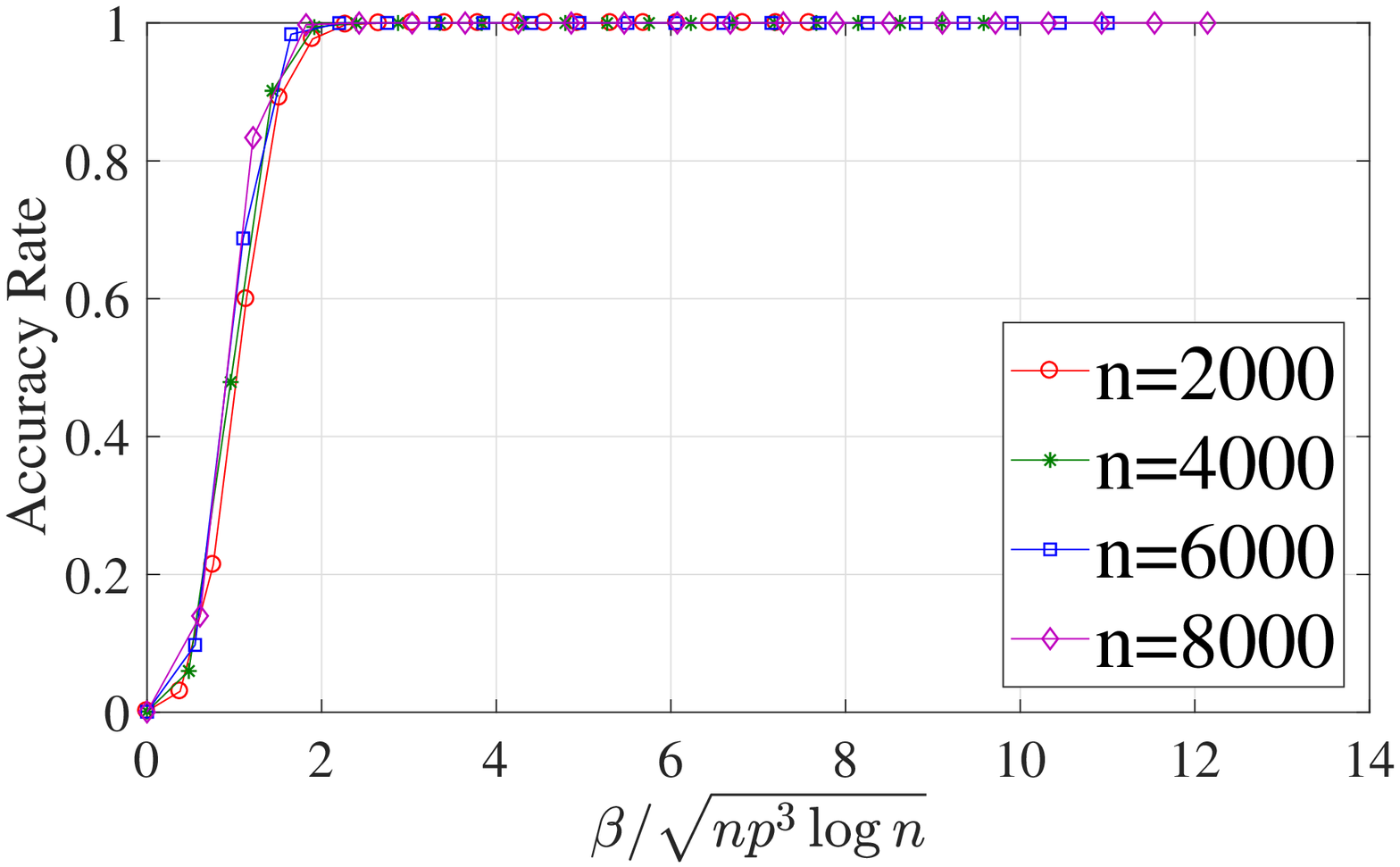}\label{fig:twohop1sub2}}
\caption{The 2-hop algorithm with varying $n$ and $p=n^{-\frac{3}{5}}$. Fix $s=0.8$.}
\label{fig:twohop1}
\end{figure}

\begin{figure}[H]
\centering
\subfigure[x-axis is $\beta$.]{
\includegraphics[scale=0.37]{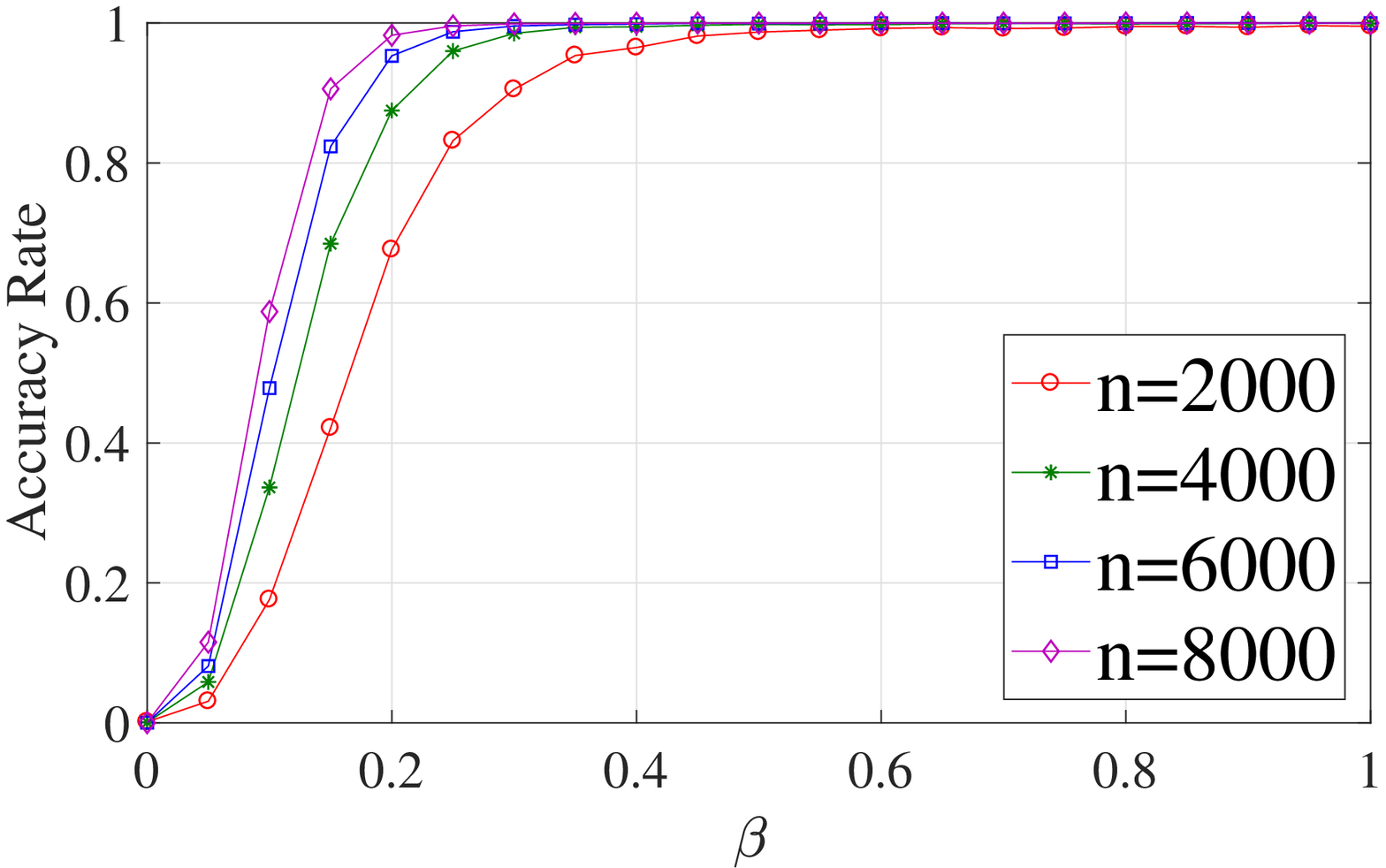}\label{fig:twohop2sub1}}
\subfigure[x-axis is $\beta/\sqrt{\log n/n}$.]{
\includegraphics[scale=0.37]{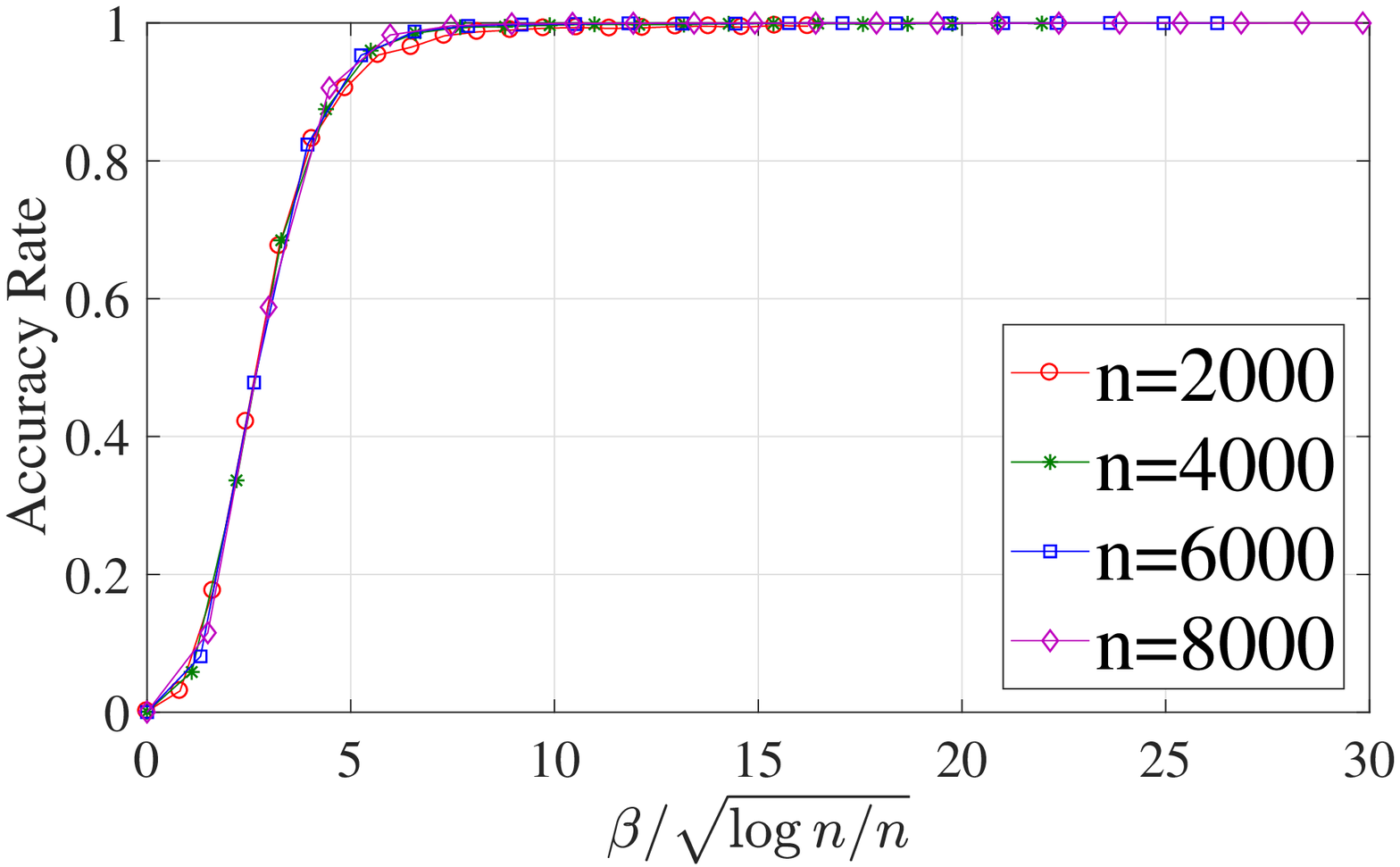}\label{fig:twohop2sub2}}
\caption{The 2-hop algorithm with varying $n$ and $p=n^{-\frac{17}{24}}$. Fix $s=0.8$.}
\label{fig:twohop2}
\end{figure}

\begin{figure}[H]
\centering
\subfigure[x-axis is $\beta$.]{
\includegraphics[scale=0.37]{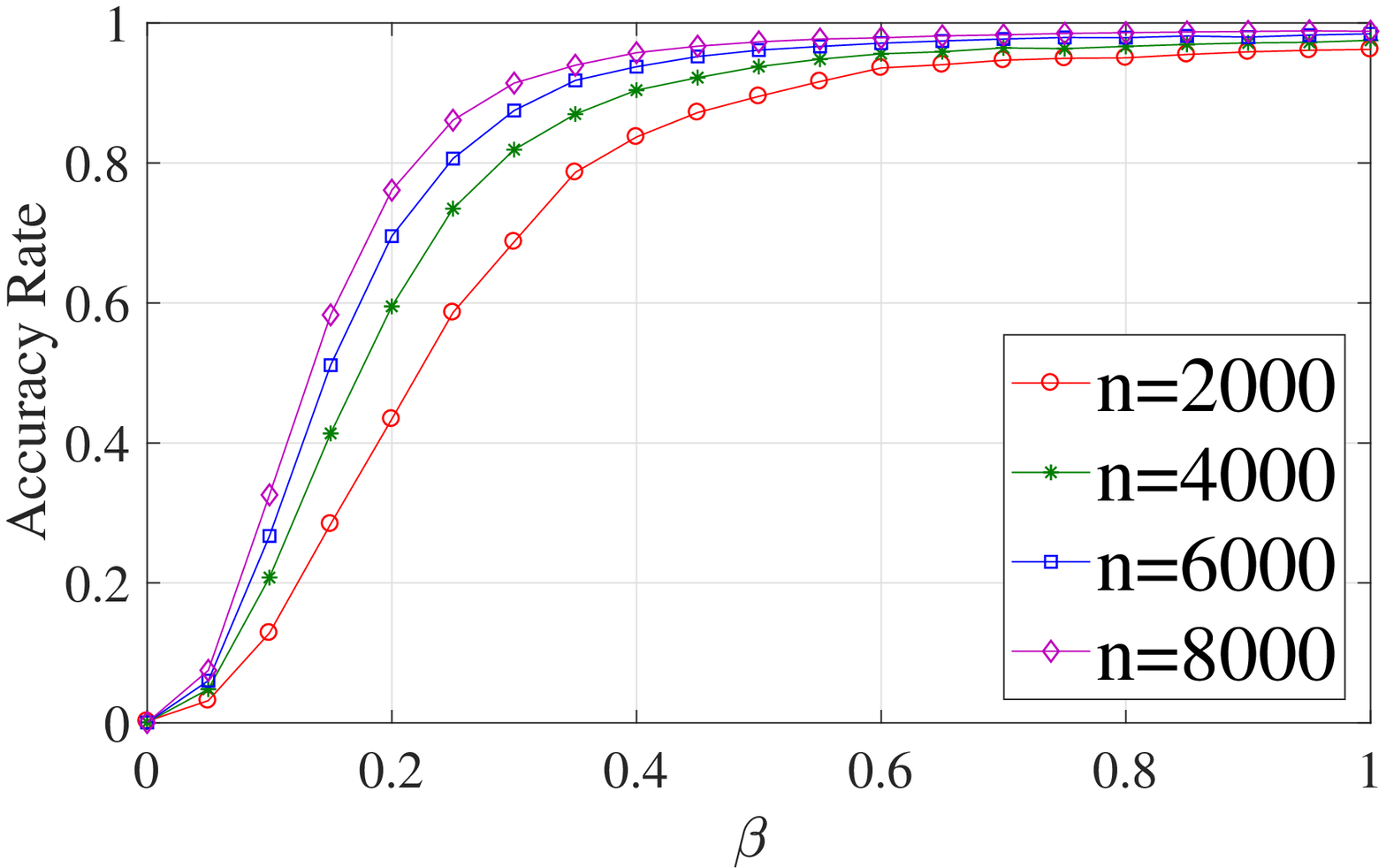}\label{fig:twohop3sub1}}
\subfigure[x-axis is $\beta/(\log n/n^2p^2)$.]{\includegraphics[scale=0.37]{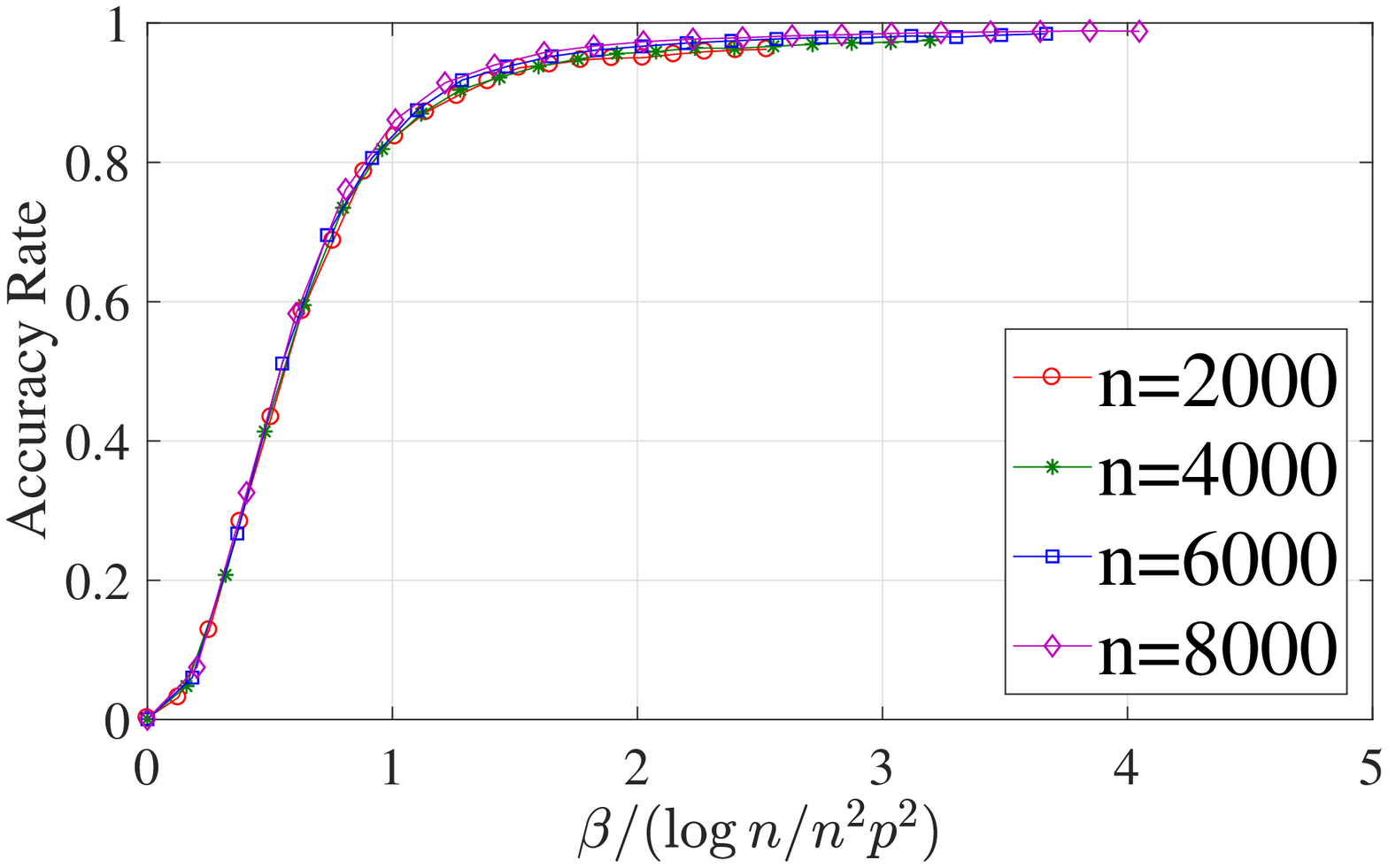}\label{fig:twohop3sub2}}
\caption{The 2-hop algorithm with varying $n$ and $p=n^{-\frac{4}{5}}$. Fix $s=0.8$.}
\label{fig:twohop3}
\end{figure}

In addition, if $s=1$ and $p=n^{-\frac{1}{2}}$, we show in \prettyref{fig:twohop4} that the curves for different $n$ align well when we rescale the x-axis as $\beta/\sqrt{\frac{\log n}{n}}$, but they do not align well with each other when the x-axis is rescaled as $\beta/\sqrt{np^3\log n}$. This result agrees with \prettyref{thm:thm2hop}, demonstrating that condition \prettyref{eq:oldcriteria} derived from the old criteria~\prettyref{eq:2_hop_cond_weak} is not tight.
\begin{figure}[H]
\centering
\subfigure[x-axis is $\beta/\sqrt{\log n/n}$.]{
\includegraphics[scale=0.37]{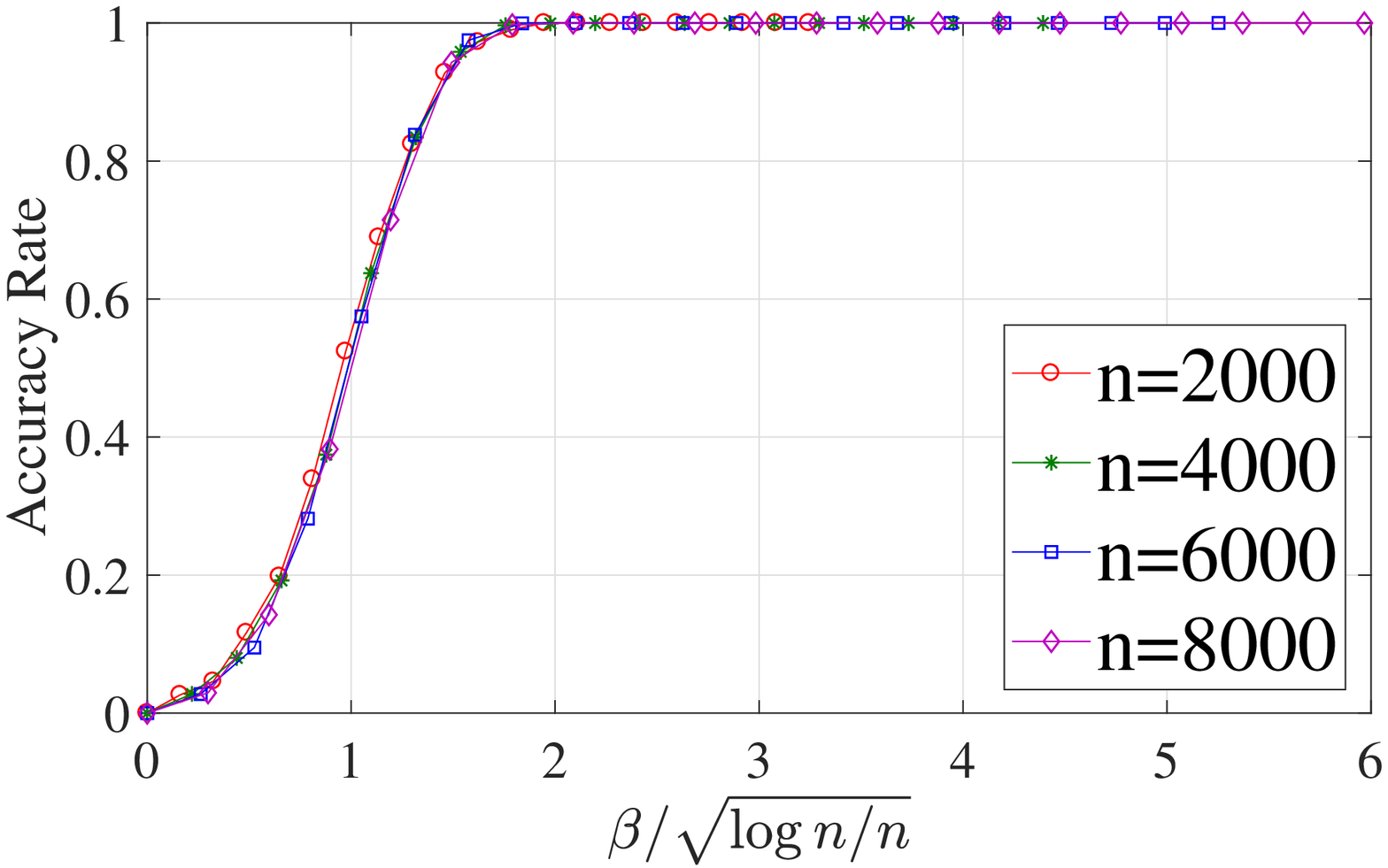}
\label{fig:twohop4sub1}}
\subfigure[x-axis is $\beta/\sqrt{np^3\log n}$.]{
\includegraphics[scale=0.37]{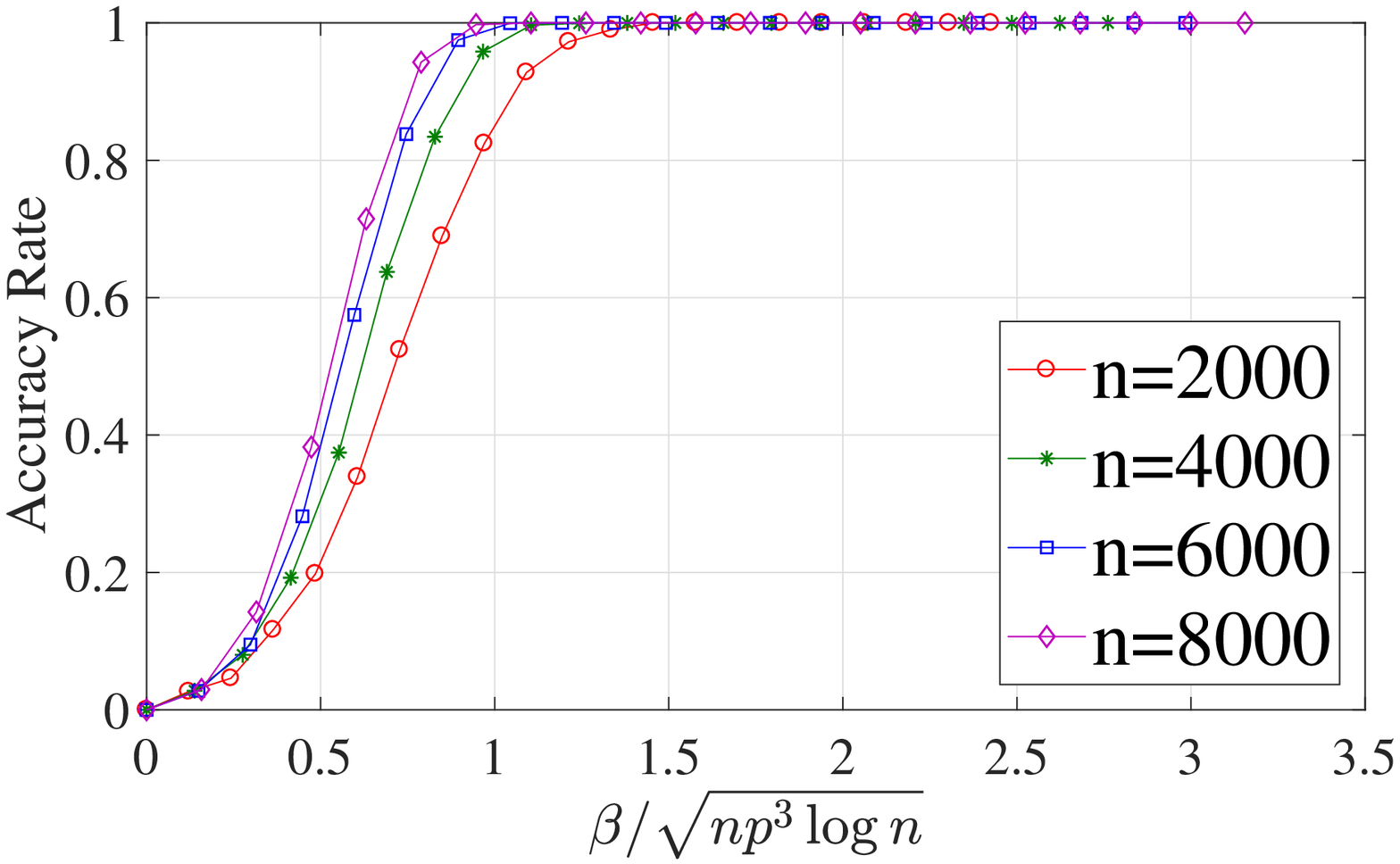}\label{fig:twohop4sub2}}
\caption{The 2-hop algorithm with varying $n$ and $p=n^{-\frac{1}{2}}$. Fix $s=1$.}
\label{fig:twohop4}
\end{figure}

\section{The scalability of our algorithm and feasible parallel implementation }\label{app:scalable}
%We believe that our algorithm is scalable. 
We may further improve the time complexity of our $2$-hop algorithm by exploiting graph sparsity and parallel computing. 
Recall that the theoretical worst-case computational complexity of our algorithm is $O(n^\omega+n^2 \log n)$ for $2 \le \omega \le 2.373$, where $n^\omega$  denotes the complexity of $n\times n$ matrix multiplication.  
For sparse graphs, the computational complexity of our 2-hop algorithm is comparable to that of the NoisySeeds algorithm in \cite{10.14778/2794367.2794371} and the 1-hop algorithm in \cite{lubars2018correcting}. To see this, note that there are only two differences in the execution of our 2-hop algorithm: (i)
To compute the number of $2$-hop witnesses for all vertex-pairs, for every seed $(w,\pi(w))$, 
our algorithm needs to 
compute the set of 2-hop neighbors of $w$ (resp.\ $\pi(w)$) in $G_1$ (resp.\ $G_2$), which takes $O\left( c^4\right)$ steps, where $c$ denotes the average degree. 
Thus in total it takes $O\left( nc^4 \right)$ steps. Hence, for sparse graphs with small average degree $c$, finding 2-hop witnesses only increases complexity slightly compared to finding 1-hop witnesses;
(ii) We use greedy max weight matching  (GMWM) rather than simple thresholding in \cite{10.14778/2794367.2794371}. As the time complexity of GMWM is $O(n^2 \log n)$ and the thresholding procedure needs to go through all the $n^2$ vertex-pairs, GMWM only incurs an additional $\log n$ factor to time complexity. Thus, the computational complexity of our 2-hop algorithm is comparable to others. 

For very large graphs, one may want to run these algorithms parallelly. Our 2-hop algorithm can be executed in parallel as follows. First, it is easy to turn (i) into parallel implementation. Second, if the complexity of (ii) is an issue, we can instead run the following modification: For each vertex $u$ in $G_1$, matches it to $v$ in $G_2$ that has the largest number of 2-hop witnesses; Output failure if there is any inconsistency in the final matching. This procedure can then be executed across all nodes in $G_1$  (or $G_2$) in parallel. This parallelizable version of the 2-hop algorithm can still provide perfect recovery if criteria \prettyref{eq:2_hop_cond_weak} holds.
%the \emph{minimum} number of 2-hop witnesses among true pairs is greater than the \emph{maximum} number of 2-hop witnesses among fake pairs. 
This criteria is satisfied with high probability under condition \prettyref{eq:oldcriteria}, as discussed in \prettyref{rmk:2_hop_cond_not_tight}. Hence the parallelizable version of the 2-hop algorithm can achieve perfect recovery under condition \prettyref{eq:oldcriteria}. Thus, we believe our $2$-hop algorithm can scale to very large graphs with strong matching performance.

%\section{Useful Concentration Inequalities}
\section{Preliminary Results}\label{sec:bound}

We first present some useful concentration inequalities
for the sum of independent random variables. 
%that are useful for the proofs of \prettyref{thm:thm1hop} and \prettyref{thm:thm2hop}. 
%Let $[n]=\{1,2,...,n\}$.

%\paragraph{Fact} Fixing a node-pair $(u,\pi(v))$, for two other node-pair $(u_i,\pi(v_i))$ and $(u_j,\pi(v_j))$, if $(u,u_i)\neq(u,u_j)$, $(u,u_i)\neq(v,v_j)$, $(v,v_i)\neq(u,u_j)$ and $(v,v_i)\neq(v,v_j)$, the event 

%\paragraph{Notation} For any positive integer $n$, let $[n]=\{1,2,...,n\}$.

\begin{theorem}\label{thm:chernoffbound} Chernoff Bound (\cite{Dubhashi2009ConcentrationOM}): Let $X=\sum_{i\in[n]}X_i$, where $X_i$'s are independent random variables taking values in $\{0,1\}$. Then, for $\delta\in (0,1)$,
\begin{align*}
\prob{X\leq (1-\delta)\expect{X}}\leq \exp\left(-\frac{{\delta}^2}{2}\expect{X}\right),\
\prob{X\geq (1+\delta)\expect{X}}\leq \exp\left(-\frac{{\delta}^2}{3}\expect{X}\right).
\end{align*}
\end{theorem}

%In this section, we present a supporting lemma that is useful for the proofs of \prettyref{lmm:Ruv} and \prettyref{lmm:Tuv}.
As a corollary of \prettyref{thm:chernoffbound},
we obtain the following lemma, which will be
useful for the proofs of \prettyref{lmm:Ruv} and \prettyref{lmm:Tuv}.
\begin{lemma}\label{lmm:bound}
Let $X$ denote a random variable such that $X\sim\Binom (n-1, \alpha)$. If $\alpha\in\left[ ps^2,1\right)$ and $nps^2\ge 128\log n$, then
\begin{align*}
    \prob{X\leq (1-\epsilon)(n-1)\alpha}\leq n^{-6},\ 
    \prob{X\geq (1+\epsilon)(n-1)\alpha}\leq n^{-4},
\end{align*}
where $\epsilon$ is given in \prettyref{eq:epsilon}, \ie, $\epsilon=\sqrt{\frac{12\log n}{(n-1)ps^2}}\leq\frac{1}{3}$.
\end{lemma}
\begin{proof}
Since $X\sim $ Binom $(n-1, \alpha)$ and $\epsilon=\sqrt{\frac{12\log n}{(n-1)ps^2}}<\frac{1}{3}$, applying Chernoff bound in \prettyref{thm:chernoffbound} and using 
$\alpha \ge ps^2$ yields
\begin{align*}
\prob{X\leq (1-\epsilon)(n-1)\alpha}& \leq \exp \left(-\frac{\epsilon^2 (n-1)\alpha}{2}\right)
%\nonumber\\
%&=\exp\left(-\frac{6\alpha\log n}{ps^2}\right) 
\leq n^{-6},\\
\prob{X\leq (1+\epsilon)(n-1)\alpha}&\leq \exp \left(-\frac{\epsilon^2 (n-1)\alpha}{3}\right)
%\nonumber\\
%&=\exp\left(-\frac{4\alpha\log n}{ps^2}\right) 
\leq n^{-4}.
\end{align*}
\end{proof}

\begin{theorem}\label{thm:bernstein} Bernstein's Inequality (\cite{Dubhashi2009ConcentrationOM}):
Let $X=\sum_{i\in[n]}X_i$, where $X_i$'s are independent random variables such that $|X_i|\leq K$ almost surely. Then, for $t> 0$, we have
\begin{align*}
\prob{X\geq \expect{X}+t}\leq \exp\left(-\frac{{t}^2}{2(\sigma^2+Kt/3)}\right),
\end{align*}
where $\sigma ^2= \sum_{i \in [n]} \var(X_i)$ is the variance of $X$. It follows then for $\gamma>0$, we have
\begin{align*}
\prob{X\geq \expect{X}+\sqrt{2\sigma^2\gamma}+\frac{2K\gamma}{3}}\leq \exp(-\gamma).
\end{align*}

%The obtained estimate holds for $\prob{X\leq\expect{X}- \sqrt{2\sigma^2\gamma}-\frac{2K\gamma}{3}}$
%too (by considering $-X$), i.e.,
Similarly, by considering $-X,$ it follows that
\begin{align*}
\prob{X\leq \expect{X}-\sqrt{2\sigma^2\gamma}-\frac{2K\gamma}{3}}\leq \exp(-\gamma).
\end{align*}

\end{theorem}

\begin{corollary}\label{lmm:Bernsteinbound}
Let $X$ denote a random variable such that $X\sim \Binom (n, \alpha)$. If $n\in[n_{\min},n_{
\max}]$, then for $\gamma>0$,
\begin{align}
  \prob{X\leq n_{\min}\alpha-\sqrt{2n_{\max}\alpha\gamma}-\frac{2\gamma}{3}} & \leq \exp(-\gamma), \label{eq:binom_low_tail}\\
    \prob{X\geq n_{\max}\alpha+\sqrt{2n_{\max}\alpha\gamma}+\frac{2\gamma}{3}} & \leq \exp(-\gamma). \label{eq:binom_upp_tail}
\end{align}
Moreover, 
\begin{align}
\prob{X\geq 2n_{\max}\alpha+\frac{4\gamma}{3}} & \leq \exp(-\gamma)
\label{eq:binom_upp_tail_2}
\end{align}
\end{corollary}
\begin{proof}
The proof of \prettyref{eq:binom_low_tail}
and \prettyref{eq:binom_upp_tail} follows by invoking
\prettyref{thm:bernstein} with $\sigma^2=n\alpha(1-\alpha)$ and $K=1$ and using
the assumption that $n\in[n_{\min},n_{
\max}]$. In view of $2\sqrt{ab} \le a+b$, \prettyref{eq:binom_upp_tail_2}
follows from \prettyref{eq:binom_upp_tail}.
% Since $X\sim $ Binom $(n, \alpha)$, applying Bernstein’s inequality in \prettyref{thm:bernstein} yields
% \begin{align}
% \prob{X\leq n\alpha -\sqrt{2n\alpha(1-\alpha)\gamma}-\frac{2\gamma}{3}}&\leq \exp \left(-\gamma\right).
% \end{align}
% Since $n_x\in[n_{min},n_{max}]$, we have
% \begin{align}
% n_x\alpha -\sqrt{2n_x\alpha(1-\alpha)\gamma}-\frac{2\gamma}{3}\geq n_{min}\alpha -\sqrt{2n_{max}\alpha\gamma}-\frac{2\gamma}{3}.
% \end{align}
% and
% \begin{align}
% \prob{X\leq n_{min}\alpha -\sqrt{2n_{max}\alpha\gamma}-\frac{2\gamma}{3}}&\leq \exp \left(-\gamma\right).
% \end{align}
% Similarly, 
% \begin{align}
% \prob{X\leq n_x\alpha +\sqrt{2n_x\alpha\gamma}+\frac{2\gamma}{3}}&\leq \exp \left(-\gamma\right).
% \end{align}
% Since $n_x\in[n_{min},n_{max}]$, we have
% \begin{align}
% n_x\alpha +\sqrt{2n_x\alpha(1-\alpha)\gamma}+\frac{2\gamma}{3}\leq n_{max}\alpha +\sqrt{2n_{max}\alpha\gamma}+\frac{2\gamma}{3}.
% \end{align}
% and
% \begin{align}
% \prob{X\leq n_{max}\alpha +\sqrt{2n_{max}\alpha(1-\alpha)\gamma}+\frac{2\gamma}{3}}&\leq \exp \left(-\gamma\right).
% \end{align}
\end{proof}

Next, we present a concentration inequality for the sum of dependent random variables. To this end, we first introduce the notion of dependency graph.
\begin{definition}
Given random variables $\{X_i\}_{i\in[n]}$, the dependency graph is a graph $\Gamma$ with vertex set $[n]$ such that if $i\in [n]$ is not connected by an edge to any vertex in $\mathcal{J}\subset [n]$, then $X_i$ is independent of $\{X_j\}_{j\in \mathcal{J}}$.
\end{definition}

\begin{theorem}\label{thm:drv} (\cite{Janson2004LargeDF}) Let $X=\sum_{i\in[n]}X_i$, where $X_i$'s are random variables such that $X_i-\expect{X_i}\leq K$ for some $K>0$. Let $\Gamma$ denote a dependency graph for $\{X_i\}$ and $\Delta(\Gamma)$ denote the maximum degree of $\Gamma$. Let $\sigma^2=\sum_{i\in [n]}\var (X_i)$. Then, for $t\geq0$, 
\begin{align*}
\prob{X\geq \expect{X} +t}\leq \exp\left(-\frac{8{t}^2}{25\Delta_1(\Gamma)(\sigma^2+Kt/3)}\right),
\end{align*}
where $\Delta_1(\Gamma)= \Delta(\Gamma)+1$.  It follows then for $\gamma>0$, we have
\begin{align*}
\prob{X\geq \expect{X} +\sqrt{\frac{25\Delta_1(\Gamma)}{8}\sigma^2\gamma}+\frac{25\Delta_1(\Gamma)K\gamma}{24}}\leq \exp(-\gamma).
\end{align*}

If the assumption $X_i-\expect{X_i}\leq K$ is reversed to $X_i-\expect{X_i}\geq -K$, 
%the obtained estimate holds for $\prob{X\leq \expect{X} -\sqrt{\frac{25\Delta_1(\Gamma)}{8}\sigma^2\gamma}-\frac{25\Delta_1(\Gamma)K\gamma}{24}}$ too (by considering $-X$), i.e.,
then by considering $-X$, it follows that
\begin{align*}
\prob{X\leq \expect{X} -\sqrt{\frac{25\Delta_1(\Gamma)}{8}\sigma^2\gamma}-\frac{25\Delta_1(\Gamma)K\gamma}{24}}\leq\exp(-\gamma).
\end{align*}
\end{theorem}

Finally, we will repeatedly use the following simple inequality. 

\begin{theorem}\label{thm:BernoulliInequality}
For $r\geq 0$, every real number $x\in (0,1)$ and $rx\leq 1$, it holds that
\begin{align*}
r\log{(1-x)}\leq \log\left(1-\frac{rx}{2} \right).
\end{align*}
\end{theorem}

\begin{proof}
Set 
$
f(x)=r\log(1-x) - \log\left(1-\frac{rx}{2} \right),
$
Then  $f(0)=0$ and $f'(x) =\frac{r(rx-x-1)}{(2-rx)(1-x)} \le 0$.
Thus $f(x) \le 0$, completing the proof.
%\begin{align*}
%f'(x)=&\frac{-r}{1-x}+\frac{r}{2-rx}\\
%=&\frac{r(rx-x-1)}{(2-rx)(1-x)}\\
%\leq &0.
%\end{align*} 
%It shows that $f(x)$ decreases and $f(0)=0$. Thus, we conclude that the statement is true for $r>0$, $x\in (0,1)$ and $rx\leq 1$.
\end{proof}

% \section{A Supporting Lemma}\label{sec:bound}
% In this section, we present a supporting lemma that is useful for the proofs of \prettyref{lmm:Ruv} and \prettyref{lmm:Tuv}.
% \begin{lemma}\label{lmm:bound}
% Let $X$ denote a random variable such that $X\sim\Binom (n-1, \alpha)$. If $\alpha\in\left[ ps^2,1\right)$, then
% \begin{align*}
%     \prob{X\leq (1-\epsilon)(n-1)\alpha}\leq n^{-6},\ 
%     \prob{X\geq (1+\epsilon)(n-1)\alpha}\leq n^{-4},
% \end{align*}
% where $\epsilon$ is given in \prettyref{eq:epsilon}, \ie, $\epsilon=\sqrt{\frac{12\log n}{(n-1)ps^2}}\leq\frac{1}{3}$.
% \end{lemma}
% \begin{proof}
% Since $X\sim $ Binom $(n-1, \alpha)$ and $\epsilon=\sqrt{\frac{12\log n}{(n-1)ps^2}}<\frac{1}{3}$, applying Chernoff bound in \prettyref{thm:chernoffbound} and using 
% $\alpha \ge ps^2$ yields
% \begin{align*}
% \prob{X\leq (1-\epsilon)(n-1)\alpha}& \leq \exp \left(-\frac{\epsilon^2 (n-1)\alpha}{2}\right)
% %\nonumber\\
% %&=\exp\left(-\frac{6\alpha\log n}{ps^2}\right) 
% \leq n^{-6},\\
% \prob{X\leq (1+\epsilon)(n-1)\alpha}&\leq \exp \left(-\frac{\epsilon^2 (n-1)\alpha}{3}\right)
% %\nonumber\\
% %&=\exp\left(-\frac{4\alpha\log n}{ps^2}\right) 
% \leq n^{-4}.
% \end{align*}
% \end{proof}

\section{Postponed Proofs for \prettyref{thm:thm1hop}}

\subsection{Proof of \prettyref{lmm:W1uv}}\label{sec:proof1hopl2}Recall that $A_1$ and $B_1$ are the adjacency matrix for $G_1$ and $G_2$, respectively. 
%Assume without loss of generality that the true mapping $\pi$ is the identity permutation.
%Let $T\triangleq \{ i: \pi(i)=i\}$ denote the set of fixed points of $\pi$. Then $T$ corresponds to the set of correct seeds with $|T|=n\beta$. 
By the definition of $1$-hop witness, we have 
\begin{equation}
    \begin{aligned}
    W_1(u,v)=\sum_{i\in [n]\setminus\{u,\pi^{-1}(v)\}} A_1(u,i)B_1(v,\pi(i)).  
    \end{aligned}
\end{equation}
Let $Z_i \triangleq A_1(u,i)B_1(v,\pi(i))$. 
Note that $Z_{v}$ is dependent on $Z_{\pi^{-1}(u)}$. Thus, we exclude these two seeds and consider the remaining seeds.
For all $i \in [n]\setminus\{u,v,\pi^{-1}(u),\pi^{-1}(v)\}$, $Z_i \iiddistr \Bern(p^2s^2)$. It follows that
\begin{equation}\label{eq:W1-fake-z}
\begin{aligned}
%\prob{X\leq x_{min}}\leq&\prob{\sum_{i=1}^{n\beta-1}X_i\leq x_{min}}\\
\prob{\sum_{i\in [n]\setminus\{u,v,\pi^{-1}(u),\pi^{-1}(v)\}} Z_i \ge \psi_{\max}-2 }
\le \prob{ \Binom(n, p^2s^2) \ge \psi_{\max}-2 }  
%\leq &\mathbb{P}\left\{\sum_{i=1}^{n\beta-1}X_i\leq(n\beta-1)ps^2-\sqrt{5(n\beta-1)ps^2(1-ps^2)\log n}-\frac{5}{3}\log n\right\}\\
%\leq &\exp\left(-\frac{5}{2}\log n\right)
\le n^{-\frac{7}{2}},
\end{aligned}
\end{equation}
where that last inequality follows from Bernstein’s inequality given in \prettyref{thm:bernstein} with  $\gamma=\frac{7}{2}\log n$ and $K=1$.

Finally, adding back $Z_{v}$ and $Z_{\pi^{-1}(u)}$
yields the desired conclusion~\prettyref{eq:W1uv}.

\subsection{Proof of \prettyref{thm:thm1hop}}\label{sec:proof1hop}
% Based on \prettyref{lmm:W1uu} and the union bound, we have
% \begin{align*}
% \prob{\bigcap_{u\in V}\left\{W_1(u,\pi(u))> x_{min}+y_{min}\right\}}\geq1-n\cdot n^{-\frac{7}{3}}=1-n^{-\frac{4}{3}}.
% \end{align*}

% Based on \prettyref{lmm:W1uv} and the union bound, we have
% \begin{align*}
% \prob{\bigcap_{v,w\in V}\left\{W_1(v,\pi(w))<z_{\max}\right\}}\geq 1-n^2\cdot n^{-\frac{7}{2}}=1-n^{-\frac{3}{2}}.
% \end{align*}
Since the bound of the number of 1-hop witnesses is provided by \prettyref{lmm:W1uv} and \prettyref{lmm:W1uu},
it remains to  verify $x_{\min}+y_{\min}-\psi_{\max}\geq0$ under the condition of \prettyref{thm:thm1hop}. Note that
\begin{align*}
&x_{\min}+y_{\min}-\psi_{\max}\\
=&n\beta p(1-p)s^2-\sqrt{5n\beta ps^2\log n}-(5+\sqrt{7})\sqrt{np^2s^2\log n}-ps^2-2p^2s^2-\frac{37}{3}\log n-2.
\end{align*}
First, by assumption that $\beta\geq \frac{45\log n}{np(1-p)^2s^2}$, we have
\begin{equation}
\begin{aligned}
\frac{1}{3}n\beta p(1-p)s^2\geq \sqrt{\frac{45\log n}{np(1-p)^2s^2}}\cdot \frac{1}{3}n\sqrt{\beta}p(1-p)s^2 =\sqrt{5n\beta ps^2\log n}.\label{eq:verify1begin}
\end{aligned}
\end{equation}
Second, by assumption that $\beta\geq 30\sqrt{\frac{\log n}{n(1-p)^2s^2}}$, we have
\begin{equation}
\begin{aligned}
\frac{1}{3}n\beta p(1-p)s^2\geq 30\sqrt{\frac{\log n}{n(1-p)^2s^2}}\cdot \frac{1}{3}np(1-p)s^2 \geq (5+\sqrt{7})\sqrt{np^2s^2\log n}.
\end{aligned}
\end{equation}
Third, by assumption that $\beta\geq \frac{45\log n}{np(1-p)^2s^2}$ and $n$ is sufficiently large, we have
\begin{equation}
\begin{aligned}
\frac{1}{3}n\beta p(1-p)s^2\geq \frac{45\log n}{np(1-p)^2s^2}\cdot \frac{1}{3}np(1-p)s^2\geq\frac{37}{3}\log n+2+ps^2+2p^2s^2.\label{eq:verify1end}
\end{aligned}
\end{equation}

Combining \prettyref{eq:verify1begin}-\prettyref{eq:verify1end}, we have $x_{\min}+y_{\min}-\psi_{\max}\geq0$. 
%Taking a union bound, we have
Thus,
\begin{align*}
&\prob{\min_{u\in [n]}W_1(u,u)>\max_{u,v\in [n]: u\neq }W_1(u,v)}\\
\geq& 1-\prob{\bigcup_{u\in [n]}\left\{W_1(u,u)\leq x_{\min}+y_{\min}\right\}} -\prob{\bigcup_{u,v\in [n]: u\neq v}\left\{W_1(u,v)\geq \psi_{\max}\right\}} \\
\geq& 1-n^{-\frac{4}{3}}-n^{-\frac{3}{2}}\geq 1-n^{-1},
\end{align*}
where the second inequality holds by combining \prettyref{lmm:W1uv}  and \prettyref{lmm:W1uu}
with the union bound. 
Thus, GMWM outputs $\tilde{\pi}$ with $\prob{\tilde{\pi}=\pi^*}\geq1-n^{-1}$ under the $1$-hop algorithm.
%if we use the 1-hop algorithm to match graphs.

\section{Postponed Proofs for \prettyref{thm:thm2hop}}

\subsection{Proof of \prettyref{lmm:Ruv}}\label{sec:proofRuv}
%Recall that $A_1$ and $B_1$ are the adjacency matrix for $G_1$ and $G_2$, respectively. 
By definition, we have $a_u
%=\sum_{i\in [n]\setminus\{u\}}A_1(u,i)
\sim \Binom (n-1,ps)$. It follows from \prettyref{lmm:bound} that
\begin{align}\label{eq:Ruv-N1}
\prob{a_u\leq(1-\epsilon)(n-1)ps}\leq n^{-6},\ 
\prob{a_u\geq(1+\epsilon)(n-1)ps}\leq n^{-4}.
\end{align}
The same lower and upper bounds hold for $b_u$ analogously. 

%For any vertex $i \in [n]\backslash \{u\}$, we have $\prob{i\in C(u,u)}=\prob{A_1(u,i)B_1(u,i)=1}=ps^2$. Note that $A_1(u,i)B_1(u,i)$ only depends on $\{u,i\}$, which are disjoint across different $i$. Thus, $A_1(u,i)B_1(u,i)$ are mutually independent across $i$.  Then, we have
Note that $c_{uu}
%=\sum_{i\in [n]\setminus \{u\}}A_1(u,i)B_1(u,i) 
\sim \Binom(n -1,ps^2)$. Applying \prettyref{lmm:bound} yields that 
\begin{align}\label{eq:Ruv-Cuu}
\prob{c_{uu}\leq(1-\epsilon)(n-1)ps^2}\leq n^{-6},\
\prob{c_{uu}\geq(1+\epsilon)(n-1)ps^2}\leq n^{-4}.
\end{align}
%The same lower and upper bound hold for $c_{vv}$ by similar proof.

%For all vertex $i \in [n]\setminus \{u,v\}$, we have $\prob{i\in C(u,v)}=\prob{A_1(u,i)B_1(v,i)=1}=p^2s^2$. Note that $A_1(u,i)B_1(v,i)$ only depends on $\{\{u,i\},\{v,i\}\}$, which are disjoint across $i\in [n]\setminus \{u,v\}$. Thus,  $A_1(u,i)B_1(v,i)$ are mutually independent across $i$. Thus, we have 
Also, for fake pairs $u \neq v$, $c_{uv}
%=\sum_{i \in [n]\setminus \{u,v\}} A_1(u,i)B_1(v,i) 
\sim\Binom(n-2,p^2s^2)$. Therefore, applying Bernstein’s inequality given in \prettyref{thm:bernstein} with  $\gamma=\frac{7}{2}\log n$ and $K=1$, we can get
%\nbr{This is not clear. Please write out the full proof. LY: I think it is clear now}
 \begin{align}\label{eq:Ruv-Cuv}
 &\prob{c_{uv} \geq \psi_{\max}}
 \overset{}{\le}\prob{\Binom(n-2,p^2s^2) \geq np^2s^2+\sqrt{7np^2s^2\log n}+\frac{7}{3}\log n}\leq n^{-\frac{7}{2}}.
\end{align}

% \nbr{Did you assume some condition on $p$ to get ths inequality?LY: addressed} \nb{It seems that you only need $np^2 \lesssim \log n$. As I mentioned to you before, it is better to keep each lemma independent.} 
 %By the definition of the 1-hop witnesses, we have $W_1(v,u)=\sum_{i\in [n]\setminus\{v,\pi^{-1}(u)\}} A_1(v,i)B_1(u,\pi(i)).$ For all $i \in [n]\setminus\{v,u,\pi^{-1}(u),\pi^{-1}(v)\}$, $A_1(v,i)B_1(u,\pi(i)) \iiddistr \Bern(p^2s^2)$. Thus, applying Bernstein’s inequality given in \prettyref{thm:bernstein} with  $\gamma=\frac{7}{2}\log n$ and $K=1$,
 According to \prettyref{lmm:W1uv}, we can get 
\begin{align}\label{eq:Ruv-Wvu}
     \prob{W_1(v,u) \geq \psi_{\max}}
   %\le\prob{\Binom(n,p^2s^2) \geq np^2s^2+\sqrt{7np^2s^2\log n}+\frac{7}{3}\log n   }
   \leq n^{-\frac{7}{2}}.
\end{align}
%where $(a)$ holds due to $np^2\leq \frac{1}{ \log n}$.

Taking the union bound over \prettyref{eq:Ruv-N1}--\prettyref{eq:Ruv-Wvu} 
yields the desired conclusion \prettyref{eq:Ruv}.
%\begin{align*}
%\prob{R_{uv}}\geq 1-4(n^{-6}+n^{-4})-2(n^{-6}+n^{-4})-2 n^{-\frac{7}{2}}  \geq 1-n^{-\frac{10}{3}} .
%\end{align*}

\subsection{Proof of \prettyref{lmm:W2uu}}\label{sec:proofW2uu}
Fixing any two vertices $u\neq v$, we condition on
$Q_{uv}$ such that the event $R_{uv}$ holds. 
Note that 
$$
W_2(u,u) = \sum_{j=1}^n A_2 (u,j) B_2 
\left(u, \pi(j) \right),
$$
where $A_2$ and $B_2$ are the $2$-hop adjacency matrix of $G_1$ and $G_2$, respectively. 
Note that for all $j \in N^{G_1}(u)\cup \{u\}$,
$A_2(u,j)=0$ by definition. Similarly,
for all $\pi(j) \in 
N^{G_2}(u)\cup \{u\}$, $B_2\left(u, \pi(j) \right)=0$. Thus, we define 
$$
J_u= N^{G_1}(u)\cup \{u\} \cup \pi^{-1}
\left( N^{G_2}(u) \right) \cup \pi^{-1} (u)
$$
and exclude the seeds in $J_u.$
%Then 
%$$
%W_2(u,u) = \sum_{j \in [n]\setminus J_u } A_2 (u,j) B_2 
%\left(u, \pi(j) \right).
%$$
Furthermore, note that we have conditioned on the $1$-hop neighborhoods of $v$ in $G_1$ and $G_2$.
In either $G_1$ or $G_2$, if $u$ and $v$ are connected, then a 1-hop neighbor of $v$ may automatically become the 2-hop neighbor of $u$. 
 Hence,  if $j$ is connected to $v$ in $G_1$
 or $\pi(j)$ is connected to $v$ in $G_2$,
 then conditioning on $Q_{uv}$ can change the probability that $A_2(u,j)B_2 
\left(u, \pi(j) \right)=1$.
To circumvent this issue, 
we further exclude the set $J_v$ of seeds and get that 
\begin{align}
W_2(u,u) 
& \ge \sum_{j \in [n]\setminus( J_u \cup J_v) } 
A_2 (u,j) B_2 
\left(u, \pi(j) \right) \nonumber\\
& =
\sum_{j \in F \setminus (J_u \cup J_v) }
A_2 (u,j) B_2 
\left(u,j\right)
+ \sum_{j \in [n] \setminus (F \cup J_u \cup J_v) }
A_2 (u,j) B_2 
\left(u, \pi(j) \right), \label{eq:W2_true_decomp}
\end{align}
where $F=\{j: \pi(j)=j\}$ corresponds to the set of correct seeds with $|F|=n\beta.$ 
Since the event $R_{uv}$ holds, 
it follows that $|J_u \cup J_v| \le 4(1+\epsilon)(n-1)ps+4\leq 6nps$,
where the last inequality holds due to $\epsilon\le 1/3$
and $nps \ge 6.$
Thus,
$n_{\mathrm{R}} \triangleq
|F \setminus (J_u \cup J_v)|
\ge n(\beta-6ps).
$

We first count the contribution to $W_2(u,u)$ by correct seeds. For each correct seed $j \in F \setminus (J_u \cup J_v)$, define an indicator variable
$
\chi_j = \indc{ \exists i \in C(u,u) \setminus \{v\}: j \in C(i,i) }.
$
In other words, $\chi_j=1$ if  $j$ is connected to  
some  ``common'' 1-hop neighbor of true pair $(u,u)$
in both $G_1$ and $G_2$, and $\chi_j=0$ otherwise.
By definition $A_2 (u,j) B_2 
\left(u,j\right) \ge \chi_j.$
Moreover,
%Now, we can bound $\prob{\chi_j=1| \Quv}$ as follows:
\begin{align*}
\prob{\chi_j=1 \mid \Quv}
= &1-\prob{ \cap_{\substack{ i \in C(u,u ) \backslash\{ v\}}}\left\{ j \notin C(i,i ) \right\} \mid \Quv}\\
\overset{(a)}{=}&1-\prod_{ \substack{ i \in C(u,u ) \backslash\{ v\}}} \prob{j \notin C(i,i) \mid \Quv}\\
%\overset{}{=}&1-\prod_{\substack{(w,\pi(w))\in C(u,\pi(u))\backslash\\\{(v,\pi(v))\}}}\left(1-\prob{(u_i,\pi(u_i))\in C(w,\pi(w))\mid\Quv}\right)\\
\overset{(b)}{=} &1- \left(1-ps^2\right)^{\abs{C(u,u)\setminus\{v\}}}\\
\overset{(c)}{\geq} 
%&1- \left(1-\frac{1}{2}\left(\abs{C(u,\pi(u))}-1\right)ps^2\right)\\{=} 
&\frac{1}{2}\left(c_{uu}-1\right)ps^2\overset{(d)}{\geq}  \frac{7}{24}np^2s^4,
\end{align*}
%We use the De Morgan's laws, which states $\overline{R\cup S}=\overline{R}\cap \overline{S}$, in the equality $(a)$. 
where $(a)$ holds because 
$\left\{ j \notin C(i,i ) \right\}=
\{A_1(i,j)=0\} \cup \{B_1(i,j)=0\}$,
which are independent across different $i$
conditional on $Q_{uv}$;
$(b)$ holds as 
$\prob{ j \in C(i,i ) }
=\prob{A_1(i,j)=B_1(i,j)=1}=ps^2$;
$(c)$ follows from \prettyref{thm:BernoulliInequality}
and the fact that $c_{uu}ps^2\leq(1+\epsilon)(n-1)p^2s^4\leq \frac{4}{3}np^2s^4<1$;
$(d)$ holds due to  $c_{uu}-1\geq(1-\epsilon)(n-1)ps^2-1\geq\frac{2}{3}(n-1)ps^2-1\geq\frac{7}{12}nps^2$.

Furthermore, note that $\chi_j$ depends on $A_1$ and $B_1$
only through the set of entries $S_j\triangleq \{ \{i,j\}: i \in C(u,u) \backslash \{v\}\}$. 
  Since $S_j \cap S_{j'} =\emptyset$ for all $j, j' \in F\setminus (J_u \cup J_v)$, it follows that $\chi_j$'s are mutually independent. Therefore,
\begin{align}
& \prob{ \sum_{j \in F \setminus (J_u \cup J_v) }
A_2 (u,j) B_2 
\left(u,j\right) \le l_{\min} \mid Q_{uv} } \nonumber  \\
 \le &\prob{ \sum_{j \in F \setminus (J_u \cup J_v) } \chi_j
\le l_{\min} \mid Q_{uv} } \nonumber \\
 \le &\prob{ \Binom\left( n_R,\frac{7}{24}np^2s^4 \right) \le l_{\min} \mid Q_{uv} }
\le n^{-\frac{15}{4}}, \label{eq:W2_true_correct}
\end{align}
where  $l_{\min}=\frac{7}{24}(\beta-6ps)n^2 p^2s^4  - \sqrt{\frac{35}{16}n^2\beta p^2s^4\log n}-\frac{5}{2}\log n$,
and the last inequality follows from Corollary \ref{lmm:Bernsteinbound} with $\gamma=\frac{15}{4}\log n$ and $ n(\beta-6ps) \le n_R \le n\beta$.

%If $C_i=1$,  then $(u_i,\pi(v_i))$ is a 2-hop witness for $(u,\pi(u))$. Thus, $\{C_i=1\}\subset\{L_i=1\}$, i.e., 
%\begin{equation}
%\begin{aligned}
%&\prob{L_i=1\mid \Quv}
%\geq \prob{C_i=1\mid \Quv}
%\geq \frac{7}{24}np^2s^4.\label{eq:lowboundLi}
%\end{aligned}
%\end{equation}

%Note that we have known the common neighbors $(w,\pi(w))\in C(u,\pi(u))$ conditional on $Q_{uv}$, and the edges $(u_i,w)$'s are different for different $u_i$ in the parent graph $G$, then $C_i$'s are mutual independent of each other conditional on $Q_{uv}$. Then, let $L$ denote the number of 2-hop witnesses contributed by correct seeds. Since $n_{\mathrm{R}}^+\in[(1-\delta_1)n\beta,n\beta]$, we have $L=\sum_{i=1}^{n_{\mathrm{R}}}L_i \geq \sum_{i=1}^{n^+_{\mathrm{R}}}C_i\overset{s.t.}{\geq}C'$, where $C'$ is a random variable such that $C'\sim\Binom(n_{\mathrm{R}}^+, \frac{7}{24}np^2s^4)$. Recall that $l_{\min}=\frac{7}{24}(1-\delta_1)n^2\beta p^2s^4  - \sqrt{\frac{35}{16}n^2\beta p^2s^4\log n}-\frac{5}{2}\log n$. We then get
%\begin{equation}
%\begin{aligned}
%\prob{L\leq l_{\min}\mid\Quv}\leq \prob{C'\leq l_{\min}\mid\Quv}
%\leq \exp\left(-\frac{15}{4}\log n\right)
%=n^{-\frac{15}{4}},\label{eq:lowboundcorrectseed}
%\end{aligned}
%\end{equation}
%where the last inequality follows from Corollary %\ref{lmm:Bernsteinbound} with $\gamma=\frac{15}{4}\log %n$.

Next, we count the contribution to $W_2(u,u)$ by the incorrect seeds. Fix an incorrect seed $(j, \pi(j))$ where $j \in [n]\backslash (F\cup J_u \cup J_v)$. Note that  $A_2(u,j)$ depends on $A_1$ through the set of entries given by $T_j \triangleq \{ \{i,j\}: i \in N^{G_1}(u) \}$ and $B_2(u,\pi(j))$
depends on $B_1$ through the set of entries given by $\tilde{T}_{\pi(j)} \triangleq \{ \{i,\pi(j) \}: i \in N^{G_2}(u) \}$. 
Thus $A_2(u,j)$ and $B_2(u,\pi(j) )$ are independent 
when $T_j \cap \tilde{T}_{\pi(j)} = \emptyset$, which occurs if and only $j \notin N^{G_2}(u)  $ or 
$\pi(j) \notin N^{G_1}(u)$.
Thus to ensure the independence between  $A_2(u,j)$ and $B_2(u,\pi(j) )$  in order to facilitate computing the probability of
$A_2(u,j)B_2(u,\pi(j) )=1$, 
we also exclude the set of seeds 
given by $\tilde{J}_u= N^{G_2}(u) \cap \pi^{-1}
\left( N^{G_1}(u) \right)$. 
Let $n_W\triangleq | [n] \setminus (F \cup J_u \cup J_v \cup \tilde{J}_u) |$. Since the event $R_{uv}$ holds, it follows that $n_W \ge n(1-\beta)-9nps.$
Now, for each $j \in [n] \setminus (F \cup J_u \cup J_v \cup \tilde{J}_u)$, 
we have
\begin{align*}
&\prob{A_2(u,j)B_2(u,\pi(j) )=1 \mid Q_{uv} } \\ =&\prob{ 
A_2(u,j) =1 \mid \Quv}\times \prob{ B_2(u,\pi(j) )=1 \mid \Quv}\\
=&\left(1-\prob{ 
A_1(i,j)=0, \forall i \in N^{G_1}(u)
\mid  \Quv}\right)
\left(1-\prob{ B_1(i,\pi(j) ) =0, \forall 
i \in N^{G_2}(u) \mid \Quv}\right)\\
= &\left(1-(1-ps)^{a_{u\backslash v}}\right)\left(1-(1-ps)^{b_{u\backslash v}}\right) \triangleq \lambda,
\end{align*}
where the last equality holds because 
$A_1(i,j)=0$
if $i=v$ as $j \notin J_v$;
otherwise $A_1(i,j) \sim \Bern(ps)$;
and similarly for $B_1(i, \pi(j)).$

Finally, note that $A_2(u,j)B_2(u,\pi(j) )$ are dependent across different $j$ and thus
we cannot directly apply Bernstein's inequality. 
To see this, observe that
conditional on $Q_{uv}$, $A_2(u,j)B_2(u,\pi(j) )$ depends on $A_1$ and $B_1$ through the set of entries given by $T_j \cup \tilde{T}_{\pi(j)} \triangleq U_j$.
Therefore, for any pair of $j, j' \in [n] \setminus (F \cup J_u \cup J_v \cup \tilde{J}_u) $ with $j\neq j'$, 
 $A_2(u,j)B_2(u,\pi(j) )$  
 and 
 $A_2(u,j')B_2(u,\pi(j') )$ 
are dependent if and only if 
$U_j \cap U_{j'} \neq \emptyset$,
which occurs if and only if
$j'=\pi(j)$ or $j'=\pi^{-1}(j)$. 
Hence, we construct a dependency graph $\Gamma$ for $\{A_2(u,j)B_2(u,\pi(j) )\}$, where
the maximum degree  $\Delta(\Gamma)$ equals to $2$. Thus,  applying \prettyref{thm:drv} with  %$\Delta_1(\Gamma)= \Delta(\Gamma)+1=3$, 
$K=1$,
$\sigma^2=n_{\mathrm{W}}\lambda(1-\lambda)$,
and $\gamma=4\log n$ yields that 
\begin{align*}
\prob{\sum_{ j \in [n] \setminus (F \cup J_u \cup J_v \cup \tilde{J}_u) }  
A_2(u,j) B_2(u,\pi(j) )  \leq n_{\mathrm{W}} \lambda -5\sqrt{\frac{3}{2}n_{\mathrm{W}}\lambda(1-\lambda)\log n}-\frac{25}{2}\log n \mid \Quv}
\leq n^{-4}.
\end{align*}

Since $ n(1-\beta-9ps) \le n_{\mathrm{W}} \le n$, we have
\begin{align*}
n_{\mathrm{W}}\lambda(1-\lambda)
\overset{}{\leq} & n\left(1-(1-ps)^{a_u}\right)(1-(1-ps)^{b_u})
\overset{(a)}{\leq} na_u b_u p^2s^2
\overset{(b)}{\leq} \frac{9}{4}n^3p^4s^4,
\end{align*}
where $(a)$ holds as $(1+x)^r\geq 1+rx$ for every integer $r\geq 0$ and every real number $x\geq -2$; 
$(b)$ holds because $a_u,b_u\leq(1+\epsilon)(n-1)ps\leq\frac{3}{2}nps$ under event $R_{uv}.$ And
$$
n_{\mathrm{W}}\lambda\geq n(1-\beta)\lambda-9nps\lambda\geq n(1-\beta)\lambda-21n^3p^5s^5.
$$
Therefore, recalling that 
$m_{\min} =n(1-\beta)\lambda-21n^3p^5s^5- \frac{15}{2}\sqrt{\frac{3}{2}n^3p^4s^4\log n}-\frac{25}{2}\log n$,
we get that
\begin{align*}
m_{\min} \le n_{\mathrm{W}} \lambda -5\sqrt{\frac{3}{2}n_{\mathrm{W}}\lambda(1-\lambda)\log n}-\frac{25}{2}\log n.
\end{align*}
It follows that 
\begin{align}
\prob{\sum_{ j \in [n] \setminus (F \cup J_u \cup J_v \cup \tilde{J}_u) }  
A_2(u,j) B_2(u,\pi(j) )  \leq m_{\min} \mid \Quv}
\leq n^{-4}.\label{eq:W2_true_incorrect}
\end{align}

% Let $M$ denote the number of 2-hop witnesses contributed by the incorrect seeds. We have $M\geq \sum_1^{n_{\mathrm{W}}^+}M_i$. Recall that $m_{\min}=(1-\delta_2)n(1-\beta)\lambda- \frac{15}{2}\sqrt{\frac{3}{2}n^3p^4s^4\log n}-\frac{25}{2}\log n$. We then get
%  \begin{equation}
% \begin{aligned}
% &\prob{M\leq m_{\min}\mid\Quv}\\
% \leq& \prob{\sum_1^{n_{\mathrm{W}}^+}M_i\leq n_{\mathrm{W}}^+\lambda -5\sqrt{\frac{3}{2}n_{\mathrm{W}}^+\lambda(1-\lambda)\log n}-\frac{25}{2}\log n\mid\Quv}
% \leq n^{-4}.\label{eq:lowboundincorrectseed}
% \end{aligned}
% \end{equation}

%Due to ${W_2(u,\pi(u))}= L+M$ and sufficiently large $n$, taking an union bound over (\ref{eq:lowboundcorrectseed}) and (\ref{eq:lowboundincorrectseed}) yields that
Combining \prettyref{eq:W2_true_decomp},
\prettyref{eq:W2_true_correct}, \prettyref{eq:W2_true_incorrect}
with a union bound, we get that 
\begin{align*}
\prob{{W_2(u,u)}\leq l_{\min}+m_{\min}\mid\Quv}\cdot\Indicator{R_{uv}}\leq n^{-\frac{15}{4}}+n^{-4}< n^{-\frac{7}{2}}.
\end{align*}

\begin{remark}
In \prettyref{eq:W2_true_correct}, 
we bound $A_2(u,j) B_2(u,j)$
from below by $\chi_j$, 
by neglecting the case that $j$ is connected to different 1-hop neighbors of $u$ in $G_1$ and $G_2$.
This lower bound is relatively tight, because
\begin{align*}
\prob{A_2(u,j) B_2(u,j) =1 , \chi_j =0 \mid \Quv}& \approx \prob{ A_2(u,j) =1 \mid \Quv}\prob{ B_2(u,j)=1 \mid \Quv}\\
& \approx  a_ub_up^2s^2\leq \frac{9}{4}n^2p^4s^4,
\end{align*}
which is much smaller than 
$\prob{ \chi_j=1 \mid \Quv}$ when $np^2 \le \frac{1}{\log n}.$
%We can observe $\prob{C_i=1\mid \Quv}\gg \prob{\{L_i=1\}\setminus\{ C_i=1\}\mid \Quv}$ because $\nps$, and $n$ is sufficiently large. Thus, we can say $\frac{7}{24}np^2s^4$ is a tight lower bound of $\prob{L_i=1\mid \Quv}$.
\end{remark}

\subsection{Proof of \prettyref{lmm:W2uv}}\label{sec:proofW2uv}
Fixing any two vertices $u\neq v$, we condition on $Q_{uv}$ such that event $R_{uv}$ holds. Note that 
$$
W_2(u,v) = \sum_{j=1}^n A_2 (u,j) B_2 
\left(v, \pi(j) \right),
$$
where $A_2$ and $B_2$ are the $2$-hop adjacency matrix of $G_1$ and $G_2$, respectively. 
%Note that for all $j \in N^{G_1}(u)\cup \{u\}$,
%$A_2(u,j)=0$ by definition. Similarly,
%for all $\pi(j) \in 
%N^{G_2}(v)\cup \{v\}$, $B_2\left(v, \pi(j) \right)=0$. Thus, we define 
Let
$$
J_0= N^{G_1}(u)\cup \{u\} \cup \pi^{-1}
\left( N^{G_2}(v) \right) \cup \pi^{-1} (v)
$$
Then $A_2(u,j)B_2\left(v, \pi(j) \right)=0$ for all $j \in J_0$.
Thus, 
$$
W_2(u,v)  = \sum_{j\in [n]\setminus J_0} 
A_2 (u,j) B_2 
\left(v, \pi(j) \right).
$$

Note that we have conditioned on the $1$-hop neighborhoods of $u$ and $v$ in $G_1$ and $G_2$.
In either $G_1$ or $G_2$, if $u$ and $v$ are connected, then a 1-hop neighbor of $u$ (or $v$) may automatically become the 2-hop neighbor of $v$ (or $u$). 
Hence,  if $j$ is connected to $v$ in $G_1$ or $\pi(j)$ is connected to $u$ in $G_2$, then conditioning on $Q_{uv}$ can change the probability that $A_2(u,j)B_2 
\left(v, \pi(j) \right)=1$.
To circumvent this issue,  we further divide the remaining seeds into five types depending on whether $j\in N^{G_1}(v)\cup \{v\}$ and $\pi(j)\in N^{G_2}(u)\cup \{u\}$, and get 
\begin{align}
W_2(u,v) = \sum_{k=1}^5 \sum_{j \in J_k}
A_2 (u,j) B_2 
\left(v, \pi(j) \right).  \label{eq:W_2_decomp}
\end{align}
Let $X_k = \sum_{j \in J_k}
A_2 (u,j) B_2 
\left(v, \pi(j) \right)$ 
denote the contribution from type $k$.
In the sequel, 
we will separately bound $X_k$ from the above for each $k \in [5]$. 

%(see \prettyref{fig:SeedsCategory} for example): 

\paragraph{Type 1:} $ J_1\triangleq\{v,\pi^{-1}(u)\}\setminus J_0$. 
We have $X_1 \le \abs{J_1}\le 2$.

\paragraph{Type 2:}  $J_2\triangleq N^{G_1}(v)\cap \pi^{-1}\left(N^{G_2}(u)\right)\setminus J_0$. For $j\in J_2$, since $A_1(v,j)=1$ and $B_1(u,\pi(j))=1$, it follows that $(j,\pi(j))$ is a 1-hop witness for $(v,u)$. Thus, we have $X_2 \le \abs{J_2}\le W_1(v,u)\le \psi_{\max}$ on event $R_{uv}$. 
%If $u$ and $v$ are connected in both $G_1$ and $G_2$, then $(j,\pi(j))$ would be a 2-hop witness for $(u,v)$.\\

\paragraph{Type 3:} $J_3\triangleq N^{G_1}(v)\setminus \left( \pi^{-1}\left(N^{G_2}(u)\right)\cup\{\pi^{-1}(u)\}\cup J_0\right)$.
We have $\abs{J_3}\le a_v\le\frac{3}{2}nps$ on event $R_{uv}$. 
%If $u$ and $v$ are connected in $G_1$, then as $v \in N^{G_1}(u)$, it follows that $A_2(u,j)=1$. 
By definition, $A_2(u,j)B_2(v,\pi(j))\le B_2(v,\pi(j))$. Moreover, 
\begin{align} \label{eq:W2-fake-connect1hop} 
\prob{B_2(v,\pi(j))=1\mid\Quv}
= &\prob{B_1(i,\pi(j))=1,\exists i\in N^{G_2}(v) \mid\Quv}\nonumber\\
=&1-\prob{ B_1(i,\pi(j)) =0, \forall 
i \in N^{G_2}(v) \mid \Quv}\nonumber\\
\overset{(a)}{=}&1-(1-ps)^{b_{v\backslash u}}
\overset{(b)}{\le} b_vps
\overset{(c)}{\leq}  \frac{3}{2}np^2s^2,
\end{align}
where $(a)$ holds because $B_1(i,\pi(j))=0$ if $i=u$ as $\pi(j)\notin N^{G_2}(u)$;
otherwise $B_1(i,\pi(j)) \iiddistr \Bern(ps)$ across different $i$;
$(b)$ follows from $(1+x)^r\geq 1+rx$ for every integer $r \geq 0$ and every real number $x \geq -2$; $(c)$ holds due to  $b_v<\frac{3}{2}nps$ on event $R_{uv}$.

Note that $B_2(v,\pi(j))$ only depends on $B_1$ through the set of entries $U_{\pi(j)}\triangleq \{ \{i,\pi(j)\}: i \in N^{G_2}(v)\}$. Since  $U_{\pi(j)}\cap U_{\pi(j')}=\emptyset$ for all $j \neq j'\notin J_0$, it follows that  $B_2(v,\pi(j))$ are mutually independent across $j \in J_3$. Therefore,
\begin{align}\label{eq:W2_fake_J3}
& \prob{ X_3 \ge \frac{9}{2}n^2p^3s^3 +5\log n \mid Q_{uv} } \nonumber  \\
& \le \prob{ \sum_{j \in J_3 } B_2(v,\pi(j))
\ge  \frac{9}{2}n^2p^3s^3 +5\log n \mid Q_{uv} } \nonumber \\
& 
%\overset{(a)}{\le}
\le \prob{ \Binom\left( \left|J_3 \right|,\frac{3}{2}np^2s^2 \right) \ge \frac{9}{2}n^2p^3s^3 +5\log n
%\lfloor\frac{3}{2}nps\rfloor\frac{3}{2}np^2s^2+\sqrt{ \lfloor\frac{3}{2}nps\rfloor\frac{45}{4}np^2s^2\log n} +\frac{5}{2}\log n
\mid Q_{uv} }\nonumber\\
%&\overset{(b)}{\le}
& \le n^{-\frac{15}{4}}, 
\end{align}
%where $(a)$ follows from the AM-GM inequality, which states that $2\sqrt{xy}\leq x+y$  for two non-negative numbers $x$ and $y$;
%$(b)$ 
where the last inequality follows from Corollary \ref{lmm:Bernsteinbound} with $\gamma=\frac{15}{4}\log n$ and $|J_3| \le \frac{3}{2} nps$.

\paragraph{Type 4:} $J_4\triangleq \pi^{-1}\left(N^{G_2}(u)\right)\setminus\left( N^{G_1}(v) \cup \{v\}\cup J_0\right)$. Following the similar proof as in Type 3, we can get 
\begin{align}
& \prob{ X_4 \ge \frac{9}{2}n^2p^3s^3 +5\log n \mid Q_{uv} } \overset{}{\le}n^{-\frac{15}{4}}. \label{eq:W2_fake_J4}
\end{align}

\paragraph{Type 5:} $j\in J_5\triangleq [n]\setminus\left( \cup_{k=0}^4 J_k \right)$. 
%In particular, $J_0\cup J_1\cup J_2 \cup J_3 \cup J_4=\{u,v,\pi^{-1}(u),\pi^{-1}(u)\}\cup N^{G_1}(u)\cup  N^{G_1}(v)\cup \pi^{-1}\left( N^{G_2}(u)\right)\cup \pi^{-1}\left( N^{G_2}(v)\right)$
This is the major type. 
We bound $X_5$ by separately considering the correct and incorrect seeds. 

\paragraph{Correct Seeds in Type 5:} Recall that $F=\{j: \pi(j)=j\}$
corresponds to the set of correct seeds. We have $\abs{F\cap J_5}\le \abs{F}=n\beta$. Note that  $A_2(u,j)$ depends on $A_1$ through the set of entries given by $T_j \triangleq \{ \{i,j\}: i \in N^{G_1}(u) \}$ and $B_2(v,j)$
depends on $B_1$ through the set of entries given by $\tilde{T}_{j} \triangleq \{ \{i,j \}: i \in N^{G_2}(v) \}$. Thus $A_2(u,j)$ and $B_2(v,j )$ are dependent on each other because $T_j \cap \tilde{T}_{j} =  \{ \{i,j \}: i \in C(u,v) \}\neq \emptyset.$  Thus, we bound $\prob{A_2(u,j)B_2(v,j)=1}$ by separately 
considering whether $j$ is connected to some vertices in $C(u,v)$. 
%If $j$ is connected to some vertices in $C(u,v)$ in $G_1$, we have
Specifically, on the one hand, 
 \begin{align}\label{eq:pf_jC1uv_1}
    &\prob{\{A_2(u,j)B_2(v,j)=1\}\cap\{A_1(i,j)=1,\exists i\in C(u,v)\}\mid \Quv}\nonumber\\
    \leq& \prob{A_1(i,j)=1,\exists i\in C(u,v)\}\mid \Quv}\nonumber\\
    =&1-\prob{A_1(i,j)=0,\forall i\in C(u,v)\}\mid \Quv}\nonumber\\
    \overset{(a)}{=}&1- (1-ps)^{c_{uv}}\nonumber\\
\overset{(b)}{\leq}&1- (1-c_{uv}ps)
\overset{(c)}{\leq}  \psi_{\max} ps,
\end{align} 
where $(a)$ holds because  $A_1(i,j) \iiddistr \Bern(ps)$;
$(b)$ follows from $(1+x)^r\geq 1+rx$ for every integer $r \geq 0$ and every real number $x \geq -2$; $(c)$ holds due to $c_{uv}<\psi_{\max}$ on event $R_{uv}$.

%If $j$ is not connected to any vertex in $C(u,v)$ in $G_1$, then $A_2(u,j)=1$ implies that $j$ is connected to some vertices in $N^{G_1}(u)\setminus C(u,v)$. 
On the other hand, letting 
$\calA_u\triangleq \{\exists i\in N^{G_1}(u)\setminus C(u,v): A_1(i,j)=1\}$, 
\begin{align}\label{eq:pf_jC1uv_0}
    &\prob{\{A_2(u,j)B_2(v,j)=1\}\cap\{A_1(i,j)=0,\forall i\in C(u,v)\}\mid \Quv}\nonumber\\
    =&\prob{\calA_u\cap \{B_2(v,j)=1\}\mid \Quv}\nonumber\\
    \overset{(a)}{=}&\prob{\calA_u\mid \Quv}\times\prob{ B_2(v,j)=1\mid \Quv}\nonumber\\
    \leq &\prob{A_1(i,j)=1,\exists i\in N^{G_1}(u)\mid \Quv}\times\prob{ B_1(i,j)=1,\exists i\in N^{G_2}(v)\mid \Quv} \overset{(b)}{\le}\frac{9}{4}n^2p^4s^4,
\end{align}
where the equality $(a)$ holds as $\calA_u$ and $\{B_2(v,j)=1\}$ are independent. This is because $\calA_u$ depends on $T_j'\triangleq\{\{i,j\}:i\in N^{G_1}(u)\setminus C(u,v)\}$, which is disjoint from $\tilde{T}_j=\{\{i,j\}:i\in N^{G_2}(v)\}$; (b) follows from the similar reasoning as in \prettyref{eq:W2-fake-connect1hop}.

Thus, by taking the union bound over \prettyref{eq:pf_jC1uv_1} and \prettyref{eq:pf_jC1uv_0}, we have
\begin{align*}
\prob{A_2(u,j)B_2(v,j)=1\mid \Quv}\leq \psi_{\max} ps+\frac{9}{4}n^2p^4s^4\triangleq\mu_1.
\end{align*}

Note that $A_2(u,j)B_2(v,j)$ only depends on $A_1$ and $B_1$ only through the set of entries $T_j\cup \tilde{T}_{j}\triangleq U_{j}$. Since  $U_j\cap U_{j'}=\emptyset$ for all $j \neq j'\notin J_0$, it follows that $A_2(u,j)B_2(v,j)$ are mutually independent for all $j \in F\cap J_5$. Therefore,
\begin{align}\label{eq:W2_fake_J5}
 &\prob{ \sum_{j \in F\cap J_5 }
A_2(u,j)B_2(v,j) \ge x_{\max} +5\log n \mid Q_{uv} } \nonumber\\  \le & \prob{ \Binom (n\beta, \mu_1)
\ge x_{\max} +5\log n
%n\beta\mu_1+\sqrt{\frac{15}{2}n\beta\mu_1\log n} +\frac{5}{2}\log n 
\mid Q_{uv} }
\le n^{-\frac{15}{4}}, 
\end{align}
where $x_{\max}=2 n\beta\left( \psi_{\max}ps+\frac{9}{4}n^2p^4s^4\right)$
and the last inequality follows from Corollary \ref{lmm:Bernsteinbound} with $\gamma=\frac{15}{4}\log n$.

\paragraph{Incorrect Seeds in Type 5:} Let $\overline{F}\triangleq[n]\setminus F$ denote the complement of $F$ in $[n]$. Then, $\overline{F}$ corresponds to the set of incorrect seeds with $\abs{\overline{F}}=n(1-\beta)$. Note that  $A_2(u,j)$ depends on $A_1$ through the set of entries given by $T_j \triangleq \{ \{i,j\}: i \in N^{G_1}(u) \}$ and $B_2(v,\pi(j))$
depends on $B_1$ through the set of entries given by $\tilde{T}_{\pi(j)} \triangleq \{ \{i,\pi(j) \}: i \in N^{G_2}(v) \}$. 
Thus $A_2(u,j)$ and $B_2(v,\pi(j) )$ are independent 
when $T_j \cap \tilde{T}_{\pi(j)} = \emptyset$, which occurs if and only if $j \notin N^{G_2}(v)  $ or $\pi(j) \notin N^{G_1}(u)$. We define $\tilde{J}=N^{G_2}(v)\cap \pi^{-1}\left(N^{G_1}(u)\right)$, and have $\abs{\tilde{J}}\leq b_v\leq \frac{3}{2}nps$ under the event $R_{uv}$. Then, we separately consider the incorrect seeds depending on whether $ j \in \tilde{J}$. 

\begin{itemize}
    \item For $j\in \overline{F}\cap J_5 \setminus\tilde{J}$,
\begin{align*}
\mu_2 \triangleq &\prob{A_2(u,j)B_2(v,\pi(j) )=1\mid \Quv}\\
\overset{}{=}&\prob{A_2(u,j)=1 \mid \Quv}\times\prob{B_2(v,\pi(j) )=1\mid\Quv}\nonumber\\
=&\left(1-(1-ps)^{a_{u\backslash v}}\right)\left(1-(1-ps)^{b_{v\backslash u}}\right)\overset{}{\leq} \frac{9}{4}n^2p^4s^4,
\end{align*}
where  the last two steps follow from the similar reasoning as in \prettyref{eq:W2-fake-connect1hop}.
\item 
For  $j \in \overline{F} \cap J_5 \cap \tilde{J}$, we divide the analysis into two cases depending on whether $A_1(j,\pi(j))=1$. On the one hand,
\begin{align*}
\prob{\{A_2(u,j)B_2(v,\pi(j) =1\}\cap\left\{A_1(j,\pi(j))=1\right\}\mid\Quv}\le\prob{\left\{A_1(j,\pi(j))=1\right\}\mid\Quv}\le ps.
\end{align*}

%If $A_1(j,\pi(j))=0$, then $A_2(u,j)=1$ implies $j$ is connected to some vertices in $N^{G_1}(u)\setminus \{\pi(j)\}$. We define $A_u\triangleq \{A_1(i,j)=1,\exists i\in N^{G_1}(u)\setminus \{\pi(j)\}\}$, and have
On the other hand, letting 
$\calA'_u \triangleq \{\exists i\in N^{G_1}(u)\setminus \{\pi(j)\}: A_1(i,j)=1\}$,
\begin{align*}
&\prob{\{A_2(u,j)B_2(v,\pi(j) =1\}\cap\left\{A_1(j,\pi(j))=0\right\}\mid\Quv}\\
&=\prob{\calA'_u\cap \{B_2(v,j)=1\}\mid\Quv}\\
&\overset{(a)}{=}\prob{\calA'_u\mid\Quv}\times\prob{\{B_2(v,j)=1\}\mid\Quv}\\
&\leq \prob{A_1(i,j)=1,\exists i\in N^{G_1}(u)\mid \Quv}\times\prob{ B_1(i,j)=1,\exists i\in N^{G_2}(v)\mid \Quv}\overset{(b)}{\le} \frac{9}{4}n^2p^4s^4,
\end{align*}
where the equality $(a)$ holds as $A_u$ and $\{B_2(v,\pi(j))=1\}$ are independent. This is because $\calA_u$ depends on $T_j'\triangleq\{\{i,j\}:i\in N^{G_1}(u)\setminus\{\pi(j)\}\}$, which is disjoint with $\tilde{T}_{\pi(j)}=\{\{i,\pi(j)\}:i\in N^{G_2}(v)\}$; $(b)$ follow from the similar proof in \prettyref{eq:W2-fake-connect1hop}.

Combining the last two displayed equations yields that 
$$
\mu_3 \triangleq \prob{A_2(u,j)B_2(v,\pi(j) =1\mid\Quv}
\le  ps+\frac{9}{4}n^2p^4s^4\leq \frac{2}{3}np^2s^2.
$$
\end{itemize}

%Let $Y$ denote the number of 2-hop witnesses contributed by the incorrect seeds $(j,\pi(j))$ where $j\in [n]\setminus (F\cup  J \cup \tilde{J}_1)$. We have $Y\leq \sum_{i=1}^{n(1-\beta)}Y_i$.  

Note that $A_2(u,j)B_2(v,\pi(j) )$ are dependent across different $j\in \overline{F}\cap J_5$ 
%\nbr{did you define $\overline{F}$ as the complement of $F$ somewhere? LY: I add the definition at the beginning of the incorrect seeds in Type 5} 
and thus
we cannot directly apply Bernstein's inequality. 
To see this, observe that
conditional on $Q_{uv}$, $A_2(u,j)B_2(v,\pi(j) )$ depends on $A_1$ and $B_1$ through the set of entries given by $T_j \cup \tilde{T}_{\pi(j)} \triangleq U_j$.
Therefore, for any pair of $j, j' \in \overline{F} \cap J_5$ with $j\neq j'$, 
 $A_2(u,j)B_2(v,\pi(j) )$  
 and 
 $A_2(u,j')B_2(v,\pi(j') )$ 
are dependent if and only if 
$U_j \cap U_{j'} \neq \emptyset$,
which occurs if and only if
$j'=\pi(j)$ or $j'=\pi^{-1}(j)$. 
Hence, we construct a dependency graph $\Gamma$ for $\{A_2(u,j)B_2(v,\pi(j) )\}$, where
the maximum degree  $\Delta(\Gamma)$ equals to $2$. Thus,  applying \prettyref{thm:drv} with  %$\Delta_1(\Gamma)= \Delta(\Gamma)+1=3$, 
$K=1$,
$\sigma^2=\abs{\overline{F}\cap J_5 \setminus\tilde{J}}\mu_2(1-\mu_2)+\abs{\overline{F}\cap J_5 \cap\tilde{J}}\mu_3(1-\mu_3)$,
and $\gamma=4\log n$ yields that 
\begin{align}\label{eq:W2-fake-J5-incorrect}
&\prob{\sum_{ j \in \overline{F}\cap J_5 }  A_2(u,j) B_2(v,\pi(j) )  \geq y_{\max}+\frac{25}{2}\log n\mid \Quv}\le n^{-4},
%\overset{(a)}{\le}&\prob{\sum_{j\in F^c\cap J_5 \setminus \tilde{J}} A_2(u,j) B_2(v,\pi(j) ) \ge y_{\max}+\frac{25}{2}\log n\mid \Quv}\nonumber\\
%&+\prob{\sum_{j\in F^c\cap J_5 \cap \tilde{J}} A_2(u,j) B_2(v,\pi(j) ) \ge n^2p^3s^3\left(\frac{3}{nps}+\frac{27}{4}np^2s^2\right) +25\log n\mid \Quv}\nonumber\\
%&\overset{(b)}{\le} n^{-4}+\prob{\sum_{j\in F^c\cap J_5 \cap \tilde{J}} A_2(u,j) B_2(v,\pi(j) ) \ge  \lfloor\frac{3}{2}nps\rfloor\mu_2+5\sqrt{\frac{3}{2}\lfloor\frac{3}{2}nps\rfloor\mu_2\log n}+\frac{25}{2}\log n\mid \Quv}\nonumber\\
%\leq& 2\cdot n^{-4},
\end{align}
where $y_{\max}=n(1-\beta)\mu_2+n^2p^3s^3+\frac{5}{2}\sqrt{15n^3p^4s^4\log n}$.
\\

%Combining all types of seeds, we have
%$$
%\sum_{j\in [n]}A_2(u,j) B_2(v,\pi(j) )\le \sum_{r=0}^{5}\sum_{j\in J_r}A_2(u,j) B_2(v,\pi(j) )\le 0+2+3\log n+ \sum_{r=3}^{5}\sum_{j\in J_r}A_2(u,j) B_2(v,\pi(j) ).
%$$

Finally, combining all types of seeds and taking an union bound on \prettyref{eq:W2_fake_J3}, \prettyref{eq:W2_fake_J4},  \prettyref{eq:W2_fake_J5} and \prettyref{eq:W2-fake-J5-incorrect}, we get
\begin{align*}
    \prob{{W_2(u,v)}\ge x_{\max}+y_{\max}+2z_{\max}+\psi_{\max}+28\log n\mid\Quv}\cdot\Indicator{R_{uv}}\le 3\cdot n^{-\frac{15}{4}}+ n^{-4}< n^{-\frac{7}{2}},
\end{align*}
where $z_{\max}=\frac{9}{2}n^2p^3s^3$.

\subsection{Proof of \prettyref{lmm:Tuv}}\label{sec:proofTuv}
%\nbr{I suggest define $d_u=\left|N^{G_0}(u)\right|$ locally in this proof to simplify the exposition. LY: I have defined it together with $a_u, b_u$ in the main part.}
Recall that $d_u=\abs{N^{G_0}(u)} \sim \Binom(n-1,p)$.
%For any vertex $i \in  [n] \setminus \{u\}$, $i$ and $u$ are connected with probability $p$ in $G_0$. Then, $d_u \sim\Binom (n-1,p)$. Since $p\geq ps^2$ 
%and $\epsilon=\sqrt{\frac{12\log n}{(n-1)ps^2}}\leq \frac{1}{3}$, 
In view of assumption $nps^2 \ge 128\log n$, applying \prettyref{lmm:bound} gives that
\begin{align}\label{eq:RuG}
\prob{d_u\geq \frac{4}{3}(n-1)p}
%\leq\prob{d_u\geq (1+\epsilon)(n-1)p}
\leq n^{-4}.
\end{align}
Let $R_u$  denote the event $\left\{d_u<\frac{4}{3}(n-1)p\right\}$.
%\nbr{Is there any particular reason that you keep the superscript $G$ here. If not, pelase delete it. LY: addressed} 

For any two vertices $u,v\in[n]$ with $u\neq v$, let $E_{uv}$ denote
\begin{align*}
E_{uv}=\left\{N^{G_0}(u),N^{G_0}(v)\right\}.
\end{align*}
Conditioning on $E_{uv}$ such that $R_u$ and $R_v$ are true, we separately consider two cases: $d_u\leq d_v$ and $d_u>d_v$. 

\paragraph{Case 1:} $d_u\leq d_v$. 
By definition, we have $a_u\sim\Binom\left(d_u,s\right)$ and $a_v\sim\Binom\left(d_v,s\right)$. Applying with  Bernstein’s inequality given in \prettyref{thm:bernstein} with  $\gamma=\frac{15}{4}\log n$ and $K=1$
 implies that
\begin{align}
   \prob{a_u\geq d_us+\sqrt{\frac{15}{2}d_us(1-s)\log n}+\frac{5}{2}\log n\mid E_{uv}} & \leq n^{-\frac{15}{4}},\label{eq:Tuv-N1u} \\
\prob{a_v\leq d_vs-\sqrt{\frac{15}{2}d_vs(1-s)\log n}-\frac{5}{2}\log n\mid E_{uv}} & \leq n^{-\frac{15}{4}}.
\label{eq:Tuv-N1v}
\end{align}

Under the events $R_u$ and $R_v$, we have
\begin{align*}
  &d_us+\sqrt{\frac{15}{2}d_us(1-s)\log n}+\frac{5}{2}\log n-\left(d_vs-\sqrt{\frac{15}{2}d_vs(1-s)\log n}-\frac{5}{2}\log n\right)\\
& \le \sqrt{\frac{15}{2}d_us(1-s)\log n}+\sqrt{\frac{15}{2}d_vs(1-s)\log n}+5\log n\\
& \le 2\sqrt{10nps(1-s)}+5\log n=\tau.
\end{align*}
Taking the union bound over \prettyref{eq:Tuv-N1u} and \prettyref{eq:Tuv-N1v}, and noting that $\overline{T_{uv}}\subset\left\{a_u-a_v\geq \tau\right\}$,  we have
$$
\prob{\overline{T_{uv}}\mid  E_{uv}} \le \prob{a_u-a_v\geq \tau\mid  E_{uv}}
\leq n^{-\frac{15}{4}}.
$$

\paragraph{Case 2:} $d_u> d_v$.
Following the similar proof, we can get
\begin{align*}
\prob{\overline{T_{uv}}\mid  E_{uv} }
\leq \prob{b_v-b_u\geq \tau\mid E_{uv}}
\leq n^{-\frac{15}{4}}.
\end{align*}
Combining the two cases gives that 
\begin{align}
\prob{\overline{T_{uv}}\mid E_{uv}}\cdot\Indicator{R_u\cap R_v} 
\leq& n^{-\frac{15}{4}}.\label{eq:TuvC}
\end{align}

Finally, since $n$ is sufficiently large, applying \prettyref{eq:RuG}, \prettyref{eq:TuvC} and the union bound yields
\begin{align*}
\prob{\overline{T_{uv}}} 
=&\mathbb{E}_{E_{uv}}\left[\prob{\overline{T_{uv}}\mid E_{uv}}\right]\\
=&\mathbb{E}_{E_{uv}}\left[\prob{\overline{T_{uv}}\mid E_{uv}}\cdot\Indicator{R_u\cap R_v} +\prob{\overline{T_{uv}}\mid E_{uv}}\cdot\Indicator{\overline{R_u\cap R_v}}\right]\\
\leq& \mathbb{E}_{E_{uv}}\left[\prob{\overline{T_{uv}}\mid E_{uv}}\cdot\Indicator{R_u\cap R_v} \right]+\mathbb{E}_{E_{uv}}\left[\Indicator{\overline{R_u\cap R_v})} \right]\\
\leq & n^{-\frac{15}{4}}+ 2\cdot n^{-4}\leq n^{-\frac{7}{2}}.
\end{align*}

\subsection{Proof of \prettyref{lmm:wmingwmax}}\label{sec:proofwmingwmax}
%\nbr{Please update the proof using the new notation $a_{u\backslash v}. LY: addressed$}
%Since the bound of 2-hop witnesses if provided by \prettyref{lmm:W2uu} and $\prettyref{lmm:W2uv}$, it
%remains to 
% We show $l_{\min}+m_{\min}-x_{\max}-y_{\max}-2z_{\max}-31\log n\geq 0$ under conditions of \prettyref{thm:thm2hop}. 
Since $\psi_{\max}=np^2s^2+\sqrt{7np^2s^2\log n}+\frac{7}{3}\log n+2\leq 3\log n$ due to $\nps$, we have
 \begin{align}
 \label{eq:lowerbound-upperbound}
 &l_{\min}+m_{\min}-x_{\max}-y_{\max}-2z_{\max}-\psi_{\max}-28\log n\nonumber\\
 & \geq \frac{7}{24}n^2\beta p^2s^4 -\frac{7}{4} n^2p^3s^5-21n^3p^5s^5- \sqrt{\frac{35}{16}n^2\beta p^2s^4\log n}- \frac{5}{2}\sqrt{15n^3p^4s^4\log n}\nonumber\\
 &-\frac{15}{2}\sqrt{\frac{3}{2}n^3p^4s^4\log n}-2n\beta\left(3 ps^2\log n+\frac{9}{4}n^2p^4s^4\right)-10n^2p^3s^3-46\log n\nonumber\\
 &+n(1-\beta)\left(1-(1-ps)^{a_{u\backslash v}}\right)\left((1-ps)^{b_{v\backslash u}}-(1-ps)^{b_{u\backslash v}}\right).
 \end{align}
To bound from below the last term in \prettyref{eq:lowerbound-upperbound},
we have
\begin{align*}
&\left(1-(1-ps)^{a_{u\backslash v}}\right)\left((1-ps)^{b_{v\backslash u}}-(1-ps)^{b_{u\backslash v}}\right)\nonumber\\
%\geq&\left(1-(1-ps)^{a_{u\backslash v}}\right)\left((1-ps)^{b_{v}-b_{u}}-1\right)\\
& \overset{(a)}{\geq} \left(1-(1-ps)^{a_u}\right)\left((1-ps)^{\tau}-1\right)\\
%\overset{(b)}{\geq} &-\tau nps\left(1-(1-ps)^{a_u}\right)\\
& \overset{(b)}{\geq} - {a_u} \tau p^2s^2\\
& \overset{(c)}{\geq} -3np^3s^3\sqrt{10nps(1-s)\log n}-\frac{15}{2}np^3s^3\log n,
\end{align*}
where $(a)$ follows from $b_{v\backslash u}-b_{u\backslash v}=b_v-b_u\leq \tau$ and $(1-ps)^{x} - (1-ps)^{y} \ge (1-ps)^{\tau} -1$ when $x-y \le \tau$; $(b)$ holds due to Bernoulli's Inequality: $(1+x)^r\geq 1+rx$ for every integer $r \geq 0$ and every real number $x \geq -2$; $(c)$ follows from  $a_u< \frac{3}{2}nps$ and the definition of $\tau$ given in \prettyref{eq:def_tau}.

Combining the last two displayed equation gives that 
\begin{align}\label{eq:lowerbound-upperbound2}
&l_{\min}+m_{\min}-x_{\max}-y_{\max}-2z_{\max}-\psi_{\max}-28\log n\nonumber\\
&\geq\frac{7}{24}n^2\beta p^2s^4 -\frac{7}{4} n^2p^3s^5-21n^3p^5s^5- \sqrt{\frac{35}{16}n^2\beta p^2s^4\log n}- 5\sqrt{15n^3p^4s^4\log n}\nonumber\\
&-2n\beta\left(3 ps^2\log n+\frac{9}{4}n^2p^4s^4\right)-10n^2p^3s^3-46\log n\nonumber\\
&-3 n^2p^3s^3\sqrt{10nps(1-s)\log n}- \frac{15}{2}n^2p^3s^3\log n .
\end{align}

In view of (\ref{eq:lowerbound-upperbound2}), we can guarantee $l_{\min}+m_{\min}-x_{\max}-y_{\max}-2z_{\max}-\psi_{\max}-28\log n\geq0$ if the following inequalities (\ref{eq:boundholdstart})-(\ref{eq:boundholdend}) hold. We next verify (\ref{eq:boundholdstart})-(\ref{eq:boundholdend}) hold. 

By assumption that $\beta\geq 600\sqrt{\frac{\log n}{ns^4}}$, $np^2\leq \frac{1}{\log n}$, and $n$ is sufficiently large, we have
\begin{equation}
\begin{aligned}
\frac{1}{40}n^2\beta p^2s^4\geq&\frac{1}{40}n^2 p^2s^4\cdot600\sqrt{\frac{\log n}{ns^4}}\\
\overset{}{\geq} &15n^2 p^2s^2\cdot p\log n \\
\geq &n^2p^3s^3(\frac{15}{2}\log n+10+\frac{7}{4}s^2+21np^2s^2),\label{eq:boundholdstart}
\end{aligned}
\end{equation}
By assumption  $\beta\geq \frac{600\log n}{n^2p^2s^4}$, we have
\begin{equation}
\begin{aligned}
\frac{1}{15}n^2\beta p^2s^4\geq\frac{1}{15}n^2\sqrt{\beta} p^2s^4\cdot\sqrt{\frac{600\log n}{n^2p^2s^4}}>\sqrt{\frac{35}{16}n^2\beta p^2s^4\log n}.
\end{aligned}
\end{equation}
By assumption  $\beta\geq 600\sqrt{\frac{\log n}{ns^4}}$, we have
\begin{equation}
\begin{aligned}
\frac{1}{30}n^2\beta p^2s^4\geq\frac{1}{30}n^2 p^2s^4\cdot600\sqrt{\frac{\log n}{ns^4}} >5\sqrt{15n^3 p^4s^4\log n}.
\end{aligned}
\end{equation}
By assumption $\beta\geq \frac{600\log n}{n^2p^2s^4}$, we have
\begin{equation}
\begin{aligned}
\frac{1}{12}n^2\beta p^2s^4\geq \frac{1}{12}n^2 p^2s^4\cdot \frac{600\log n}{n^2p^2s^4} > 46\log n.
\end{aligned}
\end{equation}
By the assumption that $nps^2\geq 128\log n$, $np^2\leq\frac{1}{\log n}$, and $n$ is sufficiently large, we have
\begin{equation}
\begin{aligned}
\frac{1}{15}n^2\beta p^2s^4\geq& \frac{1}{20}n\beta ps^2\cdot 128\log n+\frac{1}{60}n^2\beta p^2s^4\cdot np^2\log n\\
\geq&2n\beta\left(3 ps^2\log n+\frac{9}{4}n^2p^4s^4\right),
\end{aligned}
\end{equation}
By the assumption  $\beta\geq 600\sqrt{\frac{np^3(1-s)\log n}{s}}$, we have
\begin{equation}
\begin{aligned}
\frac{1}{60}n^2\beta p^2s^4\geq& \frac{1}{60}n^2 p^2s^4\cdot600\sqrt{\frac{np^3(1-s)\log n}{s}}
\geq 3 n^2p^3s^3\sqrt{10nps(1-s)\log n}.\label{eq:boundholdend}
\end{aligned}
\end{equation}
Thus, we arrive at 
$
l_{\min}+m_{\min}\ge x_{\max}+y_{\max}+2z_{\max}+\psi_{\max}+28\log n.
$

\subsection{Proof of \prettyref{thm:thm2hop}}\label{sec:proof2hop}
Given any two vertices $u,v\in [n]$ with $u\neq v$, we let $W_{uv}$ denote
\begin{align*}
W_{uv}=\left\{{W_2(u,u)}>{W_2(u,v)}\right\}\cup\left\{{W_2(v,v)}>{W_2(u,v)}\right\}.
\end{align*}
We will prove $W_{uv}$ happens with high probability. We condition on $Q_{uv}$ such that the event $R_{uv}$ is true. Then, we consider two cases: $b_v-b_u\leq \tau$ and $a_u-a_v\leq \tau$. 

\paragraph{Case 1:} $b_v-b_u\leq \tau$. 

Let $\wmin\triangleq l_{\min}+m_{\min}$ and $\wmax \triangleq x_{\max}+y_{\max}+2z_{\max}+\psi_{\max}+28\log n $.  According to \prettyref{lmm:W2uu} and \prettyref{lmm:W2uv}, ${W_2(u,u)} >w_{\min}$ with high probability, and ${W_2(u,v)}< w_{\max}$ with high probability. Since $w_{\min}\geq w_{\max}$ according \prettyref{lmm:wmingwmax}, we get that $W_2(u,u)>W_2(u,v)$ with high probability. More precisely, if $R_{uv}$ occurs,
\begin{align*}
\prob{{W_2(u,u)}\leq {W_2(u,v)}\mid\Quv}
%\overset{(a)}{\leq} &\prob{\left\{{W_2(u,u)}\leq\wmin\right\}\cup\left\{{W_2(u,v)}\geq\wmax\right\}\mid\Quv}\\
\overset{(a)}{\leq}&\prob{{W_2(u,u)}\leq \wmin\mid\Quv}+\prob{{W_2(u,v)}\geq\wmax\mid\Quv}
\overset{(b)}{\leq}2\cdot n^{-\frac{7}{2}},
\end{align*}
where $(a)$ is based on the union bound; $(b)$ is based on \prettyref{lmm:W2uu} and \prettyref{lmm:W2uv}.

Since $\{{W_2(u,u)}>{ W_2(u,v)}\}\subset W_{uv}$, it follows that,
\begin{align*}
\prob{\overline{W_{uv}}\mid\Quv}
\leq\prob{{W_2(u,u)}\leq {W_2(u,v)}\mid\Quv}
\leq 2\cdot n^{-\frac{7}{2}}.
\end{align*}

\paragraph{Case 2:}  $a_u-a_v\leq \tau$.

We can lower bound $W_2(v,v)$ analogous to \prettyref{lmm:W2uu}, and prove that the lower bound is no smaller than the upper bound of $W_2(u,v)$ in this case. Then,
\begin{align*}
\prob{\overline{W_{uv}}\mid\Quv}
\leq\prob{{W_2(v,v)}\leq {W_2(u,v)}\mid\Quv}
\leq 2\cdot n^{-\frac{7}{2}}.
\end{align*}

Since $T_{uv}=\left\{a_u-a_v\leq\tau\right\}\cup\left\{ b_v-b_u\leq \tau\right\}$, applying the union bound yields that
\begin{align*}
&\prob{\overline{W_{uv}}\mid\Quv}\cdot\Indicator{R_{uv}\cap T_{uv}}\\
=&\prob{\overline{W_{uv}}\mid\Quv}\cdot\Indicator{R_{uv}}\cdot\Indicator{T_{uv}}\\
\leq&  \prob{\overline{W_{uv}}\mid\Quv}\cdot\Indicator{R_{uv}}\cdot\Indicator{b_v-b_u\leq \tau}\\
&+\prob{\overline{W_{uv}}\mid\Quv}\cdot\Indicator{R_{uv}}\cdot\Indicator{a_u-a_v\leq \tau}
\leq  4\cdot n^{-\frac{7}{2}}.
\end{align*}

Then, applying \prettyref{lmm:Ruv} and \prettyref{lmm:Tuv} yields that
\begin{align*}
\prob{\overline{W_{uv}}}
=&\mathbb{E}_{\Quv}\left[\prob{\overline{W_{uv}}\mid\Quv}\cdot\Indicator{R_{uv}\cap T_{uv}}+\prob{\overline{W_{uv}}\mid\Quv}\cdot\Indicator{\overline{R_{uv}\cap T_{uv}}}\right]\\
\leq&\mathbb{E}_{\Quv}\left[\prob{\overline{W_{uv}}\mid\Quv}\cdot\Indicator{R_{uv}\cap T_{uv}}\right]+\mathbb{E}_{\Quv}\left[\Indicator{\overline{R_{uv}\cap T_{uv}}}\right]\\
\leq & 6\cdot n^{-\frac{7}{2}}.
\end{align*}

 %If $(u, \pi(v))$ is the first selected fake pair in GreedyMaxWeightMatching algorithm, which means $(u,\pi(u))$ and $(v,\pi(v))$ have not been selected, then $W_2(u,\pi(v))$ would be greater than $W_2(u,\pi(u))$ and $W_2(v,\pi(v))$. This contradicts that $W_2(u,\pi(u))$ or $W_2(v,\pi(v))$ is larger than $W_2(u,\pi(v))$ with probability $1-n^{-2}$ according to (\ref{eq:Wuv}). Since there are $n$ iterations, the probability that there is no fake pair being selected is greater than $1-n\cdot n^{-2}=1-n^{-1}$. Thus, the algorithm can recover $\pi$ from $\pi$  with probability $1-n^{-1}$.

Finally, applying the union bound over all pairs $(u,v)$ with $u\neq v$, we get that
\begin{align*}
&\prob{\bigcap_{u,v\in[n], u\neq v}{W_{uv}}}
%=1-\prob{\bigcup_{u,v\in [n], u\neq v}\overline{W_{uv}}}
\geq1- \sum_{u \neq v} \prob{\overline{W_{uv}}}\geq 1-6\cdot n^{-\frac{3}{2}}
\geq1-n^{-1}.
\end{align*}
 
Assuming $\bigcap_{u,v\in [n], u\neq v}{W_{uv}}$ is true, we next show that the output of GMWM, $\tilde{\pi}$, must be equal to $\pi^*$.

We prove this by contradiction. Suppose in contrary that $\tilde{\pi}\neq \pi^*$. Assume the first fake pair is chosen by GMWM in the $k$-th iteration, which implies that GMWM selects true pairs in the first $k-1$ iterations. We let $\left(u^k,v^k\right)$ denote the fake pair chosen at the $k$-th iteration. Because $\bigcap_{u,v\in [n], u\neq v}{W_{uv}}$ is true,  we have  $W_2\left(u^k,u^k\right)>W_2\left(u^k,v^k\right)$ or $W_2\left(v^k,v^k\right)>W_2\left(u^k,v^k\right)$. 
%because $W_{u^kv^k}$ is true. 
We consider two cases. The first case is that $\left(u^k,u^k\right)$ or $\left(v^k,v^k\right)$ has been selected in the first $k-1$ iterations, in which case the fake pair $\left(u^k,v^k\right)$ would have been eliminated before the $k$-th iteration. The second case is that $\left(u^k,u^k\right)$ and $\left(v^k,v^k\right)$ have not been selected in the first $k-1$ iterations. Then, GMWM would select one of them instead of $\left(u^k,v^k\right)$ in the $k$-th iteration. Thus, both cases contradict to the assumption that GMWM picks a fake pair in the $k$-th iteration.
 
Hence, GMWM outputs $n$ true pairs. Then, we have $\prob{\tilde{\pi}=\pi^*}\geq \prob{\bigcap_{u,v\in [n], u\neq v}{W_{uv}}}\geq 1-n^{-1}$.

\end{appendices}

\bibliography{main}

\end{document}